\documentclass[a4paper,11pt]{article}

\usepackage[margin=1in]{geometry}

\usepackage{amssymb,amsmath,amsfonts,amsthm}
\usepackage{graphicx}%
\usepackage{multirow}%
\usepackage{xcolor}%
\usepackage{listings}%
\usepackage{hyperref}

\usepackage{anyfontsize}
\usepackage[T1]{fontenc}
\usepackage[normalem]{ulem}
\usepackage{thm-restate}
\usepackage{float}
\usepackage{adjustbox}
\usepackage{url}
\usepackage[linesnumbered]{algorithm2e}
\SetKwComment{Comment}{/* }{ */}
\RestyleAlgo{ruled}

\newtheorem{theorem}{Theorem}
\newtheorem{lemma}[theorem]{Lemma}
\newtheorem{corollary}[theorem]{Corollary}
\newtheorem{definition}[theorem]{Definition}
\newtheorem{proposition}[theorem]{Proposition}
\newtheorem{remark}[theorem]{Remark}
\newtheorem{example}[theorem]{Example}

\newcommand{\keywords}[1]{\vspace{5mm}\noindent\textbf{Keywords:}~\textit{#1}}
\newcommand{\SuA}{\ensuremath{\mathrm{sA}}}
\newcommand{\BWT}{\ensuremath{\mathrm{BWT}}}
\newcommand{\SA}{\ensuremath{\mathrm{SA}}}
\newcommand{\PA}{\ensuremath{\mathrm{PA}}}
\newcommand{\LCP}{\ensuremath{\mathrm{LCP}}}
\newcommand{\PSV}{\ensuremath{\mathrm{PSV}}}
\newcommand{\NSV}{\ensuremath{\mathrm{NSV}}}

\newcommand{\LF}{\ensuremath{\mathrm{LF}}}
\newcommand{\SSS}{\mathcal S}
\newcommand{\rev}[1]{#1^{\text{rev}}}

\newcommand{\opsearch}{\ensuremath{\mathsf{search}}}

\title{Suffixient Arrays: a New Efficient Suffix Array\\Compression Technique\footnote{This article is a journal extension of an article that appeared in the proceedings of SPIRE 2024 \cite{conferenceVersion} and of the arXiv pre-print \cite{depuydt2024suffixient}.}}

\author{Davide Cenzato$^1$, 
Lore Depuydt$^2$,
Travis Gagie$^3$,
Sung-Hwan Kim$^1$,\\
Giovanni Manzini$^4$, 
Francisco Olivares$^{5,6}$,
Nicola Prezza$^1$\\~\\
$^1$ DAIS, Ca' Foscari University of Venice, Venice, Italy\\\texttt{\{davide.cenzato,sunghwan.kim,nicola.prezza\}@unive.it},\\\\
$^2$ Department of Information Technology - IDLab, Ghent University - imec, Ghent, Belgium\\
\texttt{lore.depuydt@ugent.be}\\\\
$^3$ Faculty of Computer Science, Dalhousie University, 6050 University Ave.,\\Halifax, B3H 1W5, Nova Scotia, Canada\\
\texttt{travis.gagie@dal.ca}\\\\
$^4$ Computer Science Dept., University of Pisa, Pisa, Italy\\
\texttt{giovanni.manzini@unipi.it}\\\\
$^5$ Department of Computer Science, University of Chile, Chile\\
$^6$ CeBiB --- Centre for Biotechnology and Bioengineering, Chile\\
\texttt{folivares@uchile.cl}
}

\date{}

\begin{document}

\maketitle


\abstract{
The Suffix Array is a classic text index enabling on-line pattern matching queries via simple binary search.
The main drawback of the Suffix Array is that it takes linear space in the text's length, even if the text itself is extremely compressible. 
Several works in the literature showed that the Suffix Array can be compressed, but they all rely on complex succinct data structures which in practice tend to exhibit poor cache locality and thus significantly slow down queries.
In this paper, we propose a new simple and very efficient solution to this problem by presenting the \emph{Suffixient Array}: a tiny subset of the Suffix Array 
\emph{sufficient} to locate on-line one pattern occurrence (in general, all its Maximal Exact Matches) via binary search, provided that random access to the text is available.
We prove that: (i) the Suffixient Array length $\chi$ is a strong repetitiveness measure, (ii) unlike most existing repetition-aware indexes such as the $r$-index, our new index is efficient in the I/O model, and (iii) Suffixient Arrays
can be computed in linear time and compressed working space. 
We show experimentally that, 
when using well-established compressed random access data structures on repetitive collections, the Suffixient Array $\SuA$ is 
\emph{simultaneously} (i) faster and orders of magnitude smaller than the Suffix Array $\SA$ and (ii) smaller and \emph{one to two orders of magnitude faster} than the $r$-index.  
With an average pattern matching query time as low as 3.5 ns per character, our new index gets very close to the ultimate lower bound: the RAM throughput of our workstation (1.18 ns per character).
}

\keywords{Suffixient set, Maximal Exact Match, Right-maximal substring, Text indexing, Compressed-space algorithms}

\maketitle

\section{Introduction and Overview of the Paper}\label{sec:introduction}

In this paper, we propose a new simple and remarkably efficient solution to the well-studied problem of indexing compressed text for fast pattern matching queries \cite{Nav22a,Nav22b}. 
This problem is motivated by the constantly-increasing amount of repetitive data produced in several applicative scenarios. Bioinformatics is an excellent example of such a phenomenon: with the quick advancement of DNA sequencing technologies in the last decade, an enormous amount of genomic data can now be produced efficiently in terms of both time and cost (for instance, since the introduction of Next-Generation Sequencing, the cost of sequencing a Human genome dropped from 1 million to just 600 dollars in roughly a decade \cite{DNASequencingCost}). 
These technological improvements opened the \emph{pangenome era}, in which typical analyses often involve the comparison of billions of sequence fragments with  large collections of reference genomes \cite{MediniDTMR05,Wang22}.
This \emph{data} explosion, however, did not correspond to an \emph{information} explosion. Since genomes of the same species are nearly-identical (any two Human genomes, for instance, share 99.9\% of their content \cite{AltshulerPC00}), a collection of same-species genomes is usually extremely compressible. 
The goal of a compressed index is to pre-process such a collection into a small (compressed) data structure in order to speed up subsequent pattern matching queries. Importantly, such queries must be supported directly on the compressed data without decompressing it.

\paragraph{Compressed indexed pattern matching: state of the art}

Suffix arrays \cite{manber1990suffix,Baeza-Yates:1989:NAT:75334.75352,gonnet1992new} are a classical (uncompressed) solution to the indexed pattern matching problem. They use linear space in the text's length and are extremely simple: the Suffix Array $\SA$ is the lexicographically-sorted list of (the starting positions of) the text's suffixes. Pattern matching on $\SA$ can be performed via simple binary search combined with text extraction, needed to compare the sought pattern with the text's suffixes during binary search. 

As mentioned above, however,  linear space is not acceptable in modern big-data scenarios. Subsequent research therefore focused on techniques for compressing the Suffix Array and for supporting pattern matching on conceptually-different compressed text representations, such as Lempel-Ziv'77~(LZ77)~\cite{ziv1977universal} and grammar compression~\cite{LarssonM00}. 
First research papers on compressed suffix arrays focused on entropy compression and exact pattern matching \cite{FM00,GV00}. 
Entropy, however, is a weak compressibility measure in the repetitive regime as it is based on symbol frequencies. As a consequence, the empirical entropy of a text $T$ and a collection composed of, say, thousands of copies of $T$ is the same (i.e. entropy is not sensitive to repetitions). 
Since then, a variety of indexes have been developed to capture and exploit the repetitiveness of the input data in different ways; see \cite{Nav22a,Nav22b} for an extensive survey on compressed indexes and repetitiveness measures. 
Notable works in this research direction include the run-length FM-index of M\"akinen~and~Navarro~\cite{MN05}, which used space proportional to the number of runs $r$ in the Burrows-Wheeler transform (BWT) of the text (where $r$ is known to effectively capture repetitiveness in the text, see \cite{Nav22a}) to count the number of occurrences of a pattern. Another key contribution is the $r$-index of Gagie~et~al.~\cite{GNP20}, which extended these functionalities to locate queries (that is, reporting the occurrences of the pattern in the indexed text) in the same asymptotic space. 
These indexes can be interpreted as techniques to compress the Suffix Array to space proportional to $O(r)$ words. Compared with the Suffix Array, however, they suffer from a big issue: being based on 
backward searching the Burrows-Wheeler transform and on
complex compressed data structures for \emph{rank}, \emph{select}, and \emph{predecessor} queries, they can no longer be considered as simple data structures and are notoriously not cache efficient. As a matter of fact, searching a pattern of length $m$ in those indexes has $\Omega(m)$ I/O complexity (i.e. cache misses). In practice (as we also investigate in this article), this translates to the fact that these indexes are orders of magnitude slower than suffix arrays. 
This has recently been mitigated by techniques based on relative Lempel-Ziv (RLZ) compression of the Suffix Array \cite{puglisi2021smaller}, which however incur in a larger space usage.

\paragraph{Extensions to approximate pattern matching}
While these works focused on \emph{exact} pattern matching, in real applications one is usually interested in finding \emph{approximate} occurrences of the pattern under some meaningful string distance (e.g. Hamming or edit distance).
For instance, for approximately matching short reads under the edit distance, the tool b-move~\cite{depuydt_et_al:LIPIcs.WABI.2024.10} offers a lossless, repetition-aware (run-length-compressed BWT) solution.
Another widely used approach to approximate pattern matching, effective for various types of reads and commonly employed in practice~\cite{Li13}, is to locate \emph{Maximal Exact Matches} (MEMs), and then extend them to full (approximate) occurrences of the pattern. Given a pattern $P$,  a MEM is a maximal (that is, not extendible either to the left nor to the right) substring $P[i,j]$ occurring without errors in the collection. 
Bannai~et~al.~\cite{BannaiGI20} showed that one can augment the $r$-index to compute MEMs efficiently within the same $O(r)$ space bound on top of a random-access oracle on the text. Their algorithm has two phases. In the first phase, it scans the pattern from right to left to compute for every suffix of the pattern the position in the text of an occurrence of the longest prefix of that suffix that occurs anywhere in the text. Then it scans the pattern again, but this time from left to right, to compute the length of a longest match of each of its suffixes using a random access oracle and the positions stored in the first phase, from which MEMs can be reported immediately.
Subsequently, an efficient construction algorithm for the data structure of Bannai~et~al.~\cite{BannaiGI20}   was presented by Rossi~et~al.~\cite{RossiOLGB22}, who also presented a practical implementation.
The two-phases algorithm of Bannai~et~al.~\cite{BannaiGI20} was later simplified into a one-pass (left-to-right scan) algorithm by Boucher~et~al.~\cite{BoucherGIKLMNPR21} which allows for processing patterns in the streaming model.

\subsection{Our contributions}

In this paper we present a new Suffix Array compression technique that is \emph{repetition-aware}, \emph{simple}, and remarkably \emph{efficient in practice}. More in detail, we introduce the \emph{Suffixient Array} ($\SuA$): a sampling scheme of the Prefix Array\footnote{Of course, one can define the same concepts symmetrically as a sampling of the Suffix Array; we decided to sample the Prefix (rather than Suffix) Array in order to support \emph{on-line} pattern matching (read next).} (the co-lexicographically-sorted list of text prefixes) with the following  properties. 

\begin{enumerate}
    \item If random access is available on the text, then a sequence of simple binary searches on $\SuA$ suffices to locate \emph{on-line} one occurrence of the pattern $P$ in the text (i.e. we return an occurrence of every prefix $P[1,i]$ as soon as $P[i]$ arrives) or, more in general, all its MEMs. Using state-of-the-art compressed random access techniques, we obtain a fully-compressed index.
    \item $\SuA$ is compressed  in the sense that its length $\chi$ is upper-bounded by $O(\bar r)$, the number of equal-letter runs in the Burrows-Wheeler transform of the reversed text. Furthermore, $\chi$ does not depend on the alphabet order (while $\bar r$ does) and  
    there exist infinite text families such that $\chi \in o(\bar r)$.
    This is confirmed in practice: on a repetitive collection containing variants of Human chromosome 19, our index is smaller than the \emph{toehold lemma} of the $r$-index (i.e. its subset for locating one pattern occurrence --- about half the size of the full $r$-index).
   \item Unlike most existing compressed indexes, pattern matching on our index  is also efficient in the I/O model. This is confirmed in practice: our index is  \emph{one to two orders of magnitude faster} than the $r$-index \cite{GNP20}.
    We compare our query times with a hard lower bound --- the time needed on our workstation to extract $|P|$ contiguous characters from a random position in a large text (i.e. the RAM throughput) --- and show that our index approaches this bound. 
\end{enumerate}

While our index can  locate only one pattern occurrence, standard techniques \cite{GNP20,nishimoto2021optimal} can be used to list efficiently all other occurrences by enumerating Suffix Array values adjacent to the pattern occurrence. We do not enter in such details in this article and will describe this line of research in a future publication under preparation.

When the uncompressed text is used for random access, our technique can be directly compared with binary searching the Suffix (Prefix) Array. 
In order to obtain a self-index, however, we need to use a compressed text representation.
A vast amount of literature has been devoted to supporting fast random access in compressed space, for example via Straight-Line programs \cite{BilleLRSSW15}, LZ77 compression \cite{KNtcs12}, block trees \cite{BCGGKNOPTjcss20}, or relative Lempel-Ziv (RLZ) \cite{kuruppu2010relative}. See \cite{Nav22a} for a complete survey on the topic. Some of these data structures are also known to be very small and fast in practice \cite{BCGGKNOPTjcss20,kuruppu2010relative}. As a consequence, they can be used as a black-box in point (2) above, arguably yielding the first \emph{simple} (and efficient, as we show experimentally) compressed self-index for repetitive sequences. 

\paragraph{Suffixient sets}
To achieve this result, we start by introducing a new class of combinatorial objects of independent interest:  \emph{suffixient sets}.
Let $T\in \Sigma^n$ be a text of length $n$.
Consider any set $S \subseteq [n]$
such that the following holds: 
for every one-character right-extension $T[i,j]$ of every right-maximal substring $T[i,j-1]$ of $T$, there exists $x\in S$ such that $T[i,j]$ is a suffix of $T[1,x]$. A substring $\alpha = T[i,j-1]$ is said to be \emph{right-maximal} if both $\alpha\cdot a$ and $\alpha\cdot b$ appear in $T$ for some characters $a \neq b$, or if $\alpha$ is a suffix of $T$.
We call a set $S$ with this property \emph{suffixient} (being \emph{sufficient} to ``capture'' all right-extensions of right-maximal strings or, equivalently, ``cover''  paths starting at the root of the suffix tree of $T$ and ending in the first character labeling every edge) and denote with $\chi$ the size of a smallest (not necessarily unique) suffixient set for $T$. 

We proceed by exhibiting a suffixient set of cardinality $2\bar r$, where $\bar r$ is the number of runs in the Burrows-Wheeler transform of $T$ reversed. In particular, this implies $\chi \le 2\bar r$, showing that $\chi$ is a new meaningful repetitiveness measure.
Observe that $\chi$ is independent of the alphabet ordering, while $\bar r$ is not. Since re-ordering the alphabet may reduce asymptotically $\bar r$ \cite{BentleyGT20}, from the bound $\chi \leq 2\bar r$ we immediately obtain that there exist string families on which $\chi = o(\bar r)$. In other words, $\chi$ is a better repetitiveness measure than $\bar r$, even though we leave open the problem of determining whether $\chi$ is \emph{reachable} (i.e. if $O(\chi)$ words are always sufficient to store the text). We get however close to this goal: we prove that any suffixient set is a string attractor \cite{KP2018}, which implies that $O(\chi\log(n/\chi))$ words of space are sufficient to store $T$.

\paragraph{Suffixient Arrays}

We  define a \emph{Suffixient Array} $\SuA$ of $T$ to be any (not necessarily unique) suffixient set $S$ of smallest cardinality $\chi$, sorted according to the co-lexicographic order of the text prefixes represented by $S$.

A Suffixient Array $\SuA$  can be used to locate one occurrence of a pattern $P[1,m]$ in $T$ with the following simple strategy.  
Assume that $P$ occurs in $T$ and that the prefix $P[1,j]$ (starting with $j=0$) is right-maximal in $T$. (i) Find (by binary search on $\SuA$ and random access on $T$) a prefix $T[1,x]$, for $x\in \SuA$, suffixed by $P[1,j+1]$. (ii) Right-extend the match by comparing $T[x+1,\dots]$ and $P[j+2, \dots]$. This way, we either match until the end of $P$ or we find a mismatch $T[x+1+t] \neq P[j+2+t]$, for some $t\geq 0$. 
In the former case, we are done. 
In the latter case, we know that $P[1,j+1+t]$ is right-maximal in $T$ so we can repeat steps (i) and (ii). Starting with the empty prefix of $P$,
this procedure yields one occurrence of $P$ in $T$ and, as we show in this paper, can be easily extended to return all MEMs of $P$ in $T$.
While this procedure extracts from the text a number of characters proportional to $|P|^2$ in the worst case, we show that it actually triggers at most $O((1+m/B)\cdot d\log\chi)$ I/O operations (i.e. cache misses), where $B$ is the I/O block size (i.e. cache line) and $d \le m$ is the node depth of $P$ in the suffix tree of $T$. In many practical applications (for example, DNA alignment), $d \ll m$ holds (on uniform random texts, $d \in O(\log_\sigma n)$ with high probability). 

We furthermore show that running time can be made linear in $|P|$ also in the worst case by employing z-fast tries \cite{BelazzouguiBPV09} and fast longest common prefix/suffix queries on $T$ in compressed space.

\paragraph{Computing the Suffixient Array}

We proceed with the following natural question: can the Suffixient Array $\SuA$ be computed efficiently?
To answer this question, we start by characterizing the smallest suffixient set. 
Intuitively, we show that a suffixient set is of smallest cardinality when it captures precisely all the $\chi$ one-character extensions of all right-maximal substrings of $T$ being also left-maximal. We call such substrings of $T$ \emph{supermaximal extensions}.
We then describe four algorithms solving the problem, all based on the Burrows-Wheeler transform ($\BWT$), the Longest Common Prefix array ($\LCP$), and the Suffix Array ($\SA$) of $T$ reversed.   
The first algorithm runs in quadratic time and is conceptually very simple. The other three algorithms are optimizations of the first one. 
The second algorithm runs in
$O(n + \bar r\log\sigma)$ time and requires only one pass on those three arrays. The third and fourth algorithms require random access on those arrays and run in optimal $O(n)$ time. While being slower than the latter two when $\bar r \in \omega(n/\log\sigma)$, the one-pass property of the second algorithm makes it ideal to be combined with Prefix Free Parsing (PFP) \cite{BoucherGKLMM19}, a technique able to stream those three arrays in compressed working space. We show experimentally that this combination of PFP with our one-pass algorithm can quickly compute a smallest suffixient set in compressed working space on massive repetitive text collections.

We conclude the paper with experimental results testing Suffixient Array on pattern matching queries.

\section{Preliminaries}\label{sec:preliminaries}

We use the following notation: $[i,j]$ for an interval indicating all values $\{i,i+1,\dots,j\}$ (if $j<i$, then $[i,j] = \emptyset$), 
$[n]$ for $[1,n]$, 
$A[1,n] \in U^n$ for an array $A$ of length $n$ over universe $U$ whose indices range from $1$ to $n$, $A[i]$ for the $i$-th element of $A$, $A[i,j]$ for the {\em sub-array} $A[i]\cdots A[j]$ of $A$, where $1 \leq i ,j \leq n$  (if $j<i$, then $A[i,j] = \epsilon$ is the empty array), and $|A|$ for the length of $A$. 
More in general, $A[a,b]$, with $a\leq b \in \mathbb N$, indicates that $A$ is an array whose indices range from $a$ to $b$: the first element of $A$ is $A[a]$, the second $A[a+1]$, until the last which is $A[b]$.

Let $\Sigma$ be a finite ordered alphabet of size $\sigma$. 
A \emph{string} of length $n$ over alphabet $\Sigma$ is denoted with $\alpha[1,n] \in \Sigma^n$.

Given two strings $\alpha, \beta\in \Sigma^*$, $LCS(\alpha,\beta)$ denotes the length of the longest common suffix between $\alpha$ and $\beta$. Symmetrically, $LCP(\alpha,\beta)$ denotes the length of the longest common prefix between $\alpha$ and $\beta$.

A {\em run} for a string $\alpha$ is a maximal substring $\alpha[i,j]$ containing only one distinct character. We denote by $\text{runs}(\alpha)$ the number of runs of $\alpha$. 
For instance, $runs(AAACCTAAGGTAA) = 7$.

A \emph{text} $T$ is a string ending with a special character $T[n] = \$ \in \Sigma$ appearing only in $T[n]$ and such that $\$<a$ for every $a\in\Sigma\setminus \{\$\}$. 
This special character is needed (as customary in the literature) for correctly defining structures such as the suffix tree and the Burrows-Wheeler transform (see below for their definition). Since we will use these structures also on the reversed text, to ensure that $\$$ always appears at the end of the processed text, we define the reverse $\rev{T}$ of a text $T$ as $\rev{T}=T[n-1]T[n-2]\cdots T[1]\$$.
To avoid confusion and improve readability, we will use symbols $T$ and $P$ to denote the input text and sought pattern (as well as their substrings using notation $T[i,j]$ and $P[i,j]$), and symbols $\alpha, \beta, \dots$ to denote general strings. 
When $T$ is the input text of our algorithms, by the above assumptions it holds that $n = |T|\geq 2$ and $T$ contains at least two distinct characters. Moreover, we assume $\sigma\leq n$. This is not a strong requirement since the alphabet can always be re-mapped to $[n]$ without changing the solutions to the problems considered in this paper. 
Notice that re-mapping the alphabet does however change the space complexity of our algorithms since one needs to store the map, using additional space proportional to $O(\sigma')$ words (for example, using a perfect hash function), where $\sigma' \leq n$ is the number of distinct characters appearing in the text. Since, however, we will show that $\sigma' \le \chi$, where $\chi$ is the size of the Suffixient Array, such map will not change asymptotically our space usage.

The following notions will be frequently used throughout the paper.
\medskip
\begin{definition}[Right-maximal substring]
    For a string $\alpha[1,n]$, a substring $\alpha[i,j]$ ($j\geq i-1$) of $\alpha$ is a \emph{right-maximal substring} (or \emph{repeat}) if either (i) $j=n$ or (ii) there exist at least two distinct characters $a,b\in \Sigma$ such that both $\alpha[i,j]\cdot a$ and $\alpha[i,j]\cdot b$ are substrings of $\alpha$. 
\end{definition}

\medskip

Observe that, by the above definition, the empty string $\epsilon$ is right-maximal.

\medskip

\begin{definition}[Maximal exact match (MEM)]
    For strings $\alpha[1,n]$ and $\beta[1,m]$, a triplet $(i,j,\ell)$ with $i\in[m]$ and $j\in[n]$ is a \emph{Maximal Exact Match (MEM)} of $\beta$ with respect to $\alpha$ if and only if the following three conditions hold: (i) $\beta[i-\ell+1,i]=\alpha[j-\ell+1,j]$, (ii) $i-\ell+1=1$ or $\beta[i-\ell,i]$ is not a substring of $\alpha$, (iii) $i=m$ or $\beta[i-\ell+1,i+1]$ is not a substring of $\alpha$.
\end{definition}

\medskip

The above definition of a MEM is also sometimes referred to as a \emph{super-maximal exact match (SMEM)}~\cite{10.1093/bioinformatics/bts280} (not to be confused with the notion of \emph{supermaximal extension} introduced later in Definition \ref{def:supermaximal extension}). Observe that we take $i$ and $j$ to be the \emph{ending}  positions of the matches. This is clearly equivalent to taking the \emph{beginning} positions of the matches  (as is common in the literature), and simplifies the presentation of our algorithms.

\medskip
Given two strings $\alpha$ and $\beta$, we define the {\em lexicographic order} $<_{\text{lex}}$ on $\Sigma^*$ as follows: $\beta<_{\text{lex}} \alpha$ if $\beta$ is a proper prefix of $\alpha$, or if there exists an index $j$ s.t.\ $\beta[j]<\alpha[j]$ and for all $i$ s.t. $1\le i<j$, $\beta[i]=\alpha[i]$.
The co-lexicographic order $<_{\text{colex}}$ is defined symmetrically: $\beta<_{\text{colex}} \alpha$ if $\beta$ is a proper suffix of $\alpha$, or if there exists an index $j$ s.t.\ $\beta[|\beta|-j+1]<\alpha[|\alpha|-j+1]$ and for all $i$ s.t. $1\le i<j$, $\beta[|\alpha|-i+1]=\alpha[|\beta|-i+1]$.

\medskip

\begin{definition}[Suffix array ($\SA$)\cite{MN93}]
    Given a 
    text $T[1, n]$, the \emph{Suffix Array} $\SA(T)$ of $T$ ($\SA$ for brevity) is the permutation of $[n]$ such that $T[\SA[i], n] <_{\text{lex}} T[\SA[j], n]$ holds for any $1\le i<j \le n$. 
\end{definition}
\medskip

Symmetrically, the prefix array sorts text prefixes co-lexicographically:

\medskip
\begin{definition}[Prefix array ($\PA$)]\label{def:PA}
    Given a 
    text $T[1, n]$, the \emph{Prefix Array} $\PA(T)$ of $T$ ($\PA$ for brevity) is the permutation of $[n]$ such that $T[1,\PA[i]] <_{\text{colex}} T[1,\PA[j]]$ holds for any $1\le i<j \le n$. 
\end{definition}
\medskip

\begin{definition}[LCP array \cite{MN93}]
    Given a text $T[1, n]$, the  \emph{longest common prefix} array  $\LCP(T)$ of $T$ ($\LCP[2,n]$ for brevity) is the length-$(n-1)$ integer array such that $\LCP[i]$ stores the length of the longest common prefix between $T[\SA[i],n]$ and $T[\SA[i-1],n]$, for all $2 \le i \le n$.
\end{definition}
\medskip

Since we will frequently refer to multiple elements of the LCP array, we use the notation $\LCP[B]$ for a set $B \subseteq [n]$ to denote the set $\{\LCP[i]\ :\ i\in B\}$.

\medskip
\begin{definition}[Suffix tree (ST), \cite{Weiner73}]
    The suffix tree of a text $T[1,n]$ is an edge-labeled rooted tree with $n$ leaves numbered from $1$ to $n$ such that (i) each edge is labeled with a non-empty string, (ii) each internal node has at least two outgoing edges, (iii) the labels of outgoing edges from the same node start with different characters, and (iv) the string obtained by concatenating the edge labels on the path from the root to the leaf node numbered $\SA[i]$ is the $i$-th smallest suffix $T[\SA[i],n]$ where $SA$ is the Suffix Array of $T$.
\end{definition}
\medskip
For a node $v$ in a suffix tree, the string obtained by concatenating the edge labels on the path from the root to node $v$ is called \emph{path label} of $v$. 
By definition, the path label $\alpha$ of every suffix tree node is a right-maximal substring of $T$. Vice-versa, for every right-maximal substring $\alpha$ of $T$ there exists exactly one suffix tree node having $\alpha$ as path label.  
For two nodes $u$ and $v$ such that their path labels are $\alpha$ and $a\alpha$, respectively, for some character $a\in \Sigma$ and string $\alpha\in\Sigma^*$, we say $u$ can be reached by the \emph{suffix link} from $v$.
\medskip
\begin{definition}[Burrows-Wheeler transform (BWT), \cite{BW94}]
    Given a text $T[1, n]$, the Burrows-Wheeler Transform of $T$, $\BWT(T)$, is a permutation of the characters of $T$ defined by concatenating the characters preceding the lexicographically-sorted suffixes of $T$, i.e. $\BWT(T)[i] = T[\SA[i]-1]$ where we define $T[0]=T[n]$ for brevity.
\end{definition}
\medskip

We indicate with $\bar r$ the number of runs in $\BWT(\rev{T})$. For example, if $T={\tt BANANA\$}$, then $\rev{T}={\tt ANANAB\$}$ (as stated above, we define the reverse of a text so that $\$$ still appears at the end) and $\text{BWT}(\rev{T}) = {\tt BNN\$AAA}$, so ${\bar r} = 4$.

We say that a randomized algorithm returns the correct result \emph{with high probability} (abbreviated \emph{w.h.p.}) if the probability of failure is at most $n^{-c}$ for any constant $c$ fixed before the algorithm begins. Here, $n$ is the size of the input (in our cases, $n$ will be the length of the input text).

\section{Suffixient Sets}
\label{sec:suffixient}

We start by defining the combinatorial object at the core of this paper: \emph{suffixient sets}. 
We give two equivalent definitions based on suffix trees and on right-maximal strings. 

We say that a set $S\subseteq [n]$ of positions in $T$ is \emph{suffixient} if the suffixes of the prefixes of $T$ ending at those positions are \emph{sufficient} to ``cover'' the \emph{suffix} tree of $T$ in the way described in the following definition: 

\medskip
\begin{definition}[Suffixient set - definition based on suffix trees]
\label{def:suffixient_set}
A {\em suffixient set} for a text $T [1, n]$ is a set $S$ of numbers between 1 and $n$ such that, for any edge descending from a node $u$ to a node $v$ in the suffix tree of $T$, there is an element $x \in S$ such that $u$'s path label is a suffix of $T [1, x - 1]$ and $T [x]$ is the first character of $(u, v)$'s edge label.
\end{definition}
\medskip

Because of the correspondence between suffix tree nodes and right-maximal substrings, an equivalent way to define suffixient sets is in terms of right-maximal substrings of $T$:

\medskip
\begin{definition}[Suffixient set - definition based on right-maximal strings]\label{def:suffixient}
    A set $S \subseteq [n]$ is  \emph{suffixient} for a text $T[1,n]$ if, for every one-character right-extension $T[i,j]$ ($j\geq i$) of every right-maximal string $T[i,j-1]$, there exists $x\in S$ such that  $T[i,j]$ is a suffix of $T[1,x]$.
\end{definition}
\medskip

The following lemma shows there is always a suffixient set for $T$ of cardinality at most $2 \bar{r}$.  We note that suffixient sets can be even smaller, however, see Figure~\ref{fig:cover}: for the text $T$ in the example it holds $\bar{r} = 9$ but $\{14, 20, 33, 35\}$ is still suffixient.

\begin{figure}[t]
\begin{center}
\includegraphics[width=.88\textwidth]{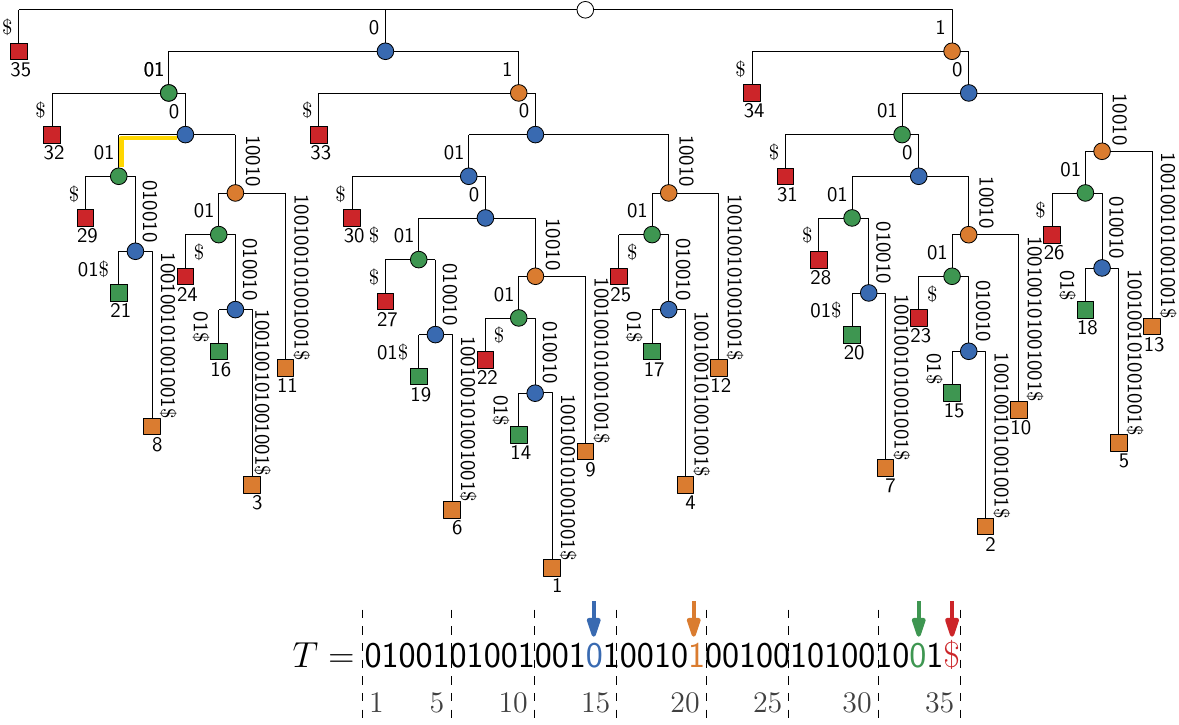}
\caption{
A text $T$ with highlighted the positions of the suffixient set $S = \{14, 20, 33, 35\}$. The figure also shows the suffix tree for $T$: each leaf is identified by the starting position in $T$ of the associated suffix; for example, the third leaf from the left has index 29 since it corresponds to the suffix $T[29,35] = \mathsf{001001\$}$. Node colors show how the set $S$ covers the suffix tree: for each edge $(u,v)$, the destination node $v$ has the same color as the position in $T$ covering it. For example, the parent of leaf 29 is colored in green since the edge $(u,v)$ leading to that node (highlighted in yellow) is covered by $T[33]$. Indeed, $u$'s path label is $\mathsf{0010}$ which is a suffix of $T[1,32]$ and $T[33]=\mathsf{0}$ is the first character of $(u,v)$.}
\label{fig:cover}
\end{center}
\end{figure}

\medskip
\begin{lemma}
\label{lem:upper_bound}
Let $T[1,n]$ be a text.
Let $\SA$ and $\LCP$ denote the Suffix Array and LCP array of $\rev T$.
The set $\{n-SA[i]+1  \ :\ i = 1\ \vee\ i=n\ \vee\ \BWT[i]\neq \BWT[i-1]\ \vee\ \BWT[i]\neq \BWT[i+1],\ 1 \le i \le n\}$ 
of positions on $T$ corresponding to characters at the at most $2 \bar{r}$ run boundaries in the BWT of $\rev T$, is suffixient for $T$.
\end{lemma}

\begin{proof}
Consider an edge $(u, v)$ descending from a node $u$ to a node $v$ in the suffix tree of $T$.  Let $\alpha$ be $u$'s path label and $c$ be the first character in $(u, v)$'s edge label.  By the definition of the BWT, the occurrences of characters immediately following occurrences of $\alpha$ in $T$ are consecutive in the BWT of the reverse of $T$.  By the definition of suffix trees, $v$ must have a sibling, so not all the characters immediately following occurrences of $\alpha$ in $T$ are copies of $c$.  Therefore, for some $x$ such that $T [x]$ is at a run boundary in the BWT of the reverse of $T$, we have $T [x - |\alpha|, x] = \alpha\,c$. Hence, $\alpha$ is a suffix of $T[1,x-1]$ and $T[x]=c$ and position $x$ ``covers'' the suffix tree edge $(u,v)$.
\end{proof}
\medskip

The next lemma shows that any suffixient set for $T$ is also a string attractor~\cite{KP2018} for $T$ --- that is, for any non-empty substring $T [i, j]$ of $T$, some occurrence $T[i',j']$ of $T [i, j]$ is such that $i'\le s \le j'$ for some $s \in S$. 

\medskip
\begin{lemma}
\label{lem:attractor}
Any suffixient set $S$ for $T$ is a string attractor for $T$.
\end{lemma}
\begin{proof}
Consider a non-empty substring $T [i, j]$ of $T$, let $v$ be the locus of $T [i, j]$ in the suffix tree of $T$ (that is, the highest node whose path label is prefixed by $T [i, j]$), let $u$ be $v$'s parent, let $\alpha$ be $u$'s path label, and let $c$ be the first character of $(u, v)$'s label.  Then $\alpha\,c$ is a prefix of $T [i, j]$ and, since $(u, v)$ is a single edge, any occurrence of $\alpha\,c$ in $T$ is contained in an occurrence of $T [i, j]$.  By Definition~\ref{def:suffixient}, there is some $x \in S$ such that $\alpha\,c$ is a suffix of $T [1, x]$, so $T [x]$ is contained in an occurrence of $\alpha\,c$ and thus contained in an occurrence of $T [i, j]$.
\end{proof}

As a consequence of the above results, if we denote the cardinality of the smallest suffixient set for $T$ by $\chi$, we have $\sigma' \le \gamma \le \chi \leq 2 \bar{r}$, where 
$\sigma'$ is the number of distinct characters appearing in $T$ and 
$\gamma$ is the size of the smallest string attractor for $T$. These relations show that $\chi$ is a new meaningful repetitiveness measure \cite{Nav22a}. 
In addition, combining Lemma~\ref{lem:attractor} with Kempa and Prezza's~\cite{KP2018} results on string attractors, we immediately obtain the following:

\medskip
\begin{corollary}\label{cor:extract}
Given a text $T$ with smallest suffixient set of size $\chi$,
we can build a data structure of size $O (\chi \log (n / \chi))$ words of $O(\log n)$ bits supporting the extraction of any length-$\ell$ substring of $T$ in $O(\log(n/\chi)+\ell)$ time.
\end{corollary}
\medskip

In particular, $O(\chi \log (n / \chi))$ words are sufficient to store the text. We leave open the problem of determining if $O(\chi)$ words are still sufficient. 
To conclude this section, note that $\chi$ is independent of the alphabet ordering, while $\bar r$ is not. Since re-ordering the alphabet may reduce asymptotically $\bar r$ \cite{BentleyGT20}, from the bound $\chi \leq 2\bar r$ we immediately obtain that there exist string families on which $\chi = o(\bar r)$. In other words, on such strings $\chi$ is (asymptotically) a smaller repetitiveness measure than $\bar r$. 
Importantly note that, while it is true that the number of runs  $\bar r$ of the $\BWT$ of $\rev T$  can be minimized by choosing the best alphabet order, (the decisional version of) the problem of finding such an order is NP-complete \cite{bentley2020complexity}. On the other hand, in this paper we will show that a smallest suffixient set, of size $\chi$, can be computed in linear time.

\section{Suffixient Arrays: String Matching with Suffixient Sets}
\label{sec:compressed_index}

Based on the notion of suffixient set, we introduce \emph{Suffixient Arrays}:

\medskip
\begin{definition}
    A \emph{Suffixient Array} $\SuA$ for a text $T$ is an integer array of size $\chi$ whose entries form a suffixient set $S$ of smallest cardinality for $T$, sorted by the co-lexicographic order of the text prefixes $T[1,x]$ corresponding to each $x\in S$.
\end{definition}
\medskip

Later, in Section \ref{sec:small_suff} it will be made clear that a text $T$ could admit multiple Suffixient Arrays. For the purposes of this section, any Suffixient Array will work.

Observe that the Suffixient Array is a sampling of the Prefix Array (Definition \ref{def:PA}). 
In this article, we decided to sample the Prefix (rather than the Suffix) Array since, as we now show, this choice allows us to perform \emph{on-line} pattern matching  via a simple binary search strategy on $\SuA$. This feature is useful in streaming scenarios when long patterns are provided one character at a time and one wishes to know an occurrence of each pattern prefix $P[1,i]$ immediately after $P[i]$ has arrived. 

We first show in Subsection \ref{subsec: alg one occ} how to use $\SuA$ to find one occurrence of each pattern prefix $P[1,i]$ given in an on-line fashion, provided that random access on $T$ is available. 
We analyze the running time of our algorithm in the worst-case scenario in the RAM and I/O models.
Then, in Subsection \ref{subsec: mem} we show how to extend this idea to find all Maximal Exact Matches (MEMs) of $P$ in $T$. To conclude, in Section \ref{subsec:z-fast tries} we speed up our queries by employing more advanced data structures (\emph{z-fast tries} and \emph{Straight-Line Programs}). Our strategies will actually work on \emph{any} sorted suffixient set for $T$; however, since in Section \ref{sec:smallest-SuffixientSet-comp} we show that a Suffixient Array $\SuA$ can be computed in linear time and space (and even in compressed space), in the following we  describe our algorithms using $\SuA$.

\subsection{Locating a single pattern occurrence} \label{subsec: alg one occ}

We present our solution in a general form in order to make it easier later (in Subsection \ref{subsec:z-fast tries}) to plug in faster data structures to replace binary search. 
Let $\SuA$ be a Suffixient Array for $T$. 
Assume that we have data structures on the pair $(T,\SuA)$ supporting the following operations: 

\medskip
\begin{definition}[{random access: $T[i]$}] \label{def:random access query}
Given an integer $i\in[n]$, return the $i$-th character of $T$.
\end{definition}

\medskip
\begin{definition}[longest common suffix: $\opsearch(\alpha)$]\label{def:LCS query}
    given a string $\alpha \in\Sigma^m$, return a pair $(x,\ell) \in \SuA \times \mathbb N$ such that 
    (i) $\ell$ is the length of the longest common suffix between $\alpha$ and $T[1,x]$, and (ii) there does not exist any other such pair $(x',\ell') \in \SuA \times \mathbb N$ with  $LCS(\alpha, T[1,x']) = \ell' > \ell$. 
    If $\alpha[m]$ does not suffix $T[1,x]$ for any $x\in \SuA$, then return the pair $(0,0)$. 
\end{definition}
\medskip

Observe that Operation $\opsearch(\alpha)$ can be supported in time $O(m\log \chi)$ by simple binary search on $\SuA$ and random access on $T$, provided that random access on $T$ can be performed in constant time per extracted character (otherwise, this time gets multiplied by time needed to access a single text character). 
This time can be sped up to $O(m + \log \chi)$ by using classic longest-common-suffix (LCS) data structures on the Suffixient Array (in \cite{manber1990suffix}, the authors show how to achieve this on the Suffix Array using the symmetric LCP array), using $O(\chi)$ additional words of space.

Algorithm \ref{alg:one-occ} shows how to locate an occurrence of every $P[1,i]$, where $P$ is a query pattern, using the above two operations. 

\medskip
\begin{algorithm}[th!]
\caption{Finding one occurrence of every prefix of $P [1, m]$ in $T [1, n]$.}
\label{alg:one-occ}
\LinesNumbered
\SetKwInOut{Output}{output}
\SetKwInOut{Input}{input}
   \Input{A pattern $P[1,m]$.}
    \Output{For every $1 \le i \le m$, a pair $(i,j)$ such that $T[j-i+1,j]=P[1,i]$ if $P[1,i]$ occurs in $T$. 
    Otherwise, terminates with \texttt{NOT\_FOUND} at the smallest $i$ such that $P[1,i]$ does not occur in $T$.}
    {$i \gets  j \gets 0$\;}
    \While{$i < m$}{
      $i\gets i+1$; $j\gets j+1$\;\label{line: incr i j}
      \If{$P[i] \neq T[j]$}{
      $(j,\ell) \gets \opsearch(P[1,i])$\;\label{line: opsearch}
          \lIf{ $\ell < i$ }{ 
              {\bf exit}(\texttt{NOT\_FOUND})
          }        
      }
      {\bf output} $(i,j)$\;\label{line: output pair}
    }
\end{algorithm}

Observe that, when operation $\opsearch()$ is implemented with binary search, Algorithm \ref{alg:one-occ} works by performing a \emph{sequence} of (at most) $m$ simple binary searches. This is more than the single binary search of pattern matching on the Suffix (Prefix) Array, which however requires linear space in $n$ (unless implemented with more complex compressed data structures such as the \texttt{rlcsa} \cite{rlcsa}). Nevertheless, experimentally we will show that in practice binary searching $\SuA$ is as fast as binary searching $\SA$.

\medskip
\begin{lemma}
    Algorithm~\ref{alg:one-occ} is correct and complete. 
\end{lemma}

\begin{proof}
The claim is equivalent to the following: 

\begin{proposition} \label{prop: search P}
    For every $1\le i\le m$ the following hold.
    Every time Algorithm \ref{alg:one-occ} reaches Line \ref{line: output pair} with values $i,j$, $P[1,i]$ is a suffix of $T[1,j]$ so, in particular,  $P[1,i]$ occurs in $T$. Conversely, if $P[1,i]$ occurs in $T$  then Algorithm \ref{alg:one-occ} will reach Line \ref{line: output pair} with values $i,j$ such that  $P[1,i]$ is a suffix of $T[1,j]$.
\end{proposition}

We prove the claim inductively on $i$. 
At the beginning of execution, in Line \ref{line: incr i j} we set $i\leftarrow j \leftarrow 1$. 
Assume $P[i]$ occurs in $T$. If $P[i] = T[j]$, then the \texttt{if} statement is not executed, we reach Line \ref{line: output pair}, and clearly Proposition \ref{prop: search P} holds true. If, on the other hand, $P[i] \neq T[j]$, then the call $j\gets \opsearch(P[1,i])$ at Line \ref{line: opsearch} will return a position $j$ with  $P[i] = T[j]$ since $P[i]$ is a one-character extension of a right maximal string (the empty string) and $\SuA$ is a suffixient set. We conclude that, again,  we reach Line \ref{line: output pair} and Proposition \ref{prop: search P} holds true.
Assume now that $P[i]$ does not occur in $T$. Then, 
it must be the case $P[i] \neq T[j]$ so
the \texttt{if} statement will be executed and the call $j\gets \opsearch(P[1,i])$ will set $j$ to $0$, causing the termination of the algorithm with message \texttt{NOT\_FOUND}. Also in this case, Proposition \ref{prop: search P} holds true.

Assume inductively that $P[1,i-1]$, with $i>1$, occurs in $T$ (otherwise, the algorithm halted in a previous step) and that the claim holds for all prefixes of $P[1,i-1]$. 
Let $(i-1,j-1)$ be the pair output by the algorithm on prefix $P[1,i-1]$.
We consider two cases. 
(a) $P[1,i]$ occurs in $T$. Then, either (a.1) $P[i] = T[j]$, in which case the \texttt{if} statement is not executed, we reach Line \ref{line: output pair}, and Proposition \ref{prop: search P} holds true, or (a.2) $P[i] \neq T[j]$. In this case, observe that $P[1,i-1]$ is right-maximal in $T$ since both $P[1,i]$ and $P[1,i-1]\cdot T[j]$ occur in $T$. Then, the 
call $j\gets \opsearch(P[1,i])$ at Line \ref{line: opsearch} will return a position $j$ with  $P[1,i] = T[j-i+1,j]$ since
$P[1,i]$ is a one-character extension of a right maximal string  and $\SuA$ is a suffixient set. We conclude that, again,  we reach Line \ref{line: output pair} and Proposition \ref{prop: search P} holds true.
(b) $P[1,i]$ does not occur in $T$. Then, necessarily $P[i] \neq T[j]$, so the \texttt{if} statement will be executed and the call $j\gets \opsearch(P[1,i])$ will set $j$ to $0$, causing the termination of the algorithm with message \texttt{NOT\_FOUND}. Also in this case, Proposition \ref{prop: search P} holds true.
\end{proof}

Algorithm \ref{alg:one-occ} calls $\opsearch(P[1,i])$ at most  $m$  times. We can say more:  it is actually not hard to see that Algorithm \ref{alg:one-occ} calls $\opsearch(P[1,i])$ at most  $d \le m$  times, where $d$ is the node depth of the locus of $P$ in the suffix tree of $T$. If Algorithm~\ref{alg:one-occ} outputs \texttt{NOT\_FOUND}, then $d$ is 1 more than the node depth of the locus of the longest prefix of $P$ that occurs in $T$. 
We conclude that, if $\opsearch(P[1,i])$ is implemented via binary search on $\SuA$ and if random access on $T$ costs constant time per extracted character, then Algorithm \ref{alg:one-occ} runs in $O(m\cdot d\log\chi) \subseteq O(m^2\log\chi)$ time. More in general we obtain:

\medskip
\begin{theorem}\label{thm: pattern matching}
    Let $T[1,n]$ be a text stored in a data structure supporting the extraction of $T[i]$, for any $i\in[n]$, in time $t_a$. 
    Then, given a pattern $P[1,m]$, using the Suffixient Array of $T$ ($\chi$ words) we can locate one occurrence of every prefix $P[1,j]$ ($j\in [m]$) in total $O(t_a\cdot m\cdot d \log \chi) \subseteq O(t_a\cdot m^2 \log \chi)$ time, where  $d\le m$ is the node depth of the locus of $P$ in the suffix tree of $T$.
\end{theorem}
\medskip

The above query time can be reduced to $O(t_a\cdot m\cdot d +d\log\chi)$ by using classic longest-common-suffix (LCS) data structures taking $O(\chi)$ additional words of space (see \cite{manber1990suffix}).

In many practical applications, $d\ll m$. This can be understood with the fact that, in the suffix tree of a uniform random string, the string depth of the deepest branching node (equivalently, the length of the longest repeated substring) is $O(\log_\sigma n)$ with high probability, which implies $d\in O(\log_\sigma n)$ with high probability. We can say more. Differently from the FM-index \cite{FM00} and the $r$-index \cite{GNP20}, a simple analysis shows that our new index is cache-efficient, provided that the text is stored contiguously in memory. While this seems a hard requirement for a \emph{compressed} random access oracle, in Section \ref{sec:experiments} we will show such an efficient oracle.
We analyze this phenomenon in the I/O model where a block (e.g. a cache line) fits $B$ characters:

\medskip
\begin{theorem}\label{thm: pattern matching I/O}
    Let $T[1,n]$ be a text stored contiguously in main memory.
    Then, given a pattern $P[1,m]$, using the Suffixient Array of $T$ ($\chi$ words) we can locate one occurrence of  $P$ with  $O( (1+m/B) \cdot d \log \chi)$ I/O complexity, where  $d\le m$ is the node depth of the locus of $P$ in the suffix tree of $T$.
\end{theorem}
\begin{proof}
    As observed above, Algorithm \ref{alg:one-occ} runs at most  $d$  binary searches, totaling $O(d\log\chi)$ binary search steps. At each step we need to compare at most $m$ consecutive characters of $P$ and $T$, an operation having $O(1+ m/B)$ I/O complexity.
\end{proof}
\medskip

In comparison, the FM-index \cite{FM00} and the $r$-index \cite{GNP20} have $\Omega(m)$ I/O complexity (i.e. every pattern character triggers a I/O operation). 
In Section \ref{sec:implementation} we will discuss heuristics that in practice tend to drastically reduce the impact of the term $d \log \chi$ of Theorem \ref{thm: pattern matching I/O}, effectively bringing the I/O complexity very close to the optimum.

\medskip

\begin{remark}[Locating all occurrences]
As mentioned in the introduction, with our mechanism it is actually possible to count and locate \emph{all} pattern occurrences (rather than just one) efficiently in the I/O model using $O(\bar r)$ space on top of the random access oracle. 
This line of research, however, is out of the scope of this article and will be described in a future publication under preparation. 
\end{remark}

\subsection{Finding maximal exact matches} \label{subsec: mem}

Algorithm \ref{alg:one-occ} 
can be extended 
so that it returns all MEMs of $P$ in $T$ within the same time complexity. We present this procedure in Algorithm \ref{alg:MEMs}.

\begin{algorithm}[th!]
\caption{Finding all MEMs of $P [1, m]$ in $T [1, n]$}
\label{alg:MEMs}
\LinesNumbered
\SetKwInOut{Output}{output}
    \Output{All MEMs $(i,j,\ell)$ of $P$ with respect to $T$}
    
    $i \gets j \gets 1$\; 
    $\ell \gets 0$\; 

    \While{$i\le m$}{\label{line: while mems}

        $(j',\ell') \gets \opsearch(P[i-\ell,i])$\;\label{line: opsearch mems}
        \lIf{$\ell'<\ell+1$ {\bf and} $\ell>0$}{{\bf report} $(i-1,j-1,\ell)$}\label{line: report1 mems}
        $\Delta \gets LCP(P[i+1,m], T[j'+1,n])$\;\label{line: LCP mems}
        $(i,j,\ell) \gets (i+\Delta+1, j'+\Delta+1, \ell'+\Delta)$\;\label{line: increment mems}
    }

    {\bf report} $(i-1,j-1,\ell)$\;\label{line: report2 mems}
    
\end{algorithm}

If $P[i]$ does not occur in $T$, then Algorithm \ref{alg:MEMs} simply does not output any triple of the form $(i,\cdot,\cdot)$. Assuming some oracle for functions $\opsearch(\cdot)$ and $LCP(\cdot, \cdot)$, we prove:

\medskip
\begin{lemma}\label{lemma:alg MEMs}
Given a pattern $P[1,m]$, Algorithm~\ref{alg:MEMs} reports all the MEMs of $P$ in $T$.
\end{lemma}

\begin{proof}

The claim is implied by the following invariant. Every time the condition of the \texttt{while} statement (Line \ref{line: while mems}) is evaluated, the following Properties hold:

\begin{enumerate}
    \item $\ell = LCS(T[1,j-1],P[1,i-1])$, 
    \item if $\ell>0$ and $i\le m$, then $P[i] \neq T[j]$,
    \item $LCS(T[1,j'-1],P[1,i-1]) \le \ell$ for every $j' \in [n]$, and
    \item All MEMs $(i',j',\ell')$ such that $i' < i-1$ have been reported.
\end{enumerate}

The invariant trivially holds the first time we enter in the \texttt{while} loop. 

We proceed by induction. Assume that the invariant holds at Line \ref{line: while mems} for some values of $j,\ell$, and $i\le m$. We want to show that it still holds after executing the operations in Lines \ref{line: opsearch mems}-\ref{line: increment mems}. First observe that, by Property (1) of the invariant, $\opsearch(P[i-\ell,i]) = \opsearch(P[1,i])$ holds in Line \ref{line: opsearch mems}. This is a first difference with Algorithm~\ref{alg:one-occ}, which instead calls $\opsearch(P[1,i])$ (we will need the variant of Line \ref{line: opsearch mems} in the next subsection). 

Line \ref{line: opsearch mems} does not modify $i,j,\ell$, so it maintains the invariant valid. After running Line \ref{line: opsearch mems}, by definition of $\opsearch$ we have that $P[i-\ell'+1,i] = T[j'-\ell'+1,j']$ is the longest string (of length $\ell'$) suffixing a text prefix $T[1,x]$ ending in a position $x\in \SuA$ (in particular, $x=j'$ holds in this case). Observe that it cannot be the case that $\ell' > \ell+1$; if this was the case, then $LCS(T[1,j'-1],P[1,i-1]) = \ell'-1 > \ell$, contradicting Property (3) of the invariant. 
Moreover, in Line \ref{line: report1 mems} we do not report any MEM with $\ell=0$, so we can assume that $\ell>0$ to analyze which MEMs are reported in Line \ref{line: report1 mems}.
We have therefore to only distinguish two cases in Line \ref{line: report1 mems}. 

(Case A) If $\ell' = \ell+1$, then $(i-1,j-1,\ell)$ is not a MEM (because $P[i-\ell,i] = T[j'-\ell,j']$, i.e. $P[i-\ell,i-1]$ can be extended to the right), and we correctly do not output it. 

(Case B) If, on the other hand, $\ell' < \ell+1$ then we claim that $(i-1,j-1,\ell)$ is a MEM (and we correctly report it). To see this, observe that (i) $P[i-\ell,i-1]$ cannot be extended to the left (i.e. $P[i-\ell-1,i-1]$ does not occur in $T$). This is implied by Property (3) of the invariant. (ii) $P[i-\ell,i-1]$ cannot be extended to the right (i.e. $P[i-\ell,i]$ does not occur in $T$). To see this, assume for a contradiction that $P[i-\ell,i]$ occurs in $T$. Let $\alpha = P[i-\ell,i-1]$. Then, property (2) of the invariant implies that $\alpha$ is right maximal, since both $\alpha\cdot P[i]$ and $\alpha\cdot T[j]$ occur in the text (with $P[i]\neq T[j]$). But then, by definition of suffixient set, there must exist $j''\in \SuA$ such that  $\alpha\cdot P[i] = P[i-\ell,i]$ (of length $\ell+1$) suffixes $T[1,j'']$. This contradicts the fact that $\opsearch(P[i-\ell,i])$ returned $(j',\ell')$ with $\ell'<\ell+1$. From (i) and (ii), we conclude that $(i-1,j-1,\ell)$ is a MEM so Line \ref{line: report1 mems} correctly reports it. 

Note that the above discussion implies that the algorithm correctly reported all MEMs ending before pattern position $i-1$ included.

At this point, we prove that $P[i-\ell'+1,i]$ cannot be extended to the left, i.e. that $P[i-\ell',i]$ does not occur in $T$. Assume, for a contradiction, that $P[i-\ell',i]$ (of length $\ell'+1$) occurs in $T$ and let $j''$ be such that $P[i-\ell',i] = T[j''-\ell',j'']$. Recall that $\ell' \le \ell+1$ must hold. 
If $\ell' = \ell+1$, then 
$LCS(T[1,j''-1], P[1,i-1]) \ge \ell' = \ell+1$, contradicting Property (3) of the invariant.
If $\ell' \le \ell$, then $T[j''-\ell',j''-1]$ (of length $\ell'$) is a suffix of $T[j-\ell,j-1]$, i.e. $T[j''-\ell',j''-1] = T[j-\ell',j-1] = P[i-\ell',i-1]$. But then $\alpha = P[i-\ell',i-1]$ is right-maximal: both $\alpha\cdot P[i]$ and $\alpha\cdot T[j]$ (with $P[i]\neq T[j]$ by Property (2) of the invariant), of length $\ell'+1$, occur in $T$. 
Again, by definition of suffixient set, this contradicts the fact that $\opsearch(i-\ell, i)$ returned $(\ell',j')$.

To conclude, let $\Delta = LCP(P[i+1,m], T[j'+1,n])$ as per Line \ref{line: LCP mems}. By definition of $LCP()$: (i) Since above we proved that $P[i-\ell'+1,i]$ cannot be extended to the left, it follows that the same holds for $P[i-\ell'+1,i+\Delta]$, (ii) $P[i+\Delta+1] \neq T[j'+\Delta+1]$ (importantly, note that $j'+\Delta+1 \le n$ because of the terminator $\$$ at the end of $T$; $i+\Delta+1$, on the other hand, can be equal to $m+1$: in that case, we exit the \texttt{while} loop), (iii) $P[i-\ell'+1,i+\Delta] = T[j'-\ell'+1,j'+\Delta]$, and (iv) $P[i-\ell'+1,i+\delta]$ does not correspond to a MEM for all $0 \le \delta < \Delta$. 
Conditions (i-iv) immediately imply the invariant is still valid 
on $i,j,\ell$ after the assignment  at Line \ref{line: increment mems}: Properties (1) and (3) of the invariant follow from (i) and (iii), Property (2) is equivalent to (ii), and property (4) follows from (iv) and from the fact that, as proved above, the algorithm correctly reported all MEMs ending before pattern position $i-1$ included. 
\end{proof}

Implementing $\opsearch(\cdot)$ by simple binary search on $\SuA$ and random access on $T$, and $LCP(\alpha,T[i,n]) = \Delta$ with a sequence of $\Delta+1$ random access queries on $T$, we obtain: 

\medskip
\begin{theorem}
    Let $T[1,n]$ be a text stored in a data structure supporting the extraction of $T[i]$, for any $i\in[n]$, in time $t_a$. 
    Then, given a pattern $P[1,m]$, using the Suffixient Array of $T$ ($\chi$ words) we can find all MEMs of $P$ in $T$ in total $O(t_a\cdot m^2\log \chi)$ time. 
\end{theorem}
\begin{proof}
    Variable $i$ in Algorithm \ref{alg:MEMs} increases at least by one unit at every iteration of the \texttt{while} loop, and the loop terminates when $i > m$. 
    In every iteration, we perform one call to $\opsearch()$ ($O(t_a\cdot m\log \chi)$ time with binary search on $\SuA$ and random access on $T$), so overall all the calls to this operation cost $O(t_a\cdot m^2\log \chi)$ time. Finally, operation $LCP(P[i+1,m], T[j'+1,n]) = \Delta$ implemented via random access costs $O(t_a\cdot \Delta)$ time. Since $i$ is incremented by $\Delta+1$ in each step (until surpassing value $m$), it follows that the total time spent on $LCP()$ is $O(t_a\cdot m)$.
\end{proof}
\medskip

\subsection{Faster queries} \label{subsec:z-fast tries}

Using state-of-the-art techniques, our query times can be reduced from quadratic to linearithmic. 
Since a solution for MEMs also implies a solution for exact pattern matching, in this subsection we only discuss solutions for MEMs. 
We need two ingredients: $z$-fast tries \cite{BelazzouguiBPV09} (in the revisited version of \cite{boldi-z-fast}\footnote{Although \cite[Thm. 6]{boldi-z-fast} states a query time of $O(\log m + \sigma)$, below the theorem's statement the authors briefly observe that, within the same asymptotic space, query time can be reduced to $O(\log m)$.}) and Straight-Line Programs.

\medskip
\begin{lemma}[z-fast trie, {\cite[Thm. 6-7]{boldi-z-fast}}]\label{lem:z-fast}
    Given a set $Z=\{\alpha_1,\cdots,\alpha_q\}$ of $q$ strings sorted in co-lexicographic order ($i<j$ implies $\alpha_i <_{\text{colex}} \alpha_j$), there exists a data structure of $O(q)$ words of space that answers the following query: given a string $P[1,m]$, 
     after a $O(m)$-time pre-processing of $P$, for any given substring $\beta$ of $P$
    return the index $i$ of the string $\alpha_i\in Z$ that has the longest common suffix of length $\ell$ with $\beta$ (or any index $i\in [q]$ if $\ell=0$). 
    The query is answered in $O(\log m)$ time for each such queried substring $\beta$ of $P$ (after the linear pre-processing). 
    The result is guaranteed to be correct if $\beta$ is a substring of $\alpha_j$ for some $j\in [q]$.
\end{lemma}
\medskip

The last condition of  Lemma \ref{lem:z-fast}  is needed because the z-fast trie data structure \cite[Thm. 6-7]{boldi-z-fast} is based on hashing, and at construction time the best one can do is to build a hash function that is guaranteed to be perfect (collision-free) only on substrings of the input. Otherwise, the structure fails with low probability only \cite[Thm. 7]{boldi-z-fast}. 

\medskip

A Straight-Line Program (SLP) for a text $T$ is a Context-Free Grammar generating only $T$ and whose right-expansions are either one alphabet symbol or the concatenation of two nonterminals  \cite{BilleLRSSW15}. For the purposes of this section, it is only important to know that the size (number of rules) $g$ of a smallest SLP for $T$ is a good measure of repetitiveness (see \cite{Nav22a}). We use the following two SLP-based structures:

\medskip

\begin{lemma}[random access on SLPs, \cite{BilleLRSSW15}] \label{lemma:slp random access}
Given an SLP with $g$ rules for $T[1,n]$, there exists a $O(g)$-space (memory words) data structure answering random access queries $T[i]$ in $O(\log n)$ time. 
\end{lemma}
\medskip

\begin{lemma}[LCP/LCS queries on SLPs, {\cite[Lemma 1]{balaz24Wheelermaps}}] \label{lemma:slp lcp lcs}
Given an SLP with $g$ rules for $T[1,n]$, there exists a $O(g)$-space (memory words) data structure with which we can preprocess any pattern $P[1,m]$ in $O(m)$ time such that later, given $i, j$ and $q$, we can return
$LCS(P[i,j], T[1,q])$ (or the symmetric $LCP(P[i,j], T[q,n])$ function) in $O(\log n)$ time \footnote{While in \cite[Lemma 1]{balaz24Wheelermaps} they only discuss $LCP()$, the symmetric function $LCS()$ can be supported by simply reversing the text, an operation that does not change the SLP size $g$.} and with no chance of error as long as $P[i,j]$ occurs somewhere in $T$.
\end{lemma}
\medskip

Again, in Lemma \ref{lemma:slp lcp lcs} the requirement that $P[i,j]$ occurs somewhere in $T$ is needed because the data structure uses hashing to compare strings, and perfect hashing is guaranteed only on $T$'s substrings.
Combining Lemma \ref{lem:z-fast} and Lemma \ref{lemma:slp lcp lcs}, we obtain the following lemma. Observe that the claim requires only $\beta[1,k-1]$ (not the full queried string $\beta$) to be a substring of $T$. This will be crucial for the correctness of our MEM-finding algorithm.  

\medskip
\begin{lemma}\label{lem:LCS structure}
    Given a text $T[1,n]$ for which there exists a SLP of size $g$ and given a Suffixient Array $\SuA[1,\chi]$ for $T$, there exists a data structure of $O(\chi + g)$ words of space returning $(x,\ell) = \opsearch(\beta)$ on $(T,\SuA)$ (Definition \ref{def:LCS query}) in $O(\log n)$ time, for any string $\beta[1,k]$. 
    The result is guaranteed to be correct if $\beta[1,k-1]$ is a substring of $T$.
\end{lemma}
\begin{proof}
    We partition $\SuA$ into the $\sigma$ sub-arrays  $\SuA_a = \langle i-1:\ i \in \SuA \wedge T[i]=a \rangle$, for all $a\in \Sigma$, containing the text positions $i-1$ such that $i\in \SuA$ and $T[i] = a$.
    Then, 
    for all $a\in \Sigma$ 
    we build the z-fast trie of Lemma \ref{lem:z-fast} on the set of strings $Z_a = \{T[1,j]\ :\ j\in \SuA_a\}$.
    Letting $\sigma' \leq \chi$ denote the effective alphabet size (number of distinct characters appearing in $T$), 
    these z-fast tries can be stored in an array of size $\sigma' \leq \chi$, whose indexes can be computed from the alphabet's characters via a perfect hash function ($O(\sigma') \subseteq O(\chi)$ space). This way, given $a\in\Sigma$ we can access the z-fast trie for $Z_a$  in $O(1)$ time.    
    Importantly, we build these z-fast tries so that they share the same hash function, constructed so as to be collision-free on the substrings of $T$.
    This implies that the z-fast tries return correct answers, provided that they are queried with substrings of $T$.
    We also build the LCS structure of Lemma \ref{lemma:slp lcp lcs} on $T$.

    Recall that the z-fast trie for $Z_a$ returns indexes $i$ in the sorted $Z_a$. 
    Let $i$ be the index for $T[1,j]$ in the z-fast trie for $Z_a$. By using $O(\chi)$ additional words (e.g. a perfect hash function) we associate each such pair $(i,a)$ to the corresponding text position $j$. 
    In the following we will therefore assume that the z-fast tries return such text positions $j$ (rather than indexes in $Z_a$).

    To find $(x,\ell) = \opsearch(\beta)$, we proceed as follows. 
    Note that  $\beta[k] \neq T[j]$ for all $j\in \SuA$ if and only if $\beta[k]$ does not appear in $T$ (because $\SuA$ is a string attractor). We can easily discover this event using a $O(\chi)$-space perfect hash function storing all distinct text's characters. In this case, we return the pair $(x,\ell) = (0,0)$.

    Let us assume therefore that $\beta[k] = T[j]$ for some $j\in \SuA$.
    We query the z-fast trie on $Z_{\beta[k]}$ with string $\beta[1,k-1]$. Since $\beta[1,k-1]$ is guaranteed to be a substring of $T$, the z-fast trie correctly returns position $j\in \SuA_{\beta[k]}$ such that (i) $j$  maximizes  $LCS(T[1,j],\beta[1,k-1])$ among all positions in $\SuA_{\beta[k]}$, and (ii) $T[j+1] = \beta[k]$. But then, this implies that $j+1\in\SuA$ maximizes $LCS(T[1,j+1],\beta[1,k])$ among all positions in $\SuA$. 
    
    After finding such $j+1 \in \SuA$, we compute $\ell-1 = LCS(T[1,j],\beta[1,k-1])$ using  Lemma \ref{lemma:slp lcp lcs}. Again, the structure is guaranteed to return the correct answer since $\beta[1,k-1]$ is a substring of $T$. 
    To conclude, the result of the query is the pair $(x,\ell)$ with $x=j+1$.    
\end{proof}
\medskip

Implementing $LCP(\alpha,T[i,n]) = \Delta$ with a sequence of $\Delta+1$ random access queries on $T$, we obtain: 

\medskip
\begin{theorem}\label{thm: faster mems}
    Let $T[1,n]$ be a text with Suffixient Array size $\chi$ for which there exists a SLP of size $g$.
    Then Algorithm \ref{alg:MEMs}, when implemented with the structure of Lemma \ref{lem:LCS structure} to support operation $\opsearch()$ and with the random access structure of Lemma \ref{lemma:slp random access} to support $LCP()$, finds all MEMs of any pattern $P[1,m]$ in $T$ in $O(m\log n)$ time using $O(\chi + g)$ words of space. The algorithm always returns the correct answer. 
\end{theorem}
\begin{proof}
    Variable $i$ in Algorithm \ref{alg:MEMs} increases at least by one unit at every iteration of the \texttt{while} loop, and the loop terminates when $i > m$. 
    In every iteration, we perform one call to $\opsearch()$ ($O(\log n)$ time with Lemma \ref{lem:LCS structure}), so overall all the calls to this operation cost $O(m\log n)$ time. Finally, operation $LCP(P[i+1,m], T[j'+1,n]) = \Delta$ implemented via random access costs $O(\Delta\log n)$ time using the structure of Lemma \ref{lemma:slp random access}. Since $i$ is incremented by $\Delta+1$ in each step (until surpassing value $m$), it follows that the total time spent on $LCP()$ is also $O(m\log n)$.

    We now prove correctness. 
    If the structure of Lemma \ref{lem:LCS structure} always returns the correct result, then the algorithm is correct by Lemma \ref{lemma:alg MEMs}. 
    Note that condition (1) of the invariant used in the proof of Lemma \ref{lemma:alg MEMs} implies that, in Line \ref{line: opsearch mems} of Algorithm \ref{alg:MEMs},  $P[i-\ell,i-1]$ is always a substring of $T$. But then, letting $\beta = P[i-\ell,i]$ and $k = |\beta|$, this means that all calls to $\opsearch$ are of the form $\opsearch(\beta)$, where $\beta[1,k-1]$ appears in $T$. This is precisely the condition required by Lemma \ref{lem:LCS structure} in order for the structure to return correct results. This concludes the proof.
\end{proof}

\subsection{Even faster queries}

We provide a faster solution in the case where pattern $P$ is not chosen \emph{adversarially}. With this term, we mean that substrings of $P$ do not collide with substrings of $T$ through the hash functions used by the structures of Lemmas 
\ref{lemma:slp lcp lcs} and \ref{lem:LCS structure} with high probability; in other words, we assume that an adversary did not craft $P$ so that it generates such collisions. This is a realistic scenario in practice since a string chosen independently from (the random choices in) those data structures  will collide with  a text substring of $T$ through those hash functions with low (inverse polynomial) probability only (see \cite{balaz24Wheelermaps,boldi-z-fast}): 

\medskip
\begin{remark}{\cite{balaz24Wheelermaps,boldi-z-fast}}\label{rem: adversarial}
    The structures of Lemma \ref{lemma:slp lcp lcs} and Lemma \ref{lem:LCS structure} return the correct result with high probability, assuming that their inputs (respectively, $P[i,j]$ and $\beta$) are not chosen by an adversary on the basis of the random seeds in the hash functions used by those data structures. To abbreviate, in such a case we will say that the inputs are ``~not chosen adversarially~''.
\end{remark}
\medskip

Observe that the non-adversarial hypothesis is the same required in other randomized data structures such as classic hash tables and Bloom filters \cite{bloom1970space} (where the random choices precede the arrival of the input) and is common in other works using z-fast tries to index text \cite{boldi-z-fast}. 

Algorithm \ref{alg:MEMs} can be understood as a traversal of $d\le m$ suffix tree edges, 
where $d$ is the number of times we would fully or partially descend edges in the suffix tree of $T$ while finding MEMs (a generalization to the context of MEMs of parameter $d$ used in Theorem \ref{thm: pattern matching}). See Figure~\ref{fig:descents} for an example showing that $d$ can be much smaller than $m$.

\begin{figure}[th!]
\begin{center}
\includegraphics[width=.88\textwidth]{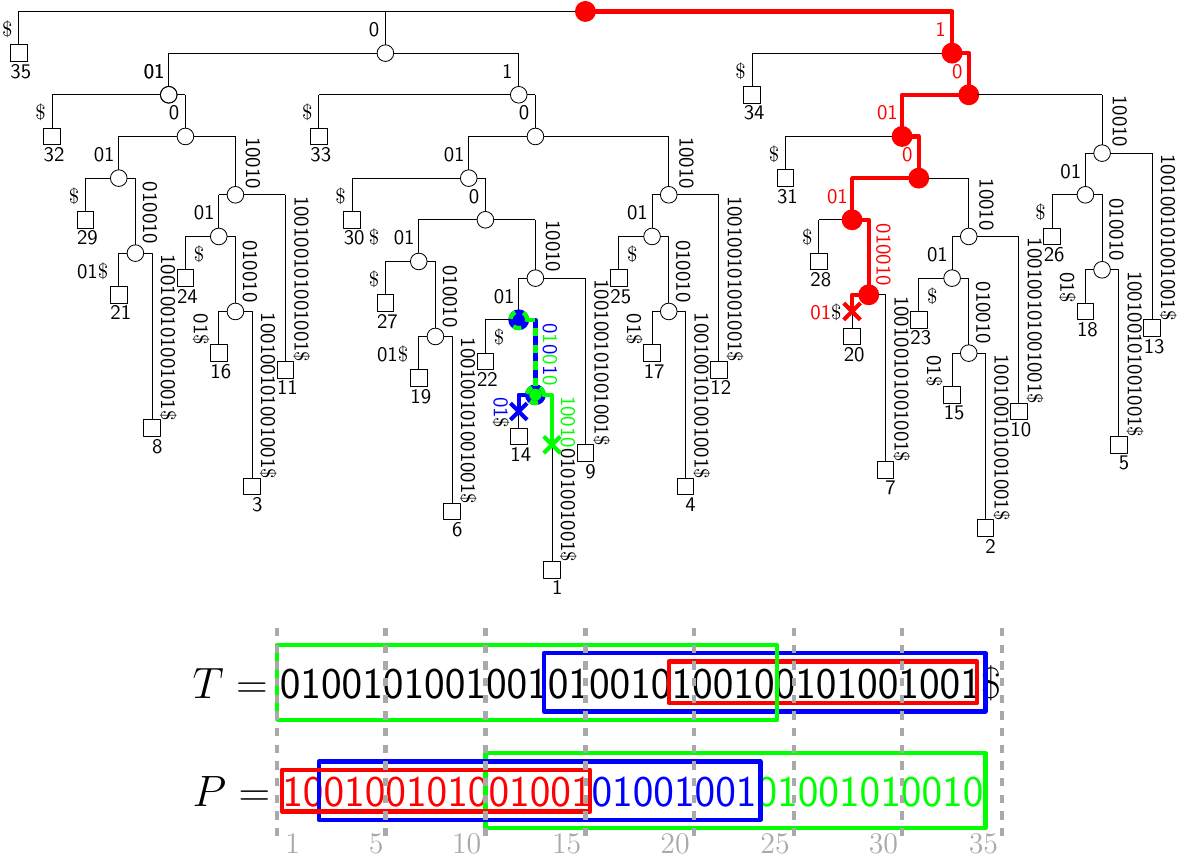}
\caption{The $d=11$ times we fully or partially descend 10 distinct edges in the suffix tree of $T$ {\bf (above)} while finding the 3 MEMs for our example {\bf (below)}.  The MEMs are shown boxed in $P$ and $T$, with the characters' colors in $P$ also indicating which path we are following in the tree when we read them.  The characters in the box for a MEM that are a different color from the box are the path label of the node we reach by suffix links and descend from when finding the end of that MEM.  We descend the line alternating blue and green twice. Notice that $11 = d \ll m = 34$.}
\label{fig:descents}
\end{center}
\end{figure}

By following the same analysis carried out in the proof of Theorem \ref{thm: faster mems}, we immediately obtain: 
\medskip

\begin{theorem}
    Let $T[1,n]$ be a text with Suffixient Array size $\chi$ for which there exists a SLP of size $g$.
    Then Algorithm \ref{alg:MEMs}, when implemented with the structure of Lemma \ref{lem:LCS structure} to support operation $\opsearch()$ and with the structure of Lemma \ref{lemma:slp lcp lcs} to support function $LCP()$, finds all MEMs of any pattern $P[1,m]$ in $T$ in $O(m + d\log n)$ time using $O(\chi + g)$ words of space, where $d\le m$ is the number of times we would fully or partially descend edges in the suffix tree of $T$ while finding those MEMs.
    The result is correct with high probability, provided that $P$ 
    is not chosen adversarially.
\end{theorem}

\subsection{Offline solutions with only random access and z-fast tries}
\label{sec:offline_queries}

Interestingly, fast random access to $T$ and a z-fast trie on the set of strings $Z = \{T[1,j]\ :\ j\in \SuA\}$ are by themselves sufficient for fast offline pattern matching.  To see why, assume $P$ does occur in $T$ and we already know the location of an occurrence $T [j - i + 1, j]$ of $P [1, i]$ in $T$.  If $T [j + 1] = P [i + 1]$ then $T [j - i + 1, j + 1] = P[1, i + 1]$.  Otherwise, querying the z-fast trie with $P [1, i + 1]$ will return in $O (\log (i + 1)) \subseteq O (\log m)$ time a position $j' + 1$ such that $T [j' - i + 1, j' + 1] = P [1, i + 1]$.

It follows that in $O (t_a \cdot m + d \log m)$ time we can find a position $j^*$ such that if $P$ occurs in $T$ then $T [j^* - m + 1, j^*] = P$, which we can then check in $O (t_a \cdot m)$ time.  Here, $d \leq m$ is again the number of times we would fully or partially descend edges in the suffix tree of $T$ while finding the MEMs of $P$ with respect to $T$ (so the node depth of the locus of $P$ if $P$ occurs in $T$).  If $P$ occurs in $T$ then our time bound holds in the worst case; otherwise, it holds with high probability assuming an adversary has not chosen $P$ on the basis of the hash function used in our z-fast trie.  In any case, our solution works in $O (m (t_a + \log m))$ time.

This solution is offline in the sense that we may not know until we have finished whether any prefix $P [1..i]$ of $P$ occurs in $T$.  Moreover, if $P$ does not occur in $T$ then we do not learn which is the longest prefix of $P$ that occurs in $T$.

\medskip

\begin{theorem}
\label{thm:offline_pattern_matching}
Suppose we already have $t_a$-time random access to a text $T[1,n]$ with Suffixient Array size $\chi$.  Then we can store in $O (\chi)$ space a z-fast trie such that in $O (m (t_a + \log m))$ time we can check whether $P [1, m]$ occurs in $T$ and, if so, return the location of one occurrence.
\end{theorem}

\medskip

We can adapt this idea to offline MEM-finding.  Suppose $T$ has a Suffixient Array $\SuA$ of size $\chi$, $\rev{T}$ has a Suffixient Array $\SuA'$ of size $\chi'$, and we have fast random access to $T$ (and thus also to $\rev{T}$) and z-fast tries for the sets of strings $Z = \{T[1,j]\ :\ j\in \SuA\}$ and $Z' = \{\rev{T} [1, j]\ :\ j \in \SuA'\}$.  These z-fast tries take $O (\chi + \chi')$ total space.

Assume we already know the ending position $j$ of a prefix $T [1, j]$ of $T$ with the longest common suffix with $P [1, i]$ of any prefix of $T$, possibly without knowing the length of that common suffix.  If $T [j + 1] = P [i + 1]$ then $T [1, j + 1]$ has the longest common suffix with $P [1, i + 1]$ of any prefix of $T$.  Otherwise, querying the z-fast trie for $Z$ with $P [1, i + 1]$ will return in $O (\log (i + 1)) \subseteq O (\log m)$ time a position $j' + 1$ such that $T [1, j' + 1]$ has the longest common suffix with $P [1, i + 1]$ of any prefix of $T$.

It follows that in $O (t_a \cdot m + d \log m))$ time we can build an array $J [1, m]$ such that $T [1, J [i]]$ has the longest common suffix with $P [1, i]$ of any prefix of $T$, for $i \leq m$, where $d \leq m$ is again the number of times we would fully or partially descend edges in the suffix tree of $T$ while finding the MEMs of $P$ with respect to $T$.

Because the prefixes of $\rev{T}$ are the suffixes of $T$, in $O (t_a \cdot m + d' \log m)$ time we can also build an array $J' [1, m]$ such that $T [J' [i], n]$ has the longest common prefix with $P [i, m]$ of any suffix of $T$, where $d'$ is the number of times we would fully or partially descend edges in the suffix tree of $\rev{T}$ while finding the MEMs of $\rev{P}$ with respect to $\rev{T}$.

We can find the ending position $e_1$ of the leftmost MEM $P [1, e_1]$ in $P$ in $O (t_a \cdot e_1)$ time, by comparing the characters in $P [1, m]$ and $T [J' [1], n]$ until we find a mismatch at $P [e_1 + 1]$.  Since $P [e_1 + 1]$ is in the second MEM $P [s_2, e_2]$ in left-to-right order, we can find that MEM's starting position $s_2$ in $O (t_a (e_1 - s_2 + 1))$ time by comparing the characters in $P [1, e_1 + 1]$ and $T [1, J [e_1 + 1]]$ from right to left until we find a mismatch at $P [s_2 - 1]$.

Since the second MEM in $P$ is the first MEM in $P [s_2, m]$, we can repeat this whole process to find the ending position of the second MEM and the starting position of the third MEM, and so on.  This step is similar to Li's~\cite{Li12} forward-backward algorithm and takes time proportional to $t_a$ times the total length of the MEMs (we could also use as a base an algorithm from~\cite{Gag24} and~\cite[Section 2.6]{10.1093/bioinformatics/bts280}, which is faster in practice).

With high probability, our algorithm works correctly and in $O (t_a \cdot m + (d + d') \log m)$ time plus time proportional to $t_a$ times the total length of the MEMs, assuming an adversary has not chosen $P$ on the basis of the hash function used in our z-fast tries.

\medskip

\begin{theorem}
\label{thm:offline_MEM-finding}
Suppose we already have $t_a$-time random access to a text $T[1,n]$ with Suffixient Array size $\chi$ whose reverse has Suffixient Array size $\chi'$.  Then we can store in $O (\chi + \chi')$ space two z-fast tries such such that with high probability we can find the MEMs of $P [1, m]$ with respect to $T$ in $O (t_a \cdot m + (d + d') \log m))$ time plus time proportional to $t_a$ times the total length of the MEMs, assuming an adversary has not chosen $P$ on the basis of the hash function used in our z-fast trie.
\end{theorem}

\section{Characterization of Smallest Suffixient Sets}
\label{sec:small_suff}

We now move to the problems of computing Suffixient Arrays.
We start in this section
by characterizing  smallest suffixient sets. We first define the concept of \emph{supermaximal extensions} (Definition \ref{def:supermaximal extension}). 
Then, in Definition \ref{def:SSS} we define a set (actually, a family of sets) $\SSS$ containing all positions in $T$ ``capturing'' all (and only the) supermaximal extensions of $T$. Finally, in Lemma \ref{lem:SSS} we prove that $\SSS$ is a suffixient set of minimum cardinality.

\medskip
\begin{definition}\label{def:supermaximal extension}
    We say that $T[i,j]$ (with $j\geq i$) is a \emph{supermaximal extension} if $T[i,j-1]$ is right-maximal and, for each right-maximal $T[i',j'-1] \neq T[i,j-1]$ (with $i' \le j' \le n$), $T[i,j]$ is not a suffix of $T[i',j']$.
\end{definition}
\medskip

The following definition depends on a particular total order $<_t$ on $[n]$. We use $<_t$ to break ties between equivalent candidate positions of the suffixient set, choosing the position of maximum rank according to $<_t$ in case of a tie. In order to make notation lighter, we do not parameterize the set on the particular tie-breaking strategy ($<_t$ will always be clear from the context).

\medskip
\begin{definition}\label{def:SSS}
Let $<_t$ be any total order on $[n]$. We define a set $\SSS \subseteq [n]$ as follows: $x \in \SSS$ if and only if there exists a supermaximal extension $T[i,j]$ such that (i) $T[i,j]$ is a suffix of $T[1,x]$, and (ii) for all prefixes $T[1,y]$ suffixed by $T[i,j]$, if $y\neq x$ then $y <_t x$.
\end{definition}
\medskip

We prove that $\SSS$ is a suffixient set of minimum cardinality:

\medskip
\begin{lemma}\label{lem:SSS}
For any tie-breaking strategy (total order) $<_t$, $\SSS$ is a suffixient set of minimum cardinality $\chi$ for $T$. 
\end{lemma}
\begin{proof}
    To see that $\SSS$ is suffixient, consider any right-maximal substring $T[i,j-1]$ ($i \le j \le n$). We want to prove that there exists $x \in \SSS$ such that $T[i,j]$ is a suffix of $T[1,x]$.
    Let $T[1,x]$ be a prefix of $T$ being suffixed by $T[i,j]$, breaking ties (on its endpoint $x$) by $<_t$. 
    If $T[i,j]$ is a supermaximal extension, then by Definition \ref{def:SSS} it holds $x\in \SSS$ and we are done. 
    Otherwise, let $T[i',j'] \neq T[i,j]$ be a 
    one-character extension of a right-maximal string $T[i',j'-1]$
    such that $T[i,j]$ is a suffix of $T[i',j']$.
    Without loss of generality, let $T[i',j']$ be the string of maximum length $|T[i',j']| = j'-i' + 1$ with this property.
    Observe that there cannot be a right-extension $T[i'',j''] \neq T[i',j']$ of a right-maximal string $T[i'',j''-1]$ such that
    $T[i',j']$ is a suffix of 
    $T[i'',j'']$, otherwise 
    $T[i,j]$ would be a suffix of $T[i'',j'']$ and
    $|T[i'',j'']|>|T[i',j']|$, contradicting the fact that $T[i',j']$ is the longest such string.  
    It follows that $T[i',j']$ is a supermaximal extension.
    Let $T[1,y]$ be a prefix of $T$ being suffixed by $T[i',j']$, breaking ties (on $y$) by $<_t$.
    Then, by Definition \ref{def:SSS}, $y\in \SSS$. Also in this case we are done, since $T[i,j]$ suffixes $T[i',j']$ and $T[i',j']$ suffixes $T[1,y]$, so by transitivity of the suffix relation, $T[i,j]$ suffixes $T[1,y]$.

    To see that $\SSS$ is of minimum cardinality, let $S$ be a suffixient set. We prove $|\SSS| \le |S|$ by exhibiting an injective function from $\SSS$ to $S$.

    Let $SE \subseteq \Sigma^+$ be the set of all supermaximal extensions.
    By definition of supermaximal extension, note that any prefix $T[1,j]$ can be suffixed by at most \emph{one} supermaximal extension $T[i,j]$: assume, for a contradiction, that $T[1,j]$ is suffixed by two distinct supermaximal extensions $T[i,j] \neq T[i',j]$. Without loss of generality, $i'<i$. But then, $T[i,j]$ is a suffix of $T[i',j]$, hence $T[i,j]$ cannot be a supermaximal extension, a contradiction.

    We first define a relation $h : \SSS \rightarrow SE$ mapping every $x\in \SSS$ to a supermaximal extension  $h(x) = T[i,j]$ such that $T[i,j]$ is a suffix of $T[1,x]$.
    From the observation above, there is at most one such string $T[i,j]$. 
    From Definition \ref{def:SSS},  there exists one such $T[i,j]$.
    We conclude that $h$ is a function. We now prove that $h$ is injective.

    To see that $h$ is injective, assume for a contradiction $h(x) = h(y)$ for some $x\neq y$ with $x,y\in \SSS$. 
    Let $h(x) = h(y) = T[i,j]$. 
    By Definition \ref{def:SSS}, there exists a supermaximal extension 
    $T[i_x,x]$ such that 
    $T[1,x]$ is the prefix of $T$ being suffixed by $T[i_x,x]$ with largest endpoint $x$, according to the total order $<_t$.
    Following the same reasoning on position $y$, we can associate an analogous supermaximal extension $T[i_y,y]$ to $y$.
    On the other hand, by the definition of $h$, also string $h(x) = h(y) = T[i,j]$ is a supermaximal extension that suffixes both $T[1,x]$ and $T[1,y]$.
    Since $T[i,j]$ and $T[i_x,x]$ are supermaximal extensions that suffix $T[1,x]$ and (as proved above) there can exist at most one  supermaximal extensions suffixing $T[1,x]$, we conclude that $T[i,j] = T[i_x,x]$. Similarly, we conclude $T[i,j] = T[i_y,y]$, hence $T[i_x,x] = T[i_y,y]$.
    We obtained a contradiction, since Definition \ref{def:SSS} requires both $x <_t y$ and $y <_t x$ and $<_t$ is a total order. We conclude that $h$ is an injective function. 

    Next, we define a relation $g : SE \rightarrow S$ from the set $SE$ of supermaximal extensions to positions in $S$. For any supermaximal extension $T[i,j] \in SE$, we define $g(T[i,j]) = x$ to be the largest element of $S$ (according to the standard order between integers) such that $T[1,x]$ is suffixed by $T[i,j]$. 
    Such a position must exist since $S$ is a suffixient set, so $g$ is a function. 
    To see that $g$ is also injective, assume for a contradiction $g(T[i,j]) = g(T[i',j']) = x$, for $T[i,j] \neq T[i',j']$ (both supermaximal extensions).  
    By the definition of $g$, both $T[i,j]$ and $T[i',j']$ are suffixes of $T[1,x]$. But then, since $T[i,j] \neq T[i',j']$, either $T[i,j]$ is a suffix of $T[i',j']$ (so $T[i,j]$ is not a supermaximal extension) or the other way round (so $T[i',j']$ is not a supermaximal extension). In both cases we get a contradiction, so we conclude that $g$ is injective. 

    We conclude the proof by observing that $g \circ h : \SSS \rightarrow S$ is the composition of two injective functions and is therefore injective. 
\end{proof}

\section{Computing Suffixient Arrays}\label{sec:smallest-SuffixientSet-comp}

To make notation lighter, in this section we focus on the core problem of computing smallest suffixient sets, i.e. (unsorted) sets $\{x\ :\ x\in \SuA\}$. 
By their nature, it will be immediate to augment our algorithms so that they output also the Prefix Array rank of every $x\in \SuA$; the Suffixient array $\SuA$ can then be obtained by sorting the positions $x\in \SuA$ according to those ranks.

We start in Subsection \ref{sec:quadratic} by giving an ``operative'' version of Definition \ref{def:SSS}. This ``operative'' definition will automatically yield a simple quadratic-time algorithm for computing a smallest suffixient set. In the next subsections, we will optimize this algorithm and reach compressed working space (Subsection \ref{sec:one pass}) and optimal linear time (Subsections \ref{sec:linear time} and \ref{sec: linear time 2}).

\subsection{A simple quadratic-time algorithm}\label{sec:quadratic}

For brevity, in the rest of the paper $\LCP$, $\BWT$, $\SA$ denote $\LCP(\rev{T})$, \\ $\BWT(\rev{T})$, and $\SA(\rev{T})$, respectively. 

First, the starting and ending positions of runs on $\BWT$ are of our particular interest:
\medskip
\begin{definition}[$c$-run break]
We say that position $1 < i \le n$ is a \emph{$c$-run break}, for $c\in\Sigma$, if $\BWT[i-1,i] = ac$ or $\BWT[i-1,i] = ca$, with $a\in \Sigma$ and $a\neq c$. 
\end{definition}
\medskip

Note that, if $i$ is a $c$-run break, then it is also an $a$-run break for some $a\neq c$.

\medskip
\begin{definition}\label{def:box}
    Given a position $i \in [n]$, we define $box(i) = [\ell,r]$ to be the maximal interval such that $i\in [\ell,r]$ and $\LCP[j] \geq \LCP[i]$ for all $j\in [\ell,r]$.
\end{definition}
\medskip

\begin{example}
See Figure \ref{fig:run-example}: $box(10) = [3,13]$ (shown with a blue box), because $\LCP[3,13] \geq \LCP[10] = 1$, and $box(19) = [17,20]$ (shown with an orange box).
\end{example}
\medskip

Next, we define the set of $c$-run breaks in $box(i)$.

\medskip
\begin{definition}\label{def:Box}
    Given $i \in [n]$ and $c\in \Sigma$, we define 
    $$
    B_{i,c} = \{ j\ :\ j\ \mathrm{is\ a\ }c\mathrm{-run\ break\ in\ } \BWT(\rev{T})\mathrm{\ and\ }j \in box(i) \}.
    $$
\end{definition}
\medskip

\begin{example}
In Figure \ref{fig:run-example}, we have $B_{10,A} = \{3, 6, 7, 10, 11, 12\}$.
\end{example}
\medskip

We indicate with $text(i) = n - \SA[i] + 1$ the position in $T$ corresponding to character $\BWT[i]$, and with $bwt(j) = \SA^{-1}[n + 1 - j]$ the inverse of function $text()$, i.e. the position in $\BWT$ corresponding to character $T[j]$.

The following Lemma \ref{lem:link LCP - supermaximal} stands at the core of all our algorithms for computing a smallest suffixient set. Intuitively, Lemma \ref{lem:link LCP - supermaximal} states that supermaximal extensions can be identified by looking at local maxima among $c$-run breaks contained in $\LCP$ ranges of the form $\LCP[box(i)]$. Since, by Definition \ref{def:SSS}, supermaximal extensions are those characterizing suffixient sets of smallest cardinality, Lemma \ref{lem:link LCP - supermaximal} gives us a tool for computing such a set using the arrays $\BWT$, $\LCP$, and $\SA$. See Figure \ref{fig:run-example}.

\medskip
\begin{lemma}\label{lem:link LCP - supermaximal}
    The following hold:
    \begin{enumerate}
        \item    Let $i$ be a $c$-run break and $i' \in \{i,i-1\}$ be such that $\BWT[i']=c$. Let $j' = text(i')$, and let $\ell = \LCP[i]$.
        If $\ell = \max \LCP[B_{i,c}]$, then the string $T[j'-\ell,j']$ is a supermaximal extension, and
        \item conversely, for any supermaximal extension $\alpha\cdot c$ ($\alpha \in \Sigma^*$, $c \in \Sigma$) there exists an occurrence $\alpha\cdot c = T[j'-\ell,j']$ and a $c$-run break $i\in [n]$ such that $\BWT[i']=c$, 
        with $i' = bwt(j') \in \{i,i-1\}$  and $\ell = \LCP[i] = \max \LCP[B_{i,c}]$.
    \end{enumerate}
\end{lemma}
\begin{proof}

\emph{(1)} 
    Let $i$ be a $c$-run break. Let $i',i'' \in \{i,i-1\}$ be such that $\BWT[i']=c$ and $\BWT[i''] = a \neq c$. Let $j' = text(i')$ and $j'' = text(i'')$. Let moreover $\ell = \LCP[i]$, and assume $\ell = \max \LCP[B_{i,c}]$. 

    Note that $T[j'-\ell,j'-1]$ is right-maximal: $\ell = \LCP[i]$ implies that $T[j'-\ell,j'-1] = T[j''-\ell,j''-1]$ and, by definition of $i', i'', j'$, and $j''$ we have that $T[j'] = \BWT[i'] \neq \BWT[i''] = T[j'']$.
    
    Assume, for a contradiction, that $T[j'-\ell,j']$ is not a supermaximal extension. Then, since $T[j'-\ell,j'-1]$ is right-maximal, it must be the case that $T[j'-\ell,j']$ is a suffix of some string $T[\hat j-\hat \ell,\hat j]$ with $\hat \ell > \ell$ and such that $T[\hat j-\hat \ell,\hat j-1]$ is right-maximal. 
    Let $\hat i = bwt(\hat j)$.
    Since $T[j'-\ell,j']$ is a suffix of $T[\hat j-\hat \ell,\hat j]$ 
    we have that $T[j'-\ell,j'-1]$ is a suffix of $T[\hat j-\hat \ell,\hat j-1]$, hence it must be the case that $\hat i \in box(i)$. 
    Since $T[\hat j-\hat \ell,\hat j-1]$ is right-maximal and $T[\hat j] = c$, 
    without loss of generality we can assume that $\hat i \in \{k-1,k\}$, where 
    $k$ is a $c$-run break (such a $c$-run break must exist in $box(i)$) and
    $\LCP[k] \geq \hat\ell$. In particular, $k \in box(i)$.
    Then,
    $\LCP[k] \geq \hat \ell > \ell$ and $k\in box(i)$ yield $\max \LCP[B_{i,c}] \geq \hat \ell > \ell$, a contradiction. 
    
    \emph{(2)}
    Let $T[j'-\ell,j']$ be a supermaximal extension and let $c = T[j']$. 
    Let $i' = bwt(j')$.    
    Since $T[j'-\ell,j'-1]$ is right-maximal, we can assume without loss of generality that either $\BWT[i'] \neq \BWT[i'-1]$ or  $\BWT[i'] \neq \BWT[i'+1]$. Let therefore $i' \in \{i,i-1\}$, where $i$ is a $c$-run break.
    By definition of $T[j'-\ell,j']$, it holds $\ell = \LCP[i]$ (because $T[j'-\ell,j'-1]$ is the longest right-maximal string suffixing $T[1,j'-1]$). 
    We need to prove that $\ell = \max \LCP[B_{i,c}]$.
    Assume, for a contradiction, that $\ell < \hat \ell = \max \LCP[B_{i,c}]$. 
    This means that there exists a supermaximal extension $T[\hat j-\hat \ell,\hat j]$ with 
    (i) $T[\hat j] = c$, 
    (ii) $\hat i \in \{k-1,k\}$, where $\hat i = bwt(\hat j)$ and
    $k \in B_{i,c}$ is a $c$-run break in $box(i)$, and
    (iii) $\LCP[k] = \hat \ell$.
    But then, since $k\in box(i)$ and $T[\hat j-\hat \ell,\hat j-1]$ is a right-maximal string of length $\hat\ell > \ell$, we have that $T[j'-\ell,j'-1]$ is a (proper) suffix of $T[\hat j-\hat \ell,\hat j-1]$. This fact and $T[\hat j]=c$  contradict the fact that $T[j'-\ell,j']$ is a supermaximal extension.
\end{proof}

\begin{figure}[!htb]
	\centering
		\includegraphics[scale=0.73]{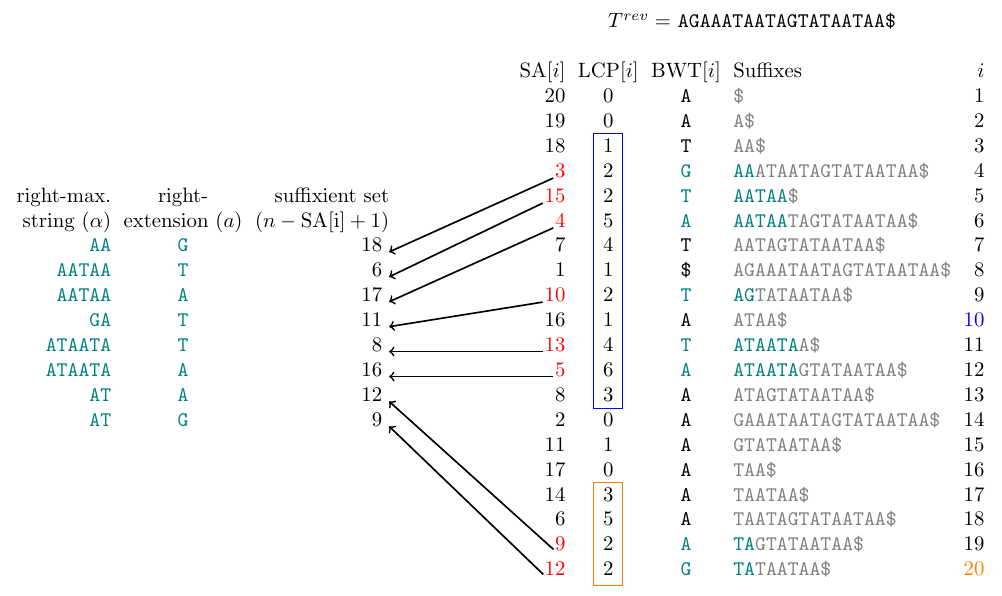}
	 \caption{
The figure shows how to construct a smallest suffixient set $\SSS$ for a text $T[1,n]$ with $n=20$ following Lemma \ref{lem:link LCP - supermaximal}. 
\emph{To the left of the black arrows}: all supermaximal extensions $\alpha\cdot a$ of $T$ and their selected ending positions in $T$, forming a suffixient set. 
\emph{To the right of the black arrows}: 
$\SA$, $\LCP$, $\BWT$ and the sorted suffixes of $\rev{T}$.
In column $\SA[i]$, we highlight in red all positions that are selected to be included in $\SSS$.
In columns \emph{Suffixes} and $\BWT[i]$, we highlight in green the (reverses of) the  supermaximal extensions of $T$.
Black arrows show how the selected $\SA$ positions are converted to positions in $T$ using the formula $n-\SA[i]+1$.
\emph{How to identify positions of $\SSS$}:
for each $c$-run break $i$, we decide if the two ranks $i'\in \{i-1,i\}$ should contribute to $\mathcal{S}$ (i.e. if they correspond to a supermaximal extension) as described in Lemma \ref{lem:link LCP - supermaximal}. 
For brevity, we show this decisional procedure only on two run breaks.
Consider the $\tt A$-run break at position $i=10$ (highlighted in blue in column $i$). 
The blue box depicts the corresponding \LCP\ 
interval, $\LCP[box(10)] = \LCP[3,13]$.
We observe that in $[3,13]$ there are other $\tt A$-run breaks $i''$, such that $\LCP[i''] > \LCP[10]$: those are $i'' = 6, 7, 11, 12$. We conclude that text position $n - \SA[10] + 1 = 5$ should not be included in $\mathcal{S}$.
Consider now the $\tt A$-run break $i=20$, highlighted in orange in column $i$.
Position $i'=20-1 = 19$ is such that $\BWT[i'] = \tt A$. 
The orange box depicts the corresponding \LCP\ interval, $\LCP[box(20)] = \LCP[17,20]$.
In this case, there is no other $\tt A$-run break in $[17,20]$ with an $\LCP$  value larger than $\LCP[20]=2$. 
We therefore insert text position $n - SA[i'] + 1 = n - SA[19] + 1 = 12$ in $\SSS$. Notice that $i=20$ is also a $\tt G$-run break; repeating the above reasoning, one can verify that position $i'=20$ is indeed associated with the supermaximal extension $\tt AT\cdot G$ ending in text position $n-SA[20]+1=9$. 
  }\label{fig:run-example}
\end{figure}

We immediately obtain a simple quadratic algorithm computing a suffixient set of smallest cardinality: see Algorithm \ref{alg:quadratic}.
Note that, in Lemma \ref{lem:link LCP - supermaximal}, if $i''\in B_{i,c}$ and $\LCP[i''] = \LCP[i]$, then $box(i) = box(i'')$. 
Since there could be multiple $c$-run breaks $i'' \in B_{i,c}$ such that $\LCP[i''] = \LCP[i]$, if $\LCP[i] = \LCP[i''] = \max \LCP[B_{i,c}]$ then we have to break ties and insert only one of the corresponding text positions in the suffixient set. In Line \ref{line:quadratic.max_run_break} of Algorithm \ref{alg:quadratic} we choose the largest such $c$-run break: as a matter of fact, this defines a particular tie-breaking strategy $<_t$ between positions $[n]$ (as required by Definition \ref{def:SSS}).

\SetKwComment{Comment}{/* }{ */}
\RestyleAlgo{ruled}
\begin{algorithm}[h!t]
\caption{Quadratic algorithm computing a smallest suffixient set}\label{alg:quadratic}
    \LinesNumbered
    \SetKwInOut{Input}{input}
    \SetKwInOut{Output}{output}
    \Input{A text $T[1,n] \in \Sigma^n$}
    \Output{A suffixient set $\SSS$ of smallest cardinality for $T$.}
{$\BWT \gets \BWT(\rev{T})$;
$\LCP \gets \LCP(\rev{T})$;
$\SA \gets \SA(\rev{T})$\;}\label{line:quadratic.compute arrays}
$\SSS \gets \emptyset$\;
\For{$i = 2, \dots, n$}{
  \If{$\BWT[i-1] \neq \BWT[i]$\label{line:quadratic.check run break}}{
    \For{$i'\in\{i-1, i\}$}{
       $c \gets \BWT[i']$\ \Comment*[r]{$i$ is a $c$-run break}
       \If{$i = \max\{i''\ :\ i''\in B_{i,c} \wedge \LCP[i''] = \max \LCP[B_{i,c}]\}$\label{line:quadratic.max_run_break}}{ 
          $\SSS \gets \SSS \cup \{n-\SA[i']+1\}\;$
       }
    } 
  }
}
\Return $\SSS$\;
\end{algorithm}

Computing $\BWT$, $\LCP$, and $\SA$ in Line \ref{line:quadratic.compute arrays} takes $O(n)$ time. 
Additionally, for each $i = 2, \dots, n$, Algorithm \ref{alg:quadratic} computes the set $B_{i,c}$ in line \ref{line:quadratic.max_run_break} by scanning $\LCP[\dots,i-1, i, i+1, \dots]$ in order to identify $\LCP[box(i)]$ ($O(n)$ time for each such scan). It follows that the total running time is bounded by $O(n^2)$. Correctness follows immediately from Lemma \ref{lem:link LCP - supermaximal}, Definition \ref{def:SSS}, and Lemma \ref{lem:SSS}.

\subsection{A one-pass algorithm}\label{sec:one pass}

In this section, we speed up the simple quadratic algorithm provided in the previous section. We will refer to this version as "one-pass" since it only requires one scan of the $\BWT$, $\SA$, and $\LCP$ arrays for $\rev T$. 
The algorithm is summarized in Algorithms \ref{alg:one-pass} and \ref{alg:eval-cand}. See below for a description.

\SetKwComment{Comment}{/* }{ */}
\RestyleAlgo{ruled}

\begin{algorithm}[ht]
	\caption{One-pass algorithm building a smallest suffixient set}\label{alg:one-pass}
	\LinesNumbered
	\SetKwInOut{Input}{input}
	\SetKwInOut{Output}{output}
	\Input{A text $T[1, n]$ over a finite alphabet $\Sigma$.}
	\Output{A suffixient set $\SSS$ of smallest cardinality for $T$.}
	{$\mathcal{S} \gets \emptyset$}\ \Comment*[r]{initialize the empty suffixient set $\mathcal{S}$}
	{$\BWT\ \gets \textrm{BWT}(\rev{T})$;
		$\LCP\ \gets \textrm{LCP}(\rev{T})$;
		$\SA\ \gets \textrm{SA}(\rev{T})$\;} 
	$R[1,\sigma] \gets ((-1,0,false), \dots, (-1,0,false))$\ \Comment*[r]{$\LCP$ local maxima}
	$m \gets \infty$\Comment*[r]{$\LCP$ minimum inside current BWT run}
	
	\For{$i = 2, \dots, n$}{

            $m \gets \min\{m, \LCP[i]\}$\;
		
		\If{$\BWT[i] \neq \BWT[i-1]$}{
			$eval(\Sigma,m,R,\mathcal{S})$\;\label{alg:one-pass line eval}
			\For{$i'\in \{i-1,i\}$}{
				\If{$\LCP[i]>R[\BWT[i']].len$ \label{alg:one-pass line check}}{
					$R[\BWT[i']] \gets (\LCP[i], n-\SA[i']+1,true)$\;\label{alg:one-pass line activate}
				}
			}
			$m \gets \infty$\;
		}
		
	}
	
	$eval(\Sigma,-1,R,\mathcal{S})$\ \Comment*[r]{evaluate last active candidates}
	\Return $\SSS$\;
\end{algorithm}

\begin{algorithm}[ht]
\caption{Procedure $eval(C,l,R,\SSS)$}\label{alg:eval-cand}
    \LinesNumbered
    \SetKwInOut{Input}{input}
    \SetKwInOut{Output}{output}
    \Input{A set of characters $C$, a LCP value $l$, the candidate's list $R$ and the suffixient set $\mathcal{S}$.}
    \Output{The updated suffixient set $\mathcal{S}$.}

    \ForEach{$c \in C$}{
        \If{$l < R[c].len$}
        {
            \If{$R[c].active$}
            {
                $\mathcal{S} \gets \mathcal{S} \cup \{R[c].pos\}$\;
            }
            $R[c] \gets \{l,0,false\}$\;\label{alg:eval-cand len decr}
        }
    }

\end{algorithm}

Observe that Algorithm \ref{alg:quadratic} runs in quadratic time due to the fact that it needs to scan  $\LCP[box(i)]$ for each run break $i$, and those boxes overlap.
Intuitively, Algorithm \ref{alg:one-pass} avoids this inefficiency by storing information about local LCP maxima on $c$-run breaks (for each $c\in \Sigma$) in a  data structure $R$ (Line 3). Intuitively, this allows us to detect if a previous box $box(j)$, with $j<i$, ends in position $i$ by simply checking if $\LCP[j]$ drops below the current LCP maxima associated with each character $c$.

Assume we scanned the arrays $\BWT$, $\SA$, and $\LCP$ up to position $i$. $R$ is a map associating  each $c \in \Sigma$ with  information related with some $c$-run break $j \le i$ corresponding to a candidate 
supermaximal extension $\alpha\cdot c$. 
Each entry $R[c] = (len, pos, active)$ is composed of three values: $R[c].len = \LCP[j]$ (i.e., the length of the right-maximal string $\alpha$), $R[c].pos = text(j') = n - \SA[j'] +1$, where $j'\in \{j-1,j\}$ is such that $\BWT[j']=c$ (i.e.\ the position of the last character of the supermaximal extension $\alpha\cdot c$ in $T$), and a boolean flag $R[c].active$ which is set to $true$ if and only if $R[c].len$ is the maximum value in $\LCP[B_{j,c}\cap[i]]$
(breaking ties by smaller $j$: we keep the first encountered such maximum). 
Depending on the value of flag $active$, we will distinguish between {\em active} and {\em inactive candidates}.

\medskip
\begin{example}\label{ex:ex1}
    See Figure \ref{fig:run-example} and assume we scanned \BWT\ up to position $i=4$. The candidate associated with letter $\tt T$ is $R[{\tt T}] = (2,3,true)$, where $R[{\tt T}].len = \LCP[4]=2$ and $R[{\tt T}].pos=n-SA[3]+1=3$. This candidate is active since $\LCP[4]$ is the largest value in $\LCP[B_{4,{\tt T}}$ $\cap [4]] = \{2\}$. When processing position $i=6$, we find another $\tt T$-run break. Here we update $R[{\tt T}] = (\LCP[i],n-SA[i']+1,true)=(5,6,true)$ where $i'=i-1=5$ since $6 \in B_{4,{\tt T}} \cap [6]$, and $5=\LCP[6] > \LCP[4]=2$.
\end{example}
\medskip

Algorithm \ref{alg:one-pass} works as follows.
We scan the $\BWT$ left-to-right: $\BWT[2]$, $\dots$, $\BWT[i]$ (for $i=2, \dots, n$).
For $c\in \Sigma$, let $j<i$ be the $c$-run break 
whose information is stored in $R[c]$, i.e. 
such that $\LCP[j] = R[c].len$. 
If $\LCP[i] < R[c].len$, then $i \notin box(j)$. But then, if $R[c].active=true$, position $R[c].pos$ must belong to the output set so we have to insert it in $\SSS$. We also set $R[c].active = false$ in order to record that the maximum in $\LCP[B_{j,c}]$ has been found (if $R[c].active$ is already equal to $false$, $\SSS$ does not need to be updated since the maximum in $\LCP[B_{j,c}]$ had  already been found).
Observe that these operations, performed by the call to $eval(\Sigma, m, R, \SSS)$ (lines 8,17) (see below for the meaning of variable $m$), have to be repeated for each $c\in \Sigma$; however, since $i$ can possibly replace the candidate in $R[c]$ only if $i$ is a $c$-run break, we can avoid calling $eval(\cdot)$ inside equal-letter $\BWT$ runs as follows: we  compute the minimum $\LCP$ value $m$ inside $\LCP$ intervals corresponding to $\BWT$ equal-letter runs (line 6), and perform the check $m < R[c].len$ for each $c\in \Sigma$ by calling procedure $eval(\Sigma, m, R, \SSS)$ (costing time $O(\sigma)$) only when the current position $i$ is a run break. This ultimately reduces the algorithm's running time from $O(n\cdot \sigma)$ to $O(n + \bar r\cdot \sigma)$.

On the other hand, if $i$ is a $c$-run break and $\LCP[i] > R[c].len$, then $i \in B_{j,c} \cap [i]$ and $R[c].len$ is not a local maximum 
in $\LCP[B_{j,c}]$ (line 10). In this case, we have to replace the candidate stored in $R[c]$ with the information associated with the new candidate $i$: 
letting $i'\in \{i-1,i\}$ be such that $\BWT[i'] = c$, we replace 
$R[c] \gets (\LCP[i], n-\SA[i']+1, true)$ (line 11).

\medskip
\begin{example}
    Continuing Example \ref{ex:ex1}. When processing the $\tt T$-run break $i=7$, we see that $\LCP[7] < R[{\tt T}].len$. Thus, we reached the end of $box(6) = [6,6]$, where max $\LCP[B_{6,{\tt T}}] = R[{\tt T}].len = 5$, and insert $R[{\tt T}].pos$ 
    in $\mathcal{S}$. Next, we update $R[{\tt T}] = (4,0,false)$.
    On $i=8$, we again update $R[{\tt T}] = (1,0,false)$. Finally, on $i=9$, we get $\LCP[9] > R[{\tt T}].len$ since $box(6)$ and $box(9)$ are now disjoint. We update $R[{\tt T}] = (\LCP[9],11,true)$ to the active state.
\end{example}
\medskip

We obtain:

\medskip
\begin{restatable}{lemma}{OnepassProof}\label{lem:one-pass-analysis}
    Given a text $T[1,n]$ over alphabet of size $\sigma$, Algorithm \ref{alg:one-pass} computes a smallest suffixient set $\mathcal{S}$ in $O(n + \bar r\cdot\sigma)$ time and $O(n)$ words of space. 
\end{restatable}
\begin{proof}
    Algorithm \ref{alg:one-pass} scans the \BWT, \LCP, and \SA\ exactly once and, for each run break, checks if any of the $\sigma$ candidates is a supermaximal extension using Algorithm \ref{alg:eval-cand}. Since Algorithm \ref{alg:eval-cand} performs a constant number of operations for each $c \in \Sigma$, the whole algorithm runs in $O(n + \bar r\cdot \sigma)$ time. In addition, the only supplementary data structure we need is $R$, which takes $O(\sigma)$ words of space; thus if $\sigma < n$, altogether, we take $O(n)$ words of space.

    We prove the correctness of Algorithm \ref{alg:one-pass} by showing it computes the same output of Algorithm \ref{alg:quadratic}. In particular, given $\mathcal{S'}$, the output of Algorithm \ref{alg:one-pass}, and $\mathcal{S}$, the output of Algorithm \ref{alg:quadratic}, we show that given any value $s = n - \SA[i] + 1$, where $i \in [n]$, \emph{(1)} $s \in \mathcal{S} \implies  s \in \mathcal{S'}$ and \emph{(2)} $s \notin \mathcal{S} \implies  s \notin \mathcal{S'}$.

    \emph{(1)} let $(n - \SA[i'] + 1) \in \mathcal{S}$, where $i$ is a $c$-run break, such that $i' \in \{i-1,i\}$ and $\BWT[i'] = c$. By Lemma \ref{lem:link LCP - supermaximal}, it follows that $\max \LCP[B_{i,c}] = \LCP[i]$. Assume we scanned the \BWT\ up to position $i$, there are three cases: $i'$ is the position of the first occurrence of $c$ in the \BWT, or there exists another $c$-run break at position $j < i$, such that either $i \in B_{j,c}$ and $\LCP[i] > \LCP[j]$, or $box(i) \cap box(j) = \emptyset$. 
    For all these three cases, $\LCP[i] > R[c].len$; thus, $i$ is set as the new active $c$ candidate in $R$ (lines 10-12). Now, if $i'$ is the position of the last occurrence of $c$ in the \BWT, $R[c].pos$ is inserted in $\mathcal{S}$ at the end of the algorithm (line 17); otherwise, let $box(i) = [b,e]$, where $b \leq i \leq e$, we have $\forall s \in \LCP[b,e] , s \geq \LCP[i]$ and no $j' \in [i+1,b]$ is a $c$-run break such that $\LCP[j'] > \LCP[i]$. Thus, $R[c]$ is not updated until we scan $\LCP[e+1]$. At this point, $R[c].len > \LCP[e+1]$ and since $R[c]$ is an active candidate; thus, Algorithm \ref{alg:one-pass} insert $R[c].pos = (n - \SA[i] + 1)$ in $\mathcal{S'}$ (Algorithm \ref{alg:eval-cand}, line 4).    

    \emph{(2)} now let $(n - \SA[i'] + 1) \notin \mathcal{S}$, it follows by Lemma \ref{lem:link LCP - supermaximal}, that $\exists j \in B_{i,c}$ such that $\LCP[j] > \LCP[i]$ (for simplicity we consider the smallest $j$). Assume we scanned the \BWT\ up to position $i$; we need to consider two cases: \emph{(i)} $j > i$ and \emph{(ii)} $j < i$. \emph{(i)} If $j > i$, let $R[c].len = \LCP[i]$, then $\forall l \in \LCP[i+1,j], l \geq R[c].len$; thus, $(n - \SA[i'] + 1)$ is not inserted in $\mathcal{S}$ (lines 2-6, Algorithm \ref{alg:eval-cand}) until position $j$. Here, we get $\LCP[j] > R[c].len$ and, update $R[c].pos = (n - \SA[j'] + 1)$ (lines 10-12), where $j'\in\{j-1,j\}$ and $\BWT[j']=c$. Due to this, $(n - \SA[i'] + 1)$ is dropped and not inserted in $\mathcal{S'}$. \emph{(ii)} If $j < i$, then $\exists l \in \LCP[j+1,i], l < \LCP[j]$; thus when we read $l$, $R[c]$ is updated to the inactive state (line 6, Algorithm \ref{alg:eval-cand}). However, since $\forall s \in B_{i,c} \cap [j+1,i], \LCP[s] < \LCP[j]$ it means that $R[c]$ is never updated to the active state for any position in $[j+1,i]$ (we skip lines 10-12), so also in this case $(n - \SA[i'] + 1)$ is not inserted in $\mathcal{S'}$.
\end{proof}

Observe that in Algorithm \ref{alg:eval-cand} (procedure $eval(\cdot)$), only entries such that $l < R[c].len$ are possibly modified. 
This suggests that $R$ could be sorted in order to speed up operations. Indeed, it turns out that $R$ can be replaced with a data structure based on balanced search trees, such that operation  $eval(\Sigma,m,R,\SSS)$ at Line \ref{alg:one-pass line eval} of Algorithm \ref{alg:one-pass} costs $O(\log\sigma)$ amortized time. With this modification, Algorithm \ref{alg:one-pass} runs in $O(n + \bar r \log\sigma)$ time.

\medskip
\begin{lemma}
    \label{lemma:eval amortized log sigma}
    Algorithm \ref{alg:one-pass} can be implemented with $O(n + \bar r\log\sigma)$ running time and $O(n)$ words of space where $n$ is the text length, $\bar{r}$ is the number of runs in $\BWT(\rev{T})$, and $\sigma$ is the alphabet size. 
\end{lemma}
\begin{proof}
    We will prove that $eval(\cdot)$ can be implemented with $O(\log\sigma)$ amortized running time.
    We maintain a balanced binary search tree (BST) in which each node is keyed by an integer $l$ such that there exists $c\in\Sigma$ such that $l=R[c].len$. Note that the BST has at most $\sigma$ nodes, thus the height is $O(\log\sigma)$. Each node stores such corresponding alphabet symbols $c$'s, maintaining them using two doubly linked lists, one for active symbols (i.e., those with $R[c].active=true$), and the other for inactive symbols (i.e., those with $R[c].active=false$).
    For each symbol $c\in\Sigma$, we store $R[c].pos$ and $R[c].active$, as well as the pointer to the BST node in which $c$ is stored so that $R[c].len$ can be retrieved by accessing $c$, and the pointer to the corresponding node of the linked list so that we can move symbols between different linked lists in $O(1)$ time.

    We implement 
    Line~\ref{alg:one-pass line activate} of Algorithm~\ref{alg:one-pass} 
    as follows. 
    We move the linked-list node corresponding to symbol $BWT[i']$ into the linked list for active symbols of the BST node associated with $l$ (after creating a new node if it does not exist) and update $R[c].pos$ and $R[c].active$ for $c=BWT[i']$ accordingly. This takes $O(\log\sigma)$ time.
    
    Our optimization of Algorithm~\ref{alg:eval-cand} then works as follows. We retrieve all nodes associated with lengths greater than $l$, which can be performed in $O(\log\sigma + k_{ret})$ time where $k_{ret}$ is the number of retrieved nodes. We iterate all active symbols $c\in\Sigma$ stored in the retrieved nodes to add $R[c].pos$ to $\mathcal{S}$ in $O(k_{act})$ time where $k_{act}$ is the total number of active symbols to report. Then we concatenate all the linked lists in the retrieved nodes in $O(k_{ret})$ time, and append it to the linked list for inactive symbols in the node with key $l$ (we create a new node if it does not exist), which takes $O(\log\sigma)$ time. Then we delete all the retrieved BST nodes (since they no longer contain any symbol) in $O(k_{ret}\log \sigma)$ time. 
    As a consequence, Algorithm~\ref{alg:eval-cand} runs in $O(k_{act}+(1+k_{ret})\log\sigma)$ time. Observe that at each $c$-run break, at most two symbols can become active and only $O(1)$ new BST nodes can be created. Therefore, the sum of $k_{ret}+k_{act}$ over all executions of $eval(\cdot)$ is $O(\bar{r})$. Therefore, the amortized running time is $O(\log\sigma)$.
\end{proof}
\medskip

\subsubsection*{One-pass algorithm in compressed space}

One important feature of our algorithm is that we only need one scan of the \BWT, \SA, and \LCP\ to compute $\mathcal{S}$.
This means that our algorithm also works in the streaming scenario where these arrays are provided one element at a time, from first to last. This feature can be exploited to run the algorithm in compressed working space using Prefix-free Parsing \cite{BoucherGKLMM19}. This optimization does not provide strong theoretical bounds, but in practice (see Section \ref{sec:experiments}) it drastically reduces the amount of time and working space needed to compute the Suffixient Array. 

Prefix-free parsing (PFP) is a technique introduced by Boucher et al.\ \cite{BoucherGKLMM19} to ease the computational burden of computing the \BWT\ of large and repetitive texts. Briefly, with one linear-time scan of the text $T$, PFP divides $T$ into overlapping segments, called {\em phrases}, of variable length, which are then used to construct what is referred to as the {\em dictionary} $D$ (i.e. the set of distinct phrases) and {\em parse} $P$ of the text (i.e. the text $T$ encoded as a sequence of phrases, represented as indexes in $D$). Then, with a separate linear-time algorithm, the \BWT\ of $T$ is directly computed from $D$ and $P$; thus using space proportional to the combined size $|D|+|P|$ of the two data structures. 
Very repetitive texts will tend to generate very small $D$, since text repetitions translate to repeated phrases. The size of $P$ (number of phrases in which $T$ is parsed), is instead controlled by a user-defined parameter and is a small fraction of $T$'s length $n$.
In \cite{KuhnleMBGLM20,RossiOLGB22}, it was shown how to modify PFP in order to compute also the \SA\ and the \LCP\ arrays. This version streams the three arrays \BWT, \LCP, and \SA\ of $T$ in $O(|D| + |P|)$ compressed space, from their first to last entry and in parallel (i.e. the triples $(\BWT[i], \LCP[i], \SA[i])$ are output for $i=1, 2, \dots, n$); this is sufficient for running Algorithm \ref{alg:one-pass} in compressed space, without affecting its running time. 
The only change consists in reading the text $T$ backwards, in order to compute $\BWT(\rev T)$, $\LCP(\rev T)$, and $\SA(\rev T)$. If the input $T$ resides on disk, this can be achieved very easily with a backward scan of $T$ (without increasing the number of I/Os - page swaps - with respect to reading $T$ forward). If, instead, the text is read from a forward stream, then PFP yields $\BWT(T)$, $\LCP(T)$, and $\SA(T)$ and our algorithm computes a suffixient set for $\rev T$. Our index still works, except that we have to reverse the query pattern before searching for it in the index. It is worth mentioning that it is actually also possible to modify PFP so that it streams $\BWT(\rev T)$, $\LCP(\rev T)$, and $\SA(\rev T)$ while reading a forward stream for $T$. However, in view of the simpler solutions described above, we do not enter into such details. 

In Section \ref{sec:experiments} we will show that our PFP-based optimization makes it possible to compute the Suffixient Array efficiently on massive repetitive datasets.

\subsection{A linear-time algorithm using LF-mapping}\label{sec:linear time}

In this section, we further speed up the one-pass algorithm provided in the previous section and achieve linear time. Our new algorithm is summarized in Algorithm \ref{alg:linear-time-algo}.

\SetKwComment{Comment}{/* }{ */}
\RestyleAlgo{ruled}
\begin{algorithm}[ht]
	\caption{First linear-time algorithm building a smallest suffixient set}\label{alg:linear-time-algo}
	\LinesNumbered
	\SetKwInOut{Input}{input}
	\SetKwInOut{Output}{output}
	\Input{A text $T[1, n]$ over a finite alphabet $\Sigma$.}
	\Output{A smallest suffixient set for $T$.}
	{$\mathcal{S} \gets \emptyset$}\ \Comment*[r]{initialize the empty suffixient set $\mathcal{S}$}
	{$\BWT\ \gets \textrm{BWT}(\rev{T})$;
		$\LCP\ \gets \textrm{LCP}(\rev{T})$;
		$\SA\ \gets \textrm{SA}(\rev{T})$\;} 
	$R[1,\sigma] \gets ((-1,0,false), \dots, (-1,0,false))$\ \Comment*[r]{$\LCP$ local maxima}
	$\LF[1,\sigma] \gets (0,\dots,0)$\Comment*[r]{LF mapping}
	$m \gets \infty$\Comment*[r]{$\LCP$ minimum inside current BWT run}
	
	\textbf{for all} $c=2,\dots, \sigma$: $\LF[c] \gets \LF[c-1] + occ(T,c-1)$\;

	$\LF[\BWT[1]] \gets \LF[\BWT[1]] +1 $\;

	\For{$i = 2, \dots, n$}{

		$\LF[\BWT[i]] \gets \LF[\BWT[i]] +1 $\;
            $m \gets \min\{m, \LCP[i]\}$\;

		\If{$\BWT[i] \neq \BWT[i-1]$}{
			
			\For{$i'\in \{i-1,i\}$}{

				\textbf{if $i'=i-1$ then} $eval(\{\BWT[i']\},m,R,\mathcal{S})$\;
                \textbf{else} \If{$R[\BWT[i']].len \neq -1$}{ $eval(\{\BWT[i']\},\LCP[\LF[\BWT[i']]]-1,R,\mathcal{S})$\;}

				\vspace{5pt}

				\If{$\LCP[i]>R[\BWT[i']].len$}{
					$R[\BWT[i']] \gets (\LCP[i], n-\SA[i']+1,true)$\;
				}
	
		}
	
			$m \gets \infty$\;
		}
		
	}

	$eval(\Sigma,-1,R,\mathcal{S})$\ \Comment*[r]{evaluate last active candidates}
	\Return $\SSS$\;
\end{algorithm}

As previously discussed, in Algorithm \ref{alg:one-pass}, 
for the one-character extension of every $\BWT$ equal-letter $\BWT[i^*,i] = cc\dots ccx$ (with $x\neq c$) we run procedure $\emph{eval}$ (Algorithm \ref{alg:eval-cand}) to check if $\min\LCP[i^*+1,i]$ drops below the current \LCP\ maxima associated to $y$-run breaks, for all $y\in \Sigma$. In the end, this step charges an additional $O(\bar r \sigma)$ term (which we mentioned can be reduced to $O(\bar r \log \sigma)$ by using opportune data structures). Intuitively, Algorithm \ref{alg:linear-time-algo} avoids this cost by calling the \emph{eval} procedure only on the two candidates $R[c']$, where $c' \in \BWT[i-1,i]$. We achieve this by introducing a new vector $\LF[1,\sigma]$ updated on-the-fly, implementing the \emph{LF-mapping} property of the Burrows-Wheeler transform. More formally: assume we have scanned the \BWT\ up to the $c$-run break $i$, and let $i' \in \{i-1,i\}$ be such that $\BWT[i']=c$. Then, $\LF$ is such that $\LF[\BWT[i']] = \LF[c] = j$, with $\SA[j] = \SA[i'] + 1$. 
In other words, $\BWT[j]$ is the character preceding $\BWT[i']$ in $\rev{T}$.
Our linear-time algorithm is based on the following idea.  
Assume for simplicity that $i$ is not the first $c$-run break. Let $i^* < i$ be the largest integer such that $i^*$ is a $c$-run break as well (i.e. $i^*$ is the $c$-run break immediately preceding $i$).
As in Algorithm \ref{alg:one-pass}, our new Algorithm \ref{alg:linear-time-algo} stores in entry $R[c].len$ the value $R[c].len = \LCP[i^*]$. 
By Lemma \ref{lem:link LCP - supermaximal}, we have to discover if $i\in B_{i^*,c}$; if this is the case, then we compare $\LCP[i]$ and $\LCP[i^*]$ and decide if $\LCP[i]$ is the new maximum in $B_{i^*,c} \cap [i]$ (i.e. $\LCP[i] > \LCP[i^*]$) or if the local maximum among $c$-run breaks in $B_{i^*,c} \cap [i]$ had already been found (i.e. $\LCP[i] \le \LCP[i^*]$).
If, on the other hand, $i\notin B_{i^*,c}$, then 
we insert in $\SSS$ the local maximum $R[c].pos$ relative to $B_{i^*,c}$ if and only if $R[c].active = true$ (if $R[c].active = false$, then the suffixient position corresponding to the local maximum of $B_{i^*,c}$ had already been stored in $\SSS$).

In order to discover if $i\in B_{i^*,c}$, we distinguish two cases. 

(i) If $\BWT[i-1]=c$ then, since $i^*$ is the previous $c$-run break, it must be the case that $\BWT[i^*-1, i^*, \dots , i-1, i] = y c c \dots c c x$, for some $x\neq c$ and $y\neq c$. It follows that $i\in B_{i^*,c}$ if and only if $\min\LCP[i^*,i] \geq R[c].len = \LCP[i^*]$. Algorithm \ref{alg:linear-time-algo} computes $\min\LCP[i^*,i]$ analogously as Algorithm \ref{alg:one-pass}, i.e. by updating a variable $m$ storing the minimum $\LCP$ inside intervals corresponding to $\BWT$ equal-letter runs. 

(ii) If, on the other hand, $\BWT[i-1]=x \neq c$, then  since $i^*$ is the previous $c$-run break, it must be that $\BWT[i^*-1,i^*,\ldots,i-1,i] = cy...xc$, for some $x \neq c$ and $y \neq c$ (and there are no other occurrences of $c$ between $\BWT[i^*-1]$ 
and $\BWT[i]$). 
As in the previous case, the goal is to compute $\min\LCP[i^*,i]$. We achieve this by using the $\LF$ array. Due to the way the array $\LF$ is defined and constructed, we know that $\min \LCP[i^*,i] = \LCP[\LF[c]]-1$. As a result, we have that $i\in B_{i^*,c}$ if and only if $\LCP[\LF[c]] - 1 \geq R[c].len = \LCP[i^*]$. 

From the above intuition, we obtain:

\medskip
\begin{lemma}\label{lem:linear-time-correctness}
    Given a text $T[1,n]$ over $\Sigma$, Algorithm \ref{alg:linear-time-algo} computes a smallest suffixient set $\mathcal{S}$ in $O(n)$ time and $O(n)$ words of space.
\end{lemma}
\begin{proof}
    Algorithm \ref{alg:linear-time-algo} scans the \BWT~exactly once and, for each position of it, performs $O(1)$ accesses to the $\SA$ and the $\LCP$ and calls Algorithm \ref{alg:eval-cand} at most once for each run break. Since Algorithm \ref{alg:eval-cand} performs a constant number of operations and the $LF$-mapping array can be computed in $O(n)$ time, the whole algorithm runs in $O(n)$ time. In addition, the only supplementary data structure we need is $R$, which consumes $O(\sigma)$ words of space; thus, altogether, we take $O(n + \sigma)$ words of space. Since we assumed $\sigma \leq n$, the space consumption of Algorithm \ref{alg:linear-time-algo} is $O(n)$ words.
    
    We prove the correctness of Algorithm \ref{alg:linear-time-algo} by showing it computes the same output as the output computed by Algorithm \ref{alg:quadratic}. 
    In particular, given $\mathcal{S'}$, the output of Algorithm \ref{alg:linear-time-algo}, and $\mathcal{S}$, the output of Algorithm \ref{alg:quadratic},
    we show that for any value $s = n - \SA[i] + 1$, where $i \in [n]$, \emph{(1)} $s \in \mathcal{S} \implies  s \in \mathcal{S'}$ and \emph{(2)} $s \notin \mathcal{S} \implies  s \notin \mathcal{S'}$.
    
    \emph{(1)} let $(n - \SA[i'] + 1) \in \mathcal{S}$, where $i$ is a $c$-run break, such that $i' \in \{i-1,i\}$ and $\BWT[i'] = c$. By Lemma \ref{lem:link LCP - supermaximal}, we have $\max \LCP[B_{i,c}] = \LCP[i]$. Assume we have scanned the \BWT~up to position $i$. There are three cases: $i'$ is the position of the first occurrence of $c$ in the \BWT, or there exists another $c$-run break at position $j < i$, such that either $i \in B_{j,c}$ or $box(i) \cap box(j) = \emptyset$. In all three cases, since $(n - \SA[i'] + 1) \in \SSS$, we have $\LCP[i] > R[c].len$; thus, $i$ is set as the new active $c$ candidate in $R$ (lines 17-19). If $i'$ is the position of the last occurrence of $\BWT[i']$, since $R[c].active = true$, after traversing the whole \BWT, Algorithm \ref{alg:linear-time-algo} inserts $R[c].pos = (n - \SA[i] + 1)$ in $\mathcal{S'}$ (line 4, Algorithm \ref{alg:eval-cand}). Otherwise, let $i < k'$ be the position of the next occurrence of $c$ in the \BWT, such that $k' \in \{k - 1, k\}$. Due to $\max \LCP[B_{i,c}] = \LCP[i]$, it must hold $\LCP[k] < \LCP[i]$ or $k \not\in box(i)$. 
    Here either, $\BWT[i-1,k'] = xcc...ccy$ where $x \neq c$ and $y \neq c$; thus, $\exists l \in \LCP[i+1,k']$ such that $l < R[c].len$, or $\LCP[\LF[c]] - 1 < R[c].len$. 
    Again, since $R[c].active = true$, then Algorithm \ref{alg:linear-time-algo} insert $R[c].pos = (n - \SA[i] + 1)$ in $\mathcal{S'}$ (line 4, Algorithm \ref{alg:eval-cand}).
    
    \emph{(2)} now let $(n - \SA[i'] + 1) \notin \mathcal{S}$, this means that $\exists j \in B_{i,c}$ such that $\LCP[j] > \LCP[i]$ and $j$ is a $c$-run break. We need to consider two cases: \emph{(i)} $j > i$ and \emph{(ii)} $j < i$. \emph{(i)} If $j > i$, $\forall l \in \LCP[i + 1, j]$, $ l \geq \LCP[i]$, which implies $\LCP[\LF[c]] - 1 \geq \LCP[i]$ 
    for all $c$-run breaks in $[i+1,j]$; thus, we never update $R[c]$ to the inactive state (lines 2-6, Algorithm \ref{alg:eval-cand}) until we scan position $j$. Here, we get $\LCP[j] > R[c].len$ and update $R[c].pos = (n - \SA[j'] + 1)$ (lines 17-19), where $j' \in \{j - 1, j\}$ and $\BWT[j'] = c$. Due to this, $(n - \SA[i'] + 1)$ is dropped and not inserted in $\mathcal{S'}$. \emph{(ii)} If $j < i$, then 
    $\exists l \in [j+1,i]$ such that $\LCP[\LF[c]]-1 < R[c].len$ or $\LCP[l] < \LCP[i]$;
    thus, when we read $l$, $R[c]$ is updated to the inactive state (line 6, Algorithm \ref{alg:eval-cand}). However, since $\forall s \in B_{i,c} \cap [j + 1, i]$, $ \LCP[s] < \LCP[j]$ it means that $R[c]$ is never updated to the active state for any position in $[j + 1, i]$ (we skip lines 17-19), so also in this case $(n - \SA[i'] + 1)$ is not inserted in $\mathcal{S'}$.
\end{proof}

In Figure \ref{fig:linear-alg-example} we show an example of how Algorithm \ref{alg:linear-time-algo} works.

\begin{figure}[!htb]
\centering
\includegraphics[scale=0.8]{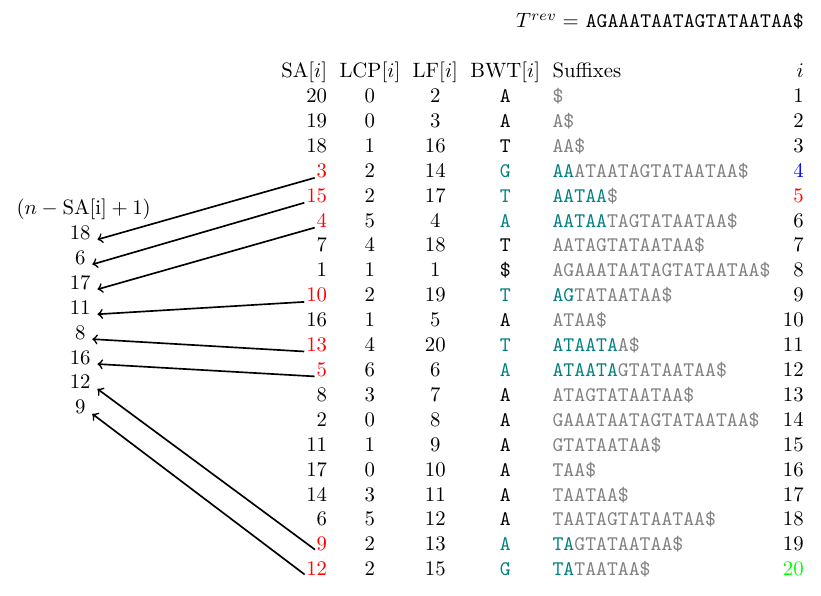}
\caption{
The figure shows the data used by Algorithm \ref{alg:linear-time-algo} to construct a smallest suffixient set $\SSS$ for a text $T[1,n]$ with $n = 20$. In addition to the data shown in Figure \ref{fig:run-example}, here we include
the $\LF$ array. Each position $i$ of this array is the result obtained when updating $\LF[\BWT[i]] = \LF[\BWT[i]] + 1$ in line 9 of Algorithm \ref{alg:linear-time-algo}.
For brevity, we show how  Algorithm \ref{alg:linear-time-algo} works only on the $\tt G$-run breaks.
Consider the first $\tt G$-run break at position $i = 4$ (highlighted in blue in column $i$). At this stage we have $\LF[G] = occ(T, {\tt \$}) + occ(T, {\tt A}) = 13$ and $R[{\tt G}] = (-1, 0, false)$, so we update $\LF[{\tt G}] = \LF[{\tt G}] + 1 = 14$ (line  9). Since $\BWT[i] ={\tt G}$ (then $i' = i$) and since $R[{\tt G}].len = -1$, then we do not call $\tt{eval}$ on either line 13 or line 14, so $R[{\tt G}].pos$ is not added to $\SSS$. Now, we have $\LCP[4] = 2 > R[{\tt G}].len = -1$, then we set $R[{\tt G}] = (\LCP[4], n - \SA[4] + 1, true) = (2, 18, true)$ (line 18) and $m = \infty$ (line 21). For the next ${\tt G}$-run break at position $i = 5$ (highlighted in red in column $i$), since we set $m = \min(\LCP[5], \infty) = 2$ and since $\LCP[\LF[{\tt G}]] - 1 = -1 < R[{\tt G}].len = 2$, then we do not add $R[G].pos$ to $\SSS$ in either line 13 nor line 14. Next, since $\LCP[5] = 2 = R[{\tt G}].len$, we do not update $R[{\tt G}]$ on line 18, and we finish this iteration. The next ${\tt G}$-run break occurs at position $i = 20$ (highlighted in green in column $i$). We set $LF[{\tt G}] = LF[{\tt G}] + 1 = 15$ and $m = \LCP[20] = 2$. Since $i' = i$, we do not call {\tt eval} on line 13, but since $\LCP[\LF[{\tt G}]] - 1 = 0 < R[{\tt G}].len = 2$, we add $R[{\tt G}].pos = 18$ to $\SSS$ and we set $R[{\tt G}] = (0, 0, false)$ on the {\tt eval} calling on line 15. Next, we have $R[{\tt G}].len = 0 < \LCP[20] = 2$, so we update $R[{\tt G}] = (2, 9, true)$. Finally, since $R[{\tt G}].active = true$, we add $R[{\tt G}].pos = 9$ on the final {\tt eval} calling on line 24.
}\label{fig:linear-alg-example}
\end{figure}

\subsection{Another linear-time algorithm via precomputing boxes}\label{sec: linear time 2}

In this section we propose a second linear algorithm
which in practice is faster than Algorithm~\ref{alg:linear-time-algo}, at the cost of increasing by an additive term $O(n)$ the space consumption. In Section \ref{sec:experiments} we show our experimental results comparing both algorithms.

For a given $c$-run break $i_f$,
let $i_f' \in \{i_f - 1, i_f\}$ be such that $\BWT[i_f'] = c$. Algorithms \ref{alg:one-pass} and \ref{alg:linear-time-algo} will add to the suffixient set the position $n - \SA[i_f'] + 1$ if and only if 

\begin{equation}i_f = \min\{i  \ :\ i\in B_{i_f, c} \wedge \LCP[i] = \max \LCP[B_{i_f, c}]\}.\label{eq:leftmost-maximum}\end{equation} 

In the following, we expose a method to evaluate Eq.~\eqref{eq:leftmost-maximum} in $O(1)$ time for each $i \in [n]$, which will lead to a linear-time algorithm.
As a first step, we define \textit{first c-maximum} positions (Definition~\ref{def:first-candidate}). Later, we prove that those positions are exactly the positions satisfying equation~\ref{eq:leftmost-maximum} (Proposition~\ref{prop:first-maximum}). Finally, we show an algorithm computing a suffixient set of smallest cardinality by finding \textit{first c-maximum} positions (Algorithm~\ref{alg:fm-algo}).

In the rest of this section, $k$ denotes the number of $c$-run breaks. Let $1 \leq i_1 < i_2 < \dots < i_k \leq n$ be the positions of those run breaks, and let $box(i_h) = [l_h, r_h]$, for each $h \in [k]$. 

\medskip
\begin{definition}[\textit{first c-maximum}]\label{def:first-candidate}
Let $i_f$ the position of a $c$-run break. We say $i_f$ is \textit{first c-candidate} if one of the following conditions hold:
\begin{itemize}
\item{$f = 1$, or}
\item{$i_{f - 1} \leq l_f - 1$}
\end{itemize}

Moreover, we say that $i_f$ is \textit{first c-maximum} if it is \textit{first c-candidate} and one of the following conditions hold:
\begin{itemize}
\item{$f = k$, or}
\item{for each $f < j \leq k$, if $i_f \leq l_{j}  - 1$ then $r_f + 1 < i_{j}$.}
\end{itemize}
\end{definition}
\medskip

Intuitively, \textit{first c-candidate} positions are those such that its $\LCP$ value is smaller than the $\LCP$ of the previous $c$-run break. On the other hand, \textit{first c-maximum} positions are \textit{first c-candidate} positions whose box only intersects $c$-run break positions not having larger $\LCP$ values.
The following proposition establishes the equivalency between finding $c$-run breaks which are the leftmost local maxima within a box and finding \textit{first c-maximum} positions.

\medskip
\begin{proposition}\label{prop:first-maximum}
Let $i_f$ be a $c$-run break. We have $$i_f = \min\{i\ :\ i\in B_{i_f, c} \wedge \LCP[i] = \max \LCP[B_{i_f, c}]\}$$ if and only if $i_f$ is \textit{first c-maximum}.
\end{proposition}
\begin{proof}
Let $B_{i_f, c} = \{i_{f - a}, \dots, i_f, \dots, i_{f + b}\}$. By definition it holds for every $j\in[f-a,f+b]$ that $\LCP[i_j]\ge \LCP[i_f]$, and it also holds that 
\begin{equation}
    i_{f-a-1}<l_f\le i_{f-a}<\cdots<i_f<\cdots<i_{f+b}\le r_f < i_{f+b+1}
    \label{eq: fm boxes}
\end{equation} where $[l_f,r_f]=box(i_f)$ and we define $i_0:=0$ and $i_{k+1}:=n+1$ for brevity. 

\medskip
\noindent$(\Rightarrow)$ Suppose $$i_f = \min\{i\ :\ i\in B_{i_f, c} \wedge \LCP[i] = \max \LCP[B_{i_f, c}]\}.$$ 
Then by the maximality of $\LCP[i_f]$, it holds that $\LCP[i_f]=\LCP[i']$ for every $i'\in B_{i_f,c}$.
This implies that $a = 0$ (due to the minimality of $i_f$) and $box(i_f) = \dots = box(i_{f + b})=[l_{f},r_{f}]$. Since it holds $i_{f-1}\le l_f-1< l_f\le i_{f}$ by Eq.\eqref{eq: fm boxes}, $i_{f}$ is \textit{first c-candidate}.
To prove that $i_f$ is also \textit{first c-maximum}, we need to prove that it holds for every $f<j\le k$ if $i_f\le l_{j}-1$ then $r_j+1<i_j$ where $[l_j,r_j]=box(i_j)$. For $f<j\le f+b$, observe that $l_{j} - 1= l_f-1 < i_f$, and the implication is true because the premise is false. Now consider $f+b<j\le k$. Then it holds that $r_f<r_f+1\le i_j$ by Eq.\eqref{eq: fm boxes}. We have two cases: (i) $r_f+1 < i_j$ and (ii) $r_f+1=i_j$. For the former case, we are done. For the latter case, by definition of $box(i_f)$, it holds that $\LCP[h]\ge \LCP[i_f] > \LCP[i_{j}]$ for every $h\in[l_f,r_f]$, which implies that $l_{j}\le l_{f}\le i_f$. Since $l_{j}-1<l_j\le i_f$, the implication is true due to the false premise. Therefore, $i_f$ is \textit{first c-maximum}.

\noindent$(\Leftarrow)$ We prove by contrapositive. Suppose $$i_f \neq i_m = \min\{i\ :\ i\in B_{i_f, c} \wedge \LCP[i] = \max \LCP[B_{i_f, c}]\}.$$ 
If $f < m$, then $b \ge 1$ and, since $i_m \in B_{i_f, c}$,  
$\LCP[i_f] < \LCP[i_{m}]$, we have $i_f \le l_{m}-1 < l_{m}$. However, it holds that $i_{f}<i_{m} \le i_{f+b} \le r_{f} < r_{f} + 1$ by Eq.\eqref{eq: fm boxes}.
Since it holds that (i) $f<m\le k$ and (ii) $i_f < l_{m}-1$ and (iii) $r_{f} + 1>i_m$, it follows that $i_f$ cannot be \textit{first c-maximum}. On the other hand, if $m < f$, then $a \ge  1$. From $l_f\le i_{f - a}$ in Eq.\eqref{eq: fm boxes}, it follows that $l_f-1 < l_f\le i_{f-a}\le i_{f-1}$. Therefore, $i_f$ cannot be \textit{first c-candidate}.
\end{proof}
\medskip

As shown in the proof of Proposition~\ref{prop:first-maximum}, if $i_f$ is \textit{first c-candidate} and there exists a position $i_f < i_e$ such that $l_{f + 1} - 1 < i_f, \dots, l_{e - 1} - 1 < i_f$, $i_f \leq l_e - 1$ and $r_f + 1 < i_e$, then for any other position $i_e < i_g$ such that $i_f \leq l_g - 1$ (if any exist) it holds $r_f + 1 < i_g$. Indeed, $i_e$ is also \textit{first c-candidate}, which means that, to find \textit{first c-maximum} positions, we just need to sequentially scan the $\BWT$ looking for \textit{first c-candidate} positions and compare their boxes' boundaries. Consider the two arrays defined based on the LCP array as the following.
\medskip
\begin{definition}[PSV/NSV arrays]
For a given length-$n$ integer array $A[1,n]$, the \textit{previous smaller value} array of $A$ is an integer array of length $n$ defined as $\PSV(A)[i] = \max(\{j \,|\, j < i, A[j] < A[i]\}\cup \{0\})$ for all $i\in [n]$. In a similar way, the \textit{next smaller value} array of $A$ is defined as $\NSV(A)[i] = \min(\{j \,|\, j > i,A[j] < A[i]\}\cup \{n+1\})$ for all $i\in [n]$.
\end{definition}
\medskip
Interestingly, if we have access to the $\PSV(\LCP)$ and $\NSV(\LCP)$ arrays, we can compute $box(i) = [l_i, r_i] = [\PSV(\LCP)[i] + 1, \NSV(\LCP)[i] - 1]$ in $O(1)$ time. Fortunately, these arrays can be computed in $O(n)$ time \cite{FMNtcs09}. For simplicity, in the rest of this section $\PSV$, $\NSV$ denote  $\PSV(\LCP)$, and $\NSV(\LCP)$, respectively.

Following the above ideas, we show in Algorithm~\ref{alg:fm-algo} a procedure adding to the suffixient set \textit{first c-maximum} positions.
In the same spirit as Algorithms~\ref{alg:one-pass} and \ref{alg:linear-time-algo}, $R[c]$ keeps the information of the last $c$-run break found up to now which is \textit{first c-candidate}. The algorithm sequentially scans the $\BWT$ looking for run breaks. For each $c$-run break at position $i$, if $i$ is \textit{first c-candidate}, $R[c]$ is updated as $(sa\_pos \gets i, text\_pos \gets n - \SA[i'] + 1, active \gets true, nsv \gets \NSV[i])$. Whenever $R[c]$ is being updated, we check if the last \textit{first c-candidate} is \textit{first c-maximum}, and report it if so. We conclude this section with the following.

\begin{algorithm}[t]
\caption{Second linear-time algorithm building a smallest suffixient set}\label{alg:fm-algo}
\SetKwInOut{Input}{input}
\SetKwInOut{Output}{output}
\Input{A text $T[1..n]$ over a finite alphabet $\Sigma$.}
\Output{A smallest suffixient set for $T$.}
{$\mathcal{S} \gets \emptyset$}\;
{$\BWT\ \gets \textrm{BWT}(\rev{T})$;
$\LCP\ \gets \textrm{LCP}(\rev{T})$;
$\SA\ \gets \textrm{SA}(\rev{T})$\label{line:fm.compute arrays}\;
$\NSV \gets \textrm{NSV}(\textrm{LCP}) $;
$\PSV \gets \textrm{PSV}(\textrm{LCP}) $\label{line:fm.compute arrays sv}\;
}
$R[1,\sigma] \gets ((sa\_pos \gets 0, text\_pos \gets 0, active \gets false, nsv \gets n + 1)\times \sigma)$\label{line:fm.init}\;
\For{$i = 2, \dots, n$}{
\If{$\BWT[i] \neq \BWT[i - 1]$}{
\For{$i'\in \{i - 1, i\}$}{
\If{$R[\BWT[i']].sa\_pos \leq \PSV[i]$ \label{line:fm.check fc}}{
\If{$R[\BWT[i']].nsv < i$ \label{line:fm.check fm}}{
$\SSS \gets \SSS \cup \{R[\BWT[i']].text\_pos\}$\;
}
$R[\BWT[i']] \gets (i, n - \SA[i'] + 1, true, \NSV[i])$\label{line:fm.update} \;
}
}
}
}
\ForEach{$c \in \Sigma$\label{line:fm.eval}}
{
\textbf{if $R[c].active = true$ then} $\SSS \gets \SSS \cup \{R[c].text\_pos\}$\;
}
\Return $\SSS$\;
\end{algorithm}

\medskip
\begin{lemma}
    Given a text $T[1,n]$ over alphabet of size $\sigma$, Algorithm~\ref{alg:fm-algo} computes a smallest suffixient set $\mathcal{S}$ in $O(n)$ time and $O(n)$ words of space. 
\end{lemma}
\begin{proof}
We prove the correctness by showing the following invariant: After we scanned the \BWT\ up to position $i$, it holds for every $c\in\Sigma$ that (i) $R[c]$ has the information of the last $c$-run break found up to position $i$ which is \textit{first $c$-candidate} unless $c$ has not appeared yet, and (ii) all \textit{first $c$-maximum} positions $i'$ with $\NSV[i']<i$ have been reported. 

Since $R[c]$ is set to $(0, 0, false, n + 1)$ at the beginning of the Algorithm (line~\ref{line:fm.init}) for each $c\in\Sigma$, the first $c$-run break always satisfies the condition in line~\ref{line:fm.check fc}, so $R[c]$ is updated and set as \textit{first c-candidate} in line~\ref{line:fm.update}. This is correct because the first $c$-run break is always \textit{first c-candidate}. Note that the condition in line~\ref{line:fm.check fm} is not satisfied for the first $c$-run break, so it does not report anything. 
Therefore after processing the first $c$-run break, the invariant holds.

Now assume we have scanned the $\BWT$ until position $i - 1$ and suppose $i$ is a $c$-run break, i.e., $\BWT[i'] = c$ with $i' \in \{i - 1, i\}$. Also, assume that $j (< i)$ is the previous $c$-run break which is \textit{first c-candidate}, thus $\BWT[j'] = c$ with $j' \in \{j - 1, j\}$. Then $R[c].sa\_pos = j$, $R[c].text\_pos = n - \SA[j'] + 1$, $R[c].active = true$, and $R[c].nsv = \NSV[j]$. If $R[c].sa\_pos \leq \PSV[i]$ (line~\ref{line:fm.check fc}), then we know $i$ is \textit{first c-candidate} and we update the attributes of $R[c]$ according of $i$'s values: $R[c].sa\_pos \gets i$, $R[c].text\_pos \gets n - \SA[i'] + 1$, $R[c].active \gets true$, and $R[c].nsv \gets \NSV[i]$; hence Invariant (i) holds after processing position $i$. Before updating $R[c]$, we check if $R[c].nsv < i$. If it is satisfied, $j$ is \textit{first c-maximum} and we add $R[c].text\_pos$ to the suffixient set built up to now (line~\ref{line:fm.check fm}), with which Invariant (ii) holds. After processing all positions, by the invariants, for every $c\in \Sigma$ that occurs in $T$, $R[c]$ has the last $c$-run break which is \textit{first c-candidate}, which is $\textit{first c-maximum}$ by definition), so we add it to the outputted suffixient set in line \ref{line:fm.eval}.
Note that for every $c\in\Sigma$, $R[c].active$ is set to $true$ if and only if $c$ occurs in $T$ because it is set to $true$ at the first $c$-break and remains $true$, which ensures that we do not add $R[c].text\_pos$ to the suffixient set if $c$ does not appear in $T$.

As far as the running time is concerned, computing $\BWT$, $\LCP$, and $\SA$ on line \ref{line:fm.compute arrays} takes $O(n)$ time. In addition, computing $\PSV$ and $\NSV$ on line \ref{line:fm.compute arrays sv} also takes $O(n)$ time using a stack-based algorithm as mentioned earlier. 
Scanning the $\BWT$ takes $O(n)$ time, and it is performed $O(1)$ operations for each position of the $\BWT$. Then, the total running time is $O(n)$.
\end{proof}

In Figure \ref{fig:fm-example} we show an example of how Algorithm \ref{alg:fm-algo} works.

\begin{figure}[!htb]
\centering
\includegraphics[scale=0.8]{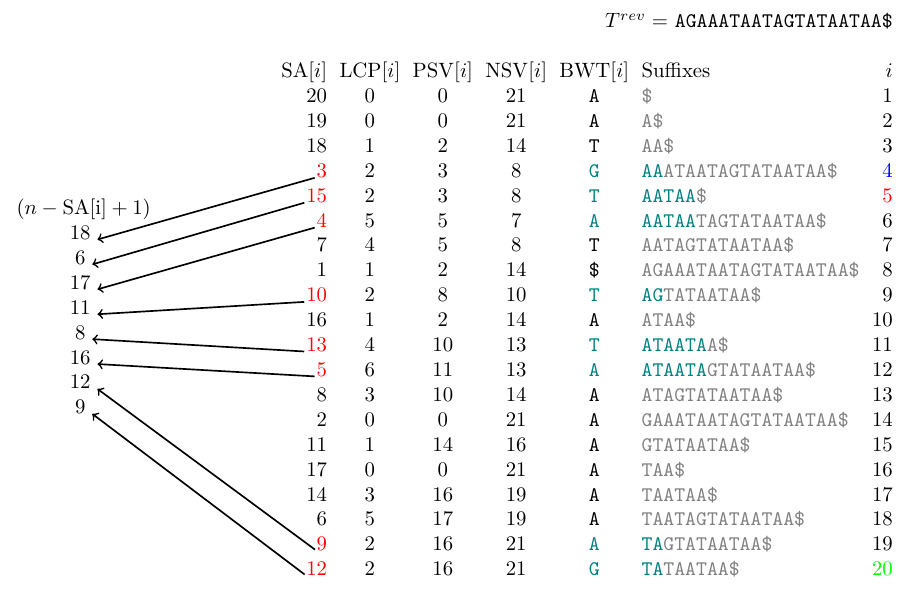}
\caption{
The figure shows the data used by Algorithm \ref{alg:fm-algo} to construct a smallest suffixient set $\SSS$ for a text $T[1,n]$ with $n = 20$. In addition to the data shown in Figure \ref{fig:run-example}, here we include the $\PSV$ and $\NSV$ arrays.
For brevity, we show how Algorithm \ref{alg:fm-algo} works only on the $\tt G$-run breaks. Consider the first $\tt G$-run break at position $i = 4$ (highlighted in blue on column $i$). At this stage we have $i' = i$ and $R[{\tt G}] = (0, 0, false, 21)$. Since $R[{\tt G}].sa\_pos = 0 < 3 = \PSV[4]$ and $i = 4 < 21 = R[{\tt G}].nsv$, then condition on line 8 is satisfied but condition on line 9 is not, so we do not add $R[{\tt G}].text\_pos$ to $\SSS$ in line 10 and we update $R[{\tt G}] = (i, n - \SA[i'] + 1, true, \NSV[i]) = (4, 18, true, 8)$ on line 12.
For the next ${\tt G}$-run break at position $i = 5$ (highlighted in red on column $i$), we have $\PSV[4] = 3 < 4 = R[{\tt G}].sa\_pos$, then condition on line 8 is not satisfied and nothing else is done in this iteration. The next ${\tt G}$-run break occurs at position $i = 20$ (highlighted in red on column $i$). We have $i' = i$. Since $R[{\tt G}].sa\_pos = 4 < 16 = \PSV[20]$ and $R[{\tt G}].nsv = 8 < 20 = i$, we add $R[{\tt G}].text\_pos = 18$ to $\SSS$ on line 10 and we update $R[{\tt G}] = (i, n - \SA[i'] + 1, true, \NSV[i]) = (20, 9, true, 21)$ on line 12. Finally, since $R[{\tt G}].active = true$, we add $R[{\tt G}].pos = 9$ on the final {\tt foreach} loop on line 18.
}\label{fig:fm-example}
\end{figure}

\section{Implementation details}\label{sec:implementation}

We discuss two optimizations to Algorithm \ref{alg:one-occ} that in practice are expected to lower the I/O complexity to near-optimal. These optimizations have been included in our implementation, which in the next section we compare with the $r$-index.
Our optimizations are motivated by a particular application: pattern matching on repetitive genomic collections (in particular, in this Section we will assume $\sigma\in O(1)$). As discussed in the introduction, a lot of interest has recently been dedicated to solving pattern matching queries of short DNA fragments ($m$ is of the order of thousands of characters) on collections composed of several genomes from the same species. Since genomes within those collections have $\ge 99.9\%$ similarity, repetitive genomic collections tend to be extremely repetitive (and compressible).

\paragraph{Optimization I}

We start from the version of Algorithm \ref{alg:one-occ} that implements
$\opsearch()$ with binary search on $\SuA$ and with random access on the oracle as described in Theorem \ref{thm: pattern matching I/O}.
The first optimization that we introduce is that we start the \texttt{while} loop of Algorithm \ref{alg:one-occ} from $i = \min\{m, k\}$, where $k = \lceil \log_\sigma \chi \rceil + O(1)$ is a parameter whose value is justified below (paragraph \emph{Optimization II}). 
Intuitively, this optimization is motivated by the fact that we expect all short (length $< k$) prefixes of $P$ to be right-maximal with large probability. This implies that those prefixes will trigger many useless binary searches in Algorithm \ref{alg:one-occ} \footnote{As a matter of fact, only one binary search is really useful: the one for the longest pattern prefix $P[1,i]$ suffixing $T[1,x]$ for some $x \in \SuA$. Since, however, we do not know $i$ in advance, our algorithm needs to run several binary searches on different pattern prefixes before finding $i$.}. 
If the first binary search with the above value of $i$ finds a match of length $i$, our search procedure  proceeds normally as described in Algorithm \ref{alg:one-occ}. Otherwise, we decrement $i$ by one unit and repeat in order to find the longest prefix of $P[1,k]$ suffixing $T[1,x]$ for some $x\in \SuA$. 

\paragraph{Optimization II}
Let $k = \lceil \log_\sigma \chi \rceil + O(1)$ be the same parameter used in the above optimization. For every $x\in \SuA$, we bit-pack $\rev{T[x-k+1,x]}$ in an integer of $k\lceil \log_2\sigma \rceil$ bits. Since we assume $\sigma \in O(1)$, each such integer is bounded by $O(\sigma^k) = O(\chi)$.
Let $L$ be the list, sorted in ascending order, containing such integers.
$L$ is encoded succinctly using an Elias-Fano predecessor data structure (see for example \cite{BelazzouguiN15}; in our code we employ the {\tt sd\_vector} implementation of {\tt sdsl-lite}\footnote{\url{https://github.com/simongog/sdsl-lite/blob/master/include/sdsl/sd_vector.hpp}}), using $\chi \log(\sigma^k/\chi) + O(\chi) = O(\chi)$ bits, which add only low-order terms on top of the space ($\chi\log n$ bits) of the suffixient array.

The implementation of $\opsearch(\beta)$ is modified as follows. If $LCS(\beta, T[1,x]) < k$ for all $x\in \SuA$, then $L$ suffices to answer $\opsearch(\beta)$ with just one predecessor query on the integer formed by packing the characters of $\rev \beta$.
Otherwise ($LCS(\beta, T[1,x]) \ge k$ for some $x\in \SuA$), a predecessor query on the last $k$ characters of $\beta$ yields the range on $\SuA$ corresponding to prefixes $T[1,x]$ ($x\in \SuA$) suffixed by $\beta[|\beta|-k+1,|\beta|]$. 
This range is then refined with the remaining characters of $\beta$ via standard binary search and 
random access on $T$. 
For the chosen value of $k$, the Elias-Fano data structure answers predecessor queries in $O(\log(\sigma^k/\chi))=O(1)$ time. 

Recall that, by Theorem \ref{thm: pattern matching I/O}, our index has worst-case $O( (1+m/B) \cdot d \log \chi)$ I/O complexity, where  $d\le m$ is the node depth of the locus of $P$ in the suffix tree of $T$ (that is, the number of suffix tree edges traversed while searching $P$ in the suffix tree of $T$).
In our application (repetitive genomic collections), Optimization I is expected to drastically reduce the  term $d$. Intuitively, this happens because Optimization I allows us to skip $k$ levels of the Suffix Tree of $T$ in constant time, therefore $d$ is replaced by the number $d'$ of edges that we traverse in the suffix tree starting from the locus of $P[1,k]$. A closer analysis shows that, rather than considering the full suffix tree, the same observation holds on the sparse suffix tree containing only the $\chi$ text suffixes starting in $j+1$ for $j\in \SuA$. To see this, observe that Algorithm \ref{alg:one-occ} always skips the longest common prefix of $T[j+1..]$ and $P[i+1..]$ (by incrementing $i$ and $j$ in Line \ref{line: incr i j}) for $j$ that always belongs to $\SuA$. But then, we expect that the subtree (of the above sparse suffix tree) rooted in the locus of $P[1,k]$ contains very few leaves (if the individual genomes were really uniform, we would expect this number of leaves to be constant with large probability). In turn, $d'$ is upper-bounded by such number of leaves. 

Optimization II, on the other hand, is expected to reduce the term $\log\chi$. The intuitive reason is that list $L$ defined above contains $\chi$ integers. Since (by Lemma \ref{lem:attractor}) $\SuA$ forms a small string attractor and individual genomes in repetitive collections have high entropy, we expect $L$ to have high entropy as well. Therefore, we expect that the seeding strategy implemented in Optimization II drastically reduces the binary search range from  $\chi$ to a very small length (in fact, this value would be constant on expectation if genomes were really uniform strings).

\section{Experimental results}\label{sec:experiments}

We implemented all construction algorithms of Section \ref{sec:smallest-SuffixientSet-comp} 
and the suffixient-based index of Theorem \ref{thm: pattern matching} (in two variants, read below)
in {\tt C++} and made the code publicly available at {\tt \url{https://github.com/regindex/suffixient-array}}. 
The one-pass Algorithm \ref{alg:one-pass} uses Prefix Free Parsing (PFP) as discussed in Section \ref{sec:one pass} and runs in compressed working space. 

We divide our experimental evaluation into four parts: first, we empirically study the new repetitiveness measure $\chi$ by comparing it with $\bar r$ under different alphabet orderings. Then, we assess the performance of our PFP-based one-pass algorithm on massive genomics datasets. 
The other construction algorithms (requiring much more working space) have been tested on smaller datasets; we provide this analysis in a separate paragraph. 
Finally, we compare the performance of our suffixient-based index with those of the Prefix Array and the $r$-index \cite{GNP20} on the task of locating one pattern occurrence, using three random access oracles (for both our index and the Prefix Array).

We ran our experiments 
using two workstations: the first is an Intel(R) Xeon(R) W-2245 CPU @ 3.90GHz with 8 cores and 128 gigabytes of RAM
running Ubuntu 18.04 LTS 64-bit. 
This workstation was used for the experiments on 
repetitiveness measures (Section \ref{sec:repetitiveness measure}), on
$\SuA$ construction on massive datasets (Section \ref{sec:massive_data}), and on indexing (Section \ref{sec:indexes}).
The second workstation is an
Intel(R) Pentium(R) IV, CPU @ 2.4 GHz with 8 cores, 10 MB cache, 256 gigabytes of RAM, 860 gigabytes of disk running Devuan GNU/LINUX 2.1. 
This workstation was used for the experiments on building Suffixient Arrays (Section \ref{sec:other construction}), which required more working memory.
We recorded the runtime and memory usage of our software by using the wall clock time and maximum resident set size from {\tt /usr/bin/time} and two {\tt C++} libraries: {\tt chrono} and {\tt malloc\_count}.

\subsection{The repetitiveness measure \texorpdfstring{$\chi$}{chi}}\label{sec:repetitiveness measure}

In this experiment, we 
use nine biological datasets divided into two corpora. The first corpus contains six datasets containing DNA sequences from the Pizza\&Chilli collection\footnote{ \url{https://pizzachili.dcc.uchile.cl}}, while the second contains three datasets made by concatenating genomic sequences of different types. The three datasets of the second corpus are formed by 36,000 {\em SARS-CoV-2} assembled viral genomes, 220 {\em Salmonella enterica} assembled bacterial genomes, and 19 copies of the human {\em Chromosome 19}, respectively. We downloaded data from different sources: the viral genomes from the COVID-19 Data Portal\footnote{\url{https://www.covid19dataportal.org}}, the bacterial genomes from NCBI\footnote{\url{https://www.ncbi.nlm.nih.gov/assembly/?term=Salmonella+enterica}}, while the Chromosome 19 sequences belong to the 1,000 Genome project\footnote{\url{https://github.com/koeppl/phoni}}.
All datasets are filtered to keep only DNA nucleotide characters ${\tt A,C,G,T}$, in addition to the string terminator $\$$. As a consequence, the alphabet size in our experiments was $\sigma = 5$.

We compare the behavior of our new measure $\chi$ with the well-known measure $\bar r$: the number of equal-letter runs in the BWT of the reversed text.
Since, unlike $\chi$, measure $\bar r$ depends on a specific alphabet ordering, we explicitly compute $\bar r$ for all $(\sigma-1)! = 24$ possible orderings (character $\$$ is always required to be the lexicographically-smallest one). Table \ref{tab:exp1} shows our results. Among all possible alphabet orderings, we only show the default one ${\tt A<C<G<T}$, and the ones yielding the largest and smallest $\bar r$.

\begin{table}[!htb]
    \centering
    \begin{adjustbox}{width=\textwidth}
\begin{tabular}{|c|l|r|r|r|r|r|}\hline
corpus & dataset & dataset length & $\chi$ & default $\bar r$ & min.\ $\bar r$ & max.\ $\bar r$ \\\hline
\multirow{6}{*}{Pizza\&Chili} & Cere & 428,111,842 & 9,945,553 & 11,556,075 & 11,512,279 & 11,578,575 \\
 & E.\ Coli & 112,682,447 & 13,126,726 & 15,041,926 & 14,974,687 & 15,043,686\\
 & Influenza & 154,795,431 & 2,234,085 & 3,015,173 & 2,979,820 & 3,023,774\\
& Para & 412,279,605 & 13,399,378 & 15,581,823 & 15,475,176 & 15,635,533\\
 & dna & 403,920,884 & 215,201,641 & 243,480,086 & 243,179,189 & 243,574,132 \\
 & dna.001.1 & 104,857,499  & 1,414,711 & 1,717,155 & 1,715,622 & 1,718,418 \\
\hline
   & Chrom.\ 19 & 1,060,302,648 & 28,265,608 & 32,501,979 & 32,413,939 & 32,516,882 \\
 biological & Salmonella & 1,043,921,520 & 19,070,836 & 22,407,573 & 22,254,453 & 22,416,480\\
  & SarsCov2 & 1,052,893,440 & 829,777 & 1,091,476 & 1,088,925 & 1,093,104 \\

 \hline
\end{tabular}
\end{adjustbox}
\vspace{0.5mm}
\caption{Summary of the results on the nine datasets. From left to right, we report the corpus name, the dataset name, the dataset length (number of characters), the size of the smallest suffixient set ($\chi$), and the number of runs of the BWT of the reversed text ($\bar r$) for three different alphabet ordering: the default lexicographic order, and the two orderings leading to the minimum and maximum $\bar r$. } \label{tab:exp1}
\end{table}

We observe that, in practice, $\chi$ is very close to (and always smaller than) $\bar r$ under any alphabet ordering; 
notice that theory only predicts $\chi \leq 2\bar r$ (Lemma~\ref{lem:upper_bound}).
In this experiment, the minimum of $\bar r$ is always between 1.13 (dna) and 1.33 (Influenza) times larger than $\chi$. 
We also observe that, while the alphabet ordering does not have a big influence on $\bar r$, the default alphabet ordering never yields the smallest $\bar r$.
This suggests experimentally that $\chi$ is a better repetitiveness measure than $\bar r$, since it is not affected by the alphabet order (while always being close to $\bar r$). 

\subsection{Building Suffixient Arrays of massive datasets}
\label{sec:massive_data}
In this experiment, we tested our PFP-based one-pass algorithm on two large genomic datasets: the first dataset is formed by  1,000 human Chromosome 19 copies (59GB), while the second contains 2,005,773 Sars-CoV-2 viral sequences (60GB).

As discussed in Section \ref{sec:one pass}, 
we read the text $T$ backwards from the disk keeping in RAM only the PFP data structures. We employed the PFP implementation provided in the {\tt pscan.cpp} file contained in {\tt Big-BWT} ({\tt https://github.com/alshai/Big-BWT}) and ran it using 16 threads.

In both datasets, our algorithm performed very well with a maximum resident set size of 20.6 GB  and 29.5 GB and a wall clock time of 55:55 (min:sec) and 2:16:48 (h:min:sec), on Chromosome 19 and Sars-CoV-2, respectively. 

\subsection{Other Suffixient Array construction algorithms}\label{sec:other construction}

In this experiment, we measured the wall clock time and internal memory peak, as reported by the {\tt usr/bin/time} utility, needed by the 4 suffixient array construction algorithms presented in Section \ref{sec:smallest-SuffixientSet-comp}: the one-pass algorithm (\texttt{one-pass}) using the non-compressed $\SA$, $\LCP$, $\BWT$ arrays, the PFP-based one-pass algorithm (\texttt{pfp}), and the first (\texttt{linear}) and second (\texttt{fm}) linear time algorithms. As input for the algorithms, we used several prefixes of the Sars-Cov-2 dataset described in Section \ref{sec:massive_data}. The $i$-th prefix of such dataset was created by concatenating the first $m \in \{2^i : i \in[14, 19]\}$ genomic sequences from the original dataset.

In Figure \ref{fig:results-all-algs} we provide a summary of the experimental results.
We note that \texttt{pfp} is always faster and uses less space than the other three algorithms. In particular, on average, this algorithm was $3.74$ times faster than {\texttt{fm}}, $5.59$ times faster than \texttt{linear}, and $4.66$ times faster than \texttt{one-pass}. At the same time, \texttt{pfp} used $57.96$ times less space than {\texttt{fm}} and $20.56$ times less space than \texttt{linear} and \texttt{one-pass}. This is due to the PFP preprocessing which exploits natural repetitiveness of genomic data to reduce time/space execution requirements. The {\texttt{fm}} algorithm was always faster than the other non-PFP-based algorithms. More in detail, {\texttt{fm}}  was on average $1.41$ times faster than \texttt{linear} and it was $1.18$ times faster than \texttt{one-pass}. 
On the other hand, \texttt{fm} used much more working space than \texttt{linear} ($3.03$ times), to the point that it could not finish the computation for the largest prefix containing $2^{19}$ sequences due to memory requirements exceeding the available 256 GB of RAM. Interestingly, the (non-PFP based) \texttt{one-pass} algorithm was, on average, $1.2$ times  faster than \texttt{linear} while using the same amount of space. This could be explained given that, for constant alphabets (as the one of the Sars-Cov-2 dataset) the $O(n + \overline{r}\log \sigma)$ time complexity of the one-pass algorithm becomes $O(n)$. In addition, \texttt{one-pass} runs in a more chache-friendly way than \texttt{linear}  since it only performs one linear scan of the BWT, SA and LCP vectors and $O(\bar r)$ linear scans of the vector $R$ containing the $\sigma$ one-character right-maximal extension candidates.

\begin{figure}
\begin{minipage}{0.49\textwidth}
\centering
\includegraphics[width=\textwidth]{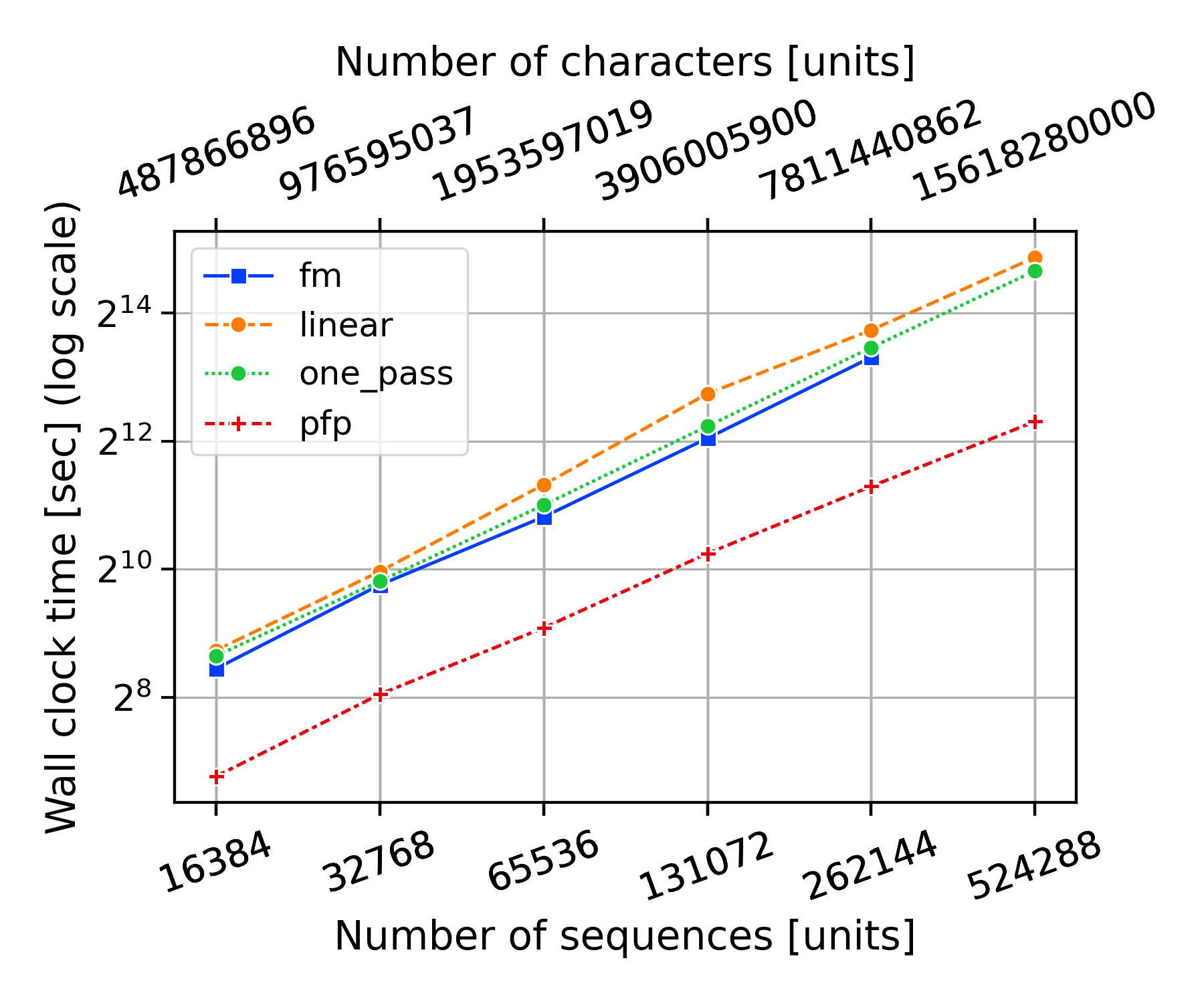} 
\end{minipage}
\hfill
\begin{minipage}{0.49\textwidth}
\centering
\includegraphics[width=\textwidth]{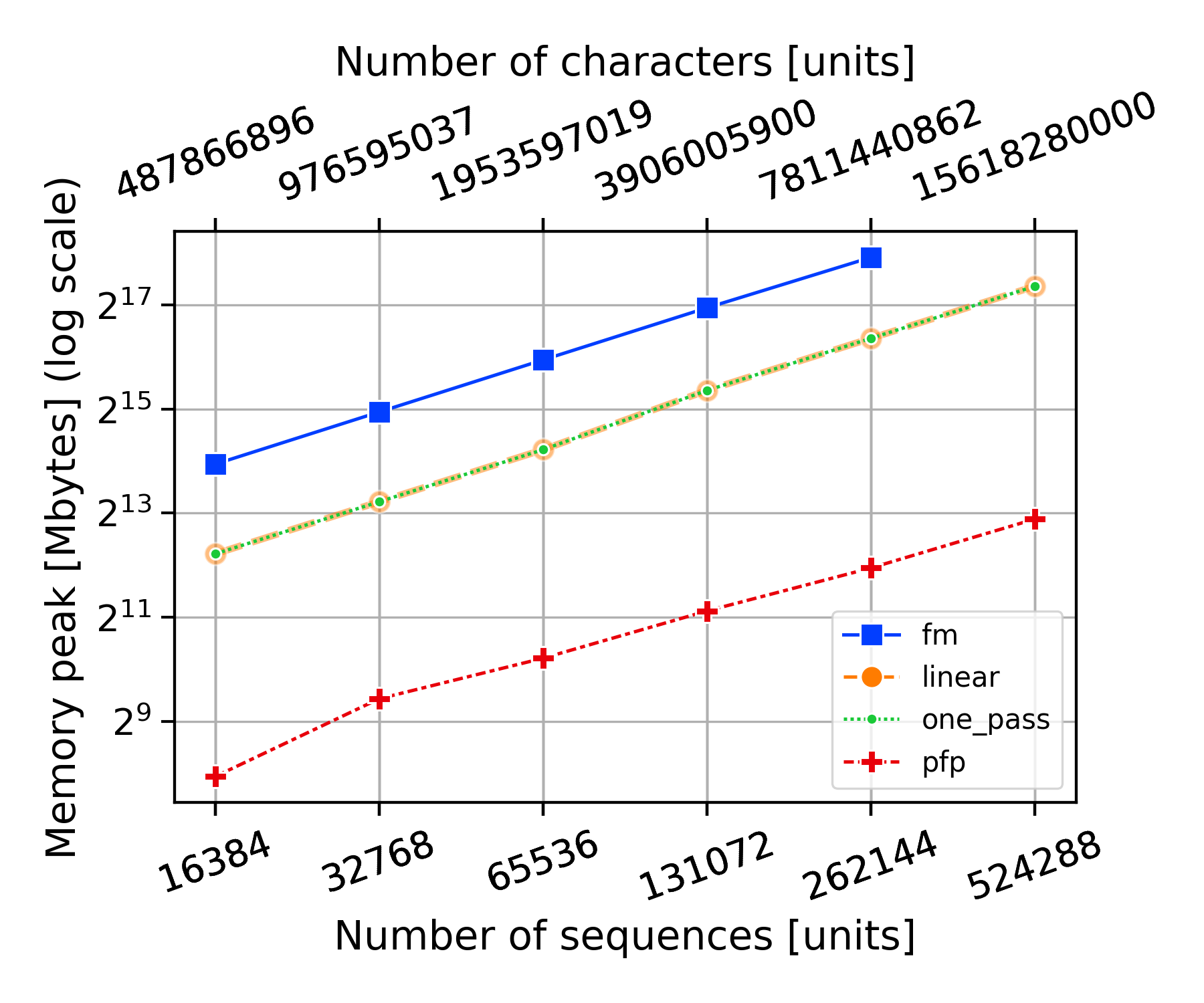} 
\end{minipage}
\caption{Wall clock time (left) and memory peak (right) for computing a smallest suffixient set for six SarsCoV2 datasets using the four algorithms presented in Section \ref{sec:smallest-SuffixientSet-comp}. We generate the datasets by taking prefixes of increasing size containing $2^i : i \in [14,19]$ genomic sequences.}
\label{fig:results-all-algs}
\end{figure}

\subsection{Pattern matching queries}\label{sec:indexes}

In this last experiment, we compare our compressed index based on the Suffixient Array described in Section 4.1 with the Prefix Array and with the $r$-index of Gagie et al.\ \cite{GNP20} on the task of locating one pattern occurrences.

\subsubsection{Indexes and random access oracles}

We tested four indexes, whose details are given below: Suffixient-Array based indexes ($\SuA$ and opt-$\SuA$), Prefix Array ($\PA$), and the toehold lemma of the $r$-index (\texttt{toehold}), i.e. the subset of the $r$-index sufficient to locate one pattern occurrence. In the plots, we also show the size of the full $r$-index for a comparison.

\paragraph*{Random access oracles}

The Suffixient-Array based indexes ($\SuA$ and opt-$\SuA$) and the Prefix Array ($\PA$) have been tested in combination with three random access oracles: the \texttt{bitpacked} text (a plain text representation using 2 bits per character), the subset \texttt{lz77} of the Lempel-Ziv 77 index of Kreft and Navarro\footnote{\url{https://github.com/migumar2/uiHRDC/tree/master/uiHRDC/self-indexes/LZ}} sufficient to perform random access (i.e. from the original index we removed all data structures needed for locating patterns), and an ad-hoc optimized implementation of Relative Lempel-Ziv \cite{kuruppu2010relative} (\texttt{rlz}). The latter random access data structure: (i) packs the reference characters using 2 bits per character, (ii) automatically detects the best prefix of the input string to be used as the reference optimizing the overall compression ratio (trying lengths being a power of $(1+\epsilon)$, for a small $\epsilon>0$), and (iii) optimizes the extraction of long substrings by accessing the predecessor structure (storing phrase borders) only when needed, i.e. a number of times equal to the number of phrases overlapping with the string to be extracted. 

\paragraph{Suffixient Array-based indexes ($\SuA$ and \texttt{opt-}$\SuA$)}

We implemented two variants of our suffixient-based index and combined them with the three random access oracles described above. 
The two variants implement the Optimizations described in Section \ref{sec:implementation}. The first variant, deemed $\SuA$ in the plots below, implements only Optimization I. The second variant, deemed opt-$\SuA$, implements both Optimizations I and II. Note that we implemented these optimizations assuming DNA alphabet, i.e., $\sigma = 4$, and we always set $k = 14$.

\paragraph{Prefix Array index (PA)}

This index  implements classic binary search on the Prefix Array, combining it with the three random access oracles described above.

\paragraph{Toehold Lemma of the $r$-index (\texttt{toehold})}

We used the subset of the $r$-index \cite{GNP20} (we employ the original {C++} implementation\footnote{\url{https://github.com/nicolaprezza/r-index}}) 
sufficient to locate one pattern occurrence (the so-called \emph{Toehold lemma}); from the original implementation, we removed all data structures required to locate all other pattern occurrences. 
This made the index substantially smaller (compared to the full $r$-index). To ease reproducibility of our experiments, we included this subset of the $r$-index in the repository containing our Suffixient Array index.

\subsubsection{Datasets and patterns extraction}
In this experiment, we used the three genomic datasets from Table \ref{tab:exp1}: Chromosome 19, Salmonella, and SARS-CoV-2. Since our \texttt{opt-}$\SuA$ implementation is designed for DNA alphabets of four characters, we preprocess the datasets to remove all non-DNA characters. This ensures that all four competing methods use the same three input texts.
For each dataset, we generated three sets of 100,000 patterns of lengths 10, 100, and 1000 each. Specifically, given a text $T[1,n]$ and a pattern length $m$, we selected patterns by drawing uniform random positions $i \in [1,n-m+1]$.

\subsubsection{Results}

The plots below show a comparison between the query times of the four indexes and the RAM throughput of our workstation when extracting substrings of length $m = $ 10, 100, and 1000 from uniformly-chosen positions in a uniform text of 1 billion characters. On our workstation, we obtained the following RAM throughput benchmarks: 6.23341 ns/character for $m=10$, 2.31762 ns/character for $m=100$, and 1.41736 ns/character for $m=1000$. The RAM throughput for reading the entire text of 1 billion characters was of 1.17591 ns/character.
Note that the RAM throughput indicates a hard lower bound for processing a pattern of a certain length since it provides the maximum rate at which a block of length $m$ can be retrieved from our system's internal memory.
In turn, this is a hard lower bound for any deterministic pattern matching algorithm guaranteeing a correct answer.

We found out that the \texttt{lz77} oracle was always orders of magnitude slower than the \texttt{rlz} oracle, while always using essentially the same space. Symmetrically, the \texttt{rlz} oracle was always as fast as the \texttt{bitpacked} text, while using orders of magnitude less space. Since \texttt{rlz} always dominated the other two oracles, in Figure \ref{fig:results} we only show results using the \texttt{rlz} oracle. 

As the plots show,
the optimization opt-$\SuA$ only slightly increases the space usage of $\SuA$ while dramatically speeding up query times (even by one order of magnitude in some cases). 
Our results show that opt-$\SuA$ always dominates
the Prefix Array by a wide margin: opt-$\SuA$ was always faster than the Prefix Array while using orders of magnitude less space. 
As expected from our analysis of opt-$\SuA$ in the I/O model (Theorem \ref{thm: pattern matching I/O}), opt-$\SuA$ dominates
the toehold lemma of the $r$-index by a wide margin on the moderately-repetitive Chr 19 dataset. On this dataset, opt-$\SuA$ was always smaller than the toehold lemma and about \emph{half the size} of the full $r$-index. At the same time, opt-$\SuA$ was always \emph{from one to two orders of magnitude faster} than the toehold lemma. 
On the more repetitive Salmonella and SarsCov2 datasets, opt-$\SuA$ was only slightly larger than the toehold lemma and always smaller than the full $r$-index while at the same time being always from one to two orders of magnitude faster.

The plots show that opt-$\SuA$ was always at most 1 order of magnitude slower than the RAM throughput. This is the case of Chrom.\ 19 and patterns of length 100. On the other hand, for patterns of length 10 and 1000 the results were much more favorable. In the best case, for SarsCoV2 and length 1000, opt-$\SuA$ (about 3.5 ns per character) was only 2.5 times slower than the RAM throughput (about 1.4 ns per character). This aligns with our analysis in the I/O model, as for long patterns and repetitive collections made up of variations of a nearly-uniform string, fewer binary search steps are expected. We remark that our implementation is not yet heavily optimized and we expect further improvements with more optimization work.

\begin{figure}[htp]
    \centering

    \begin{minipage}{0.33\textwidth}
        \centering
        \includegraphics[width=\textwidth]{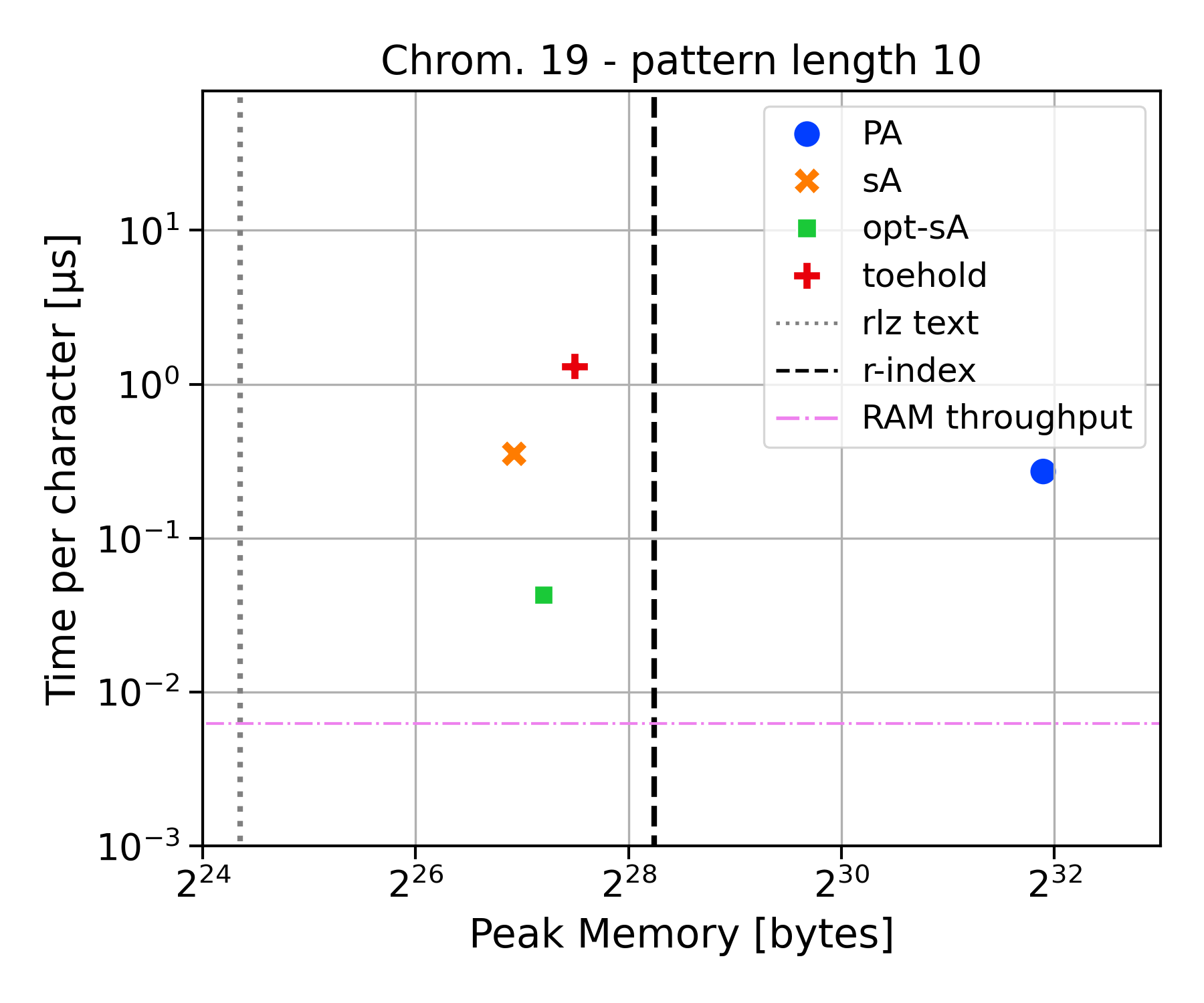} 
    \end{minipage}
    \hspace{-3mm}
    \begin{minipage}{0.33\textwidth}
        \centering
        \includegraphics[width=\textwidth]{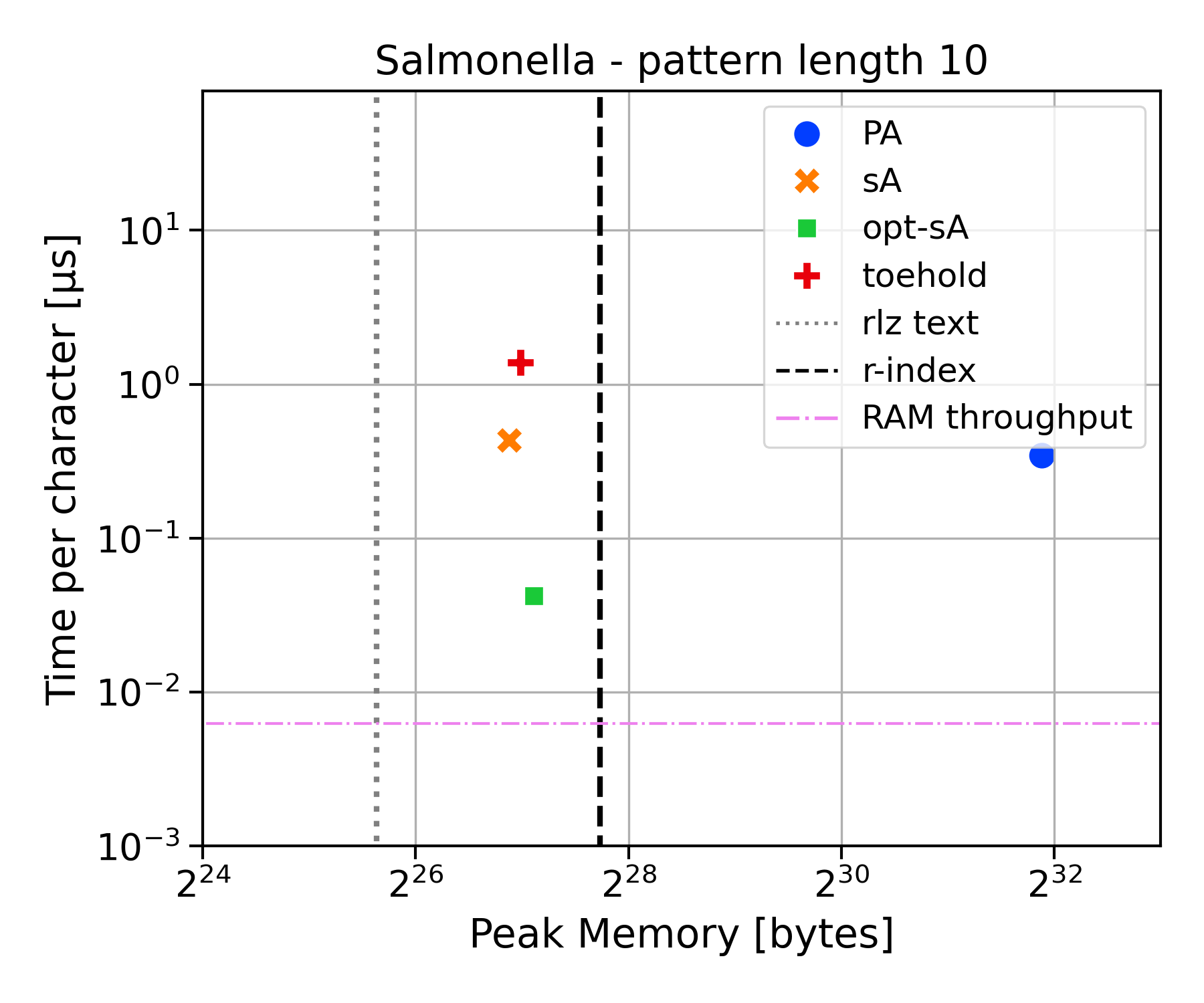} 
    \end{minipage}
    \hspace{-3mm}
    \begin{minipage}{0.33\textwidth}
        \centering
        \includegraphics[width=\textwidth]{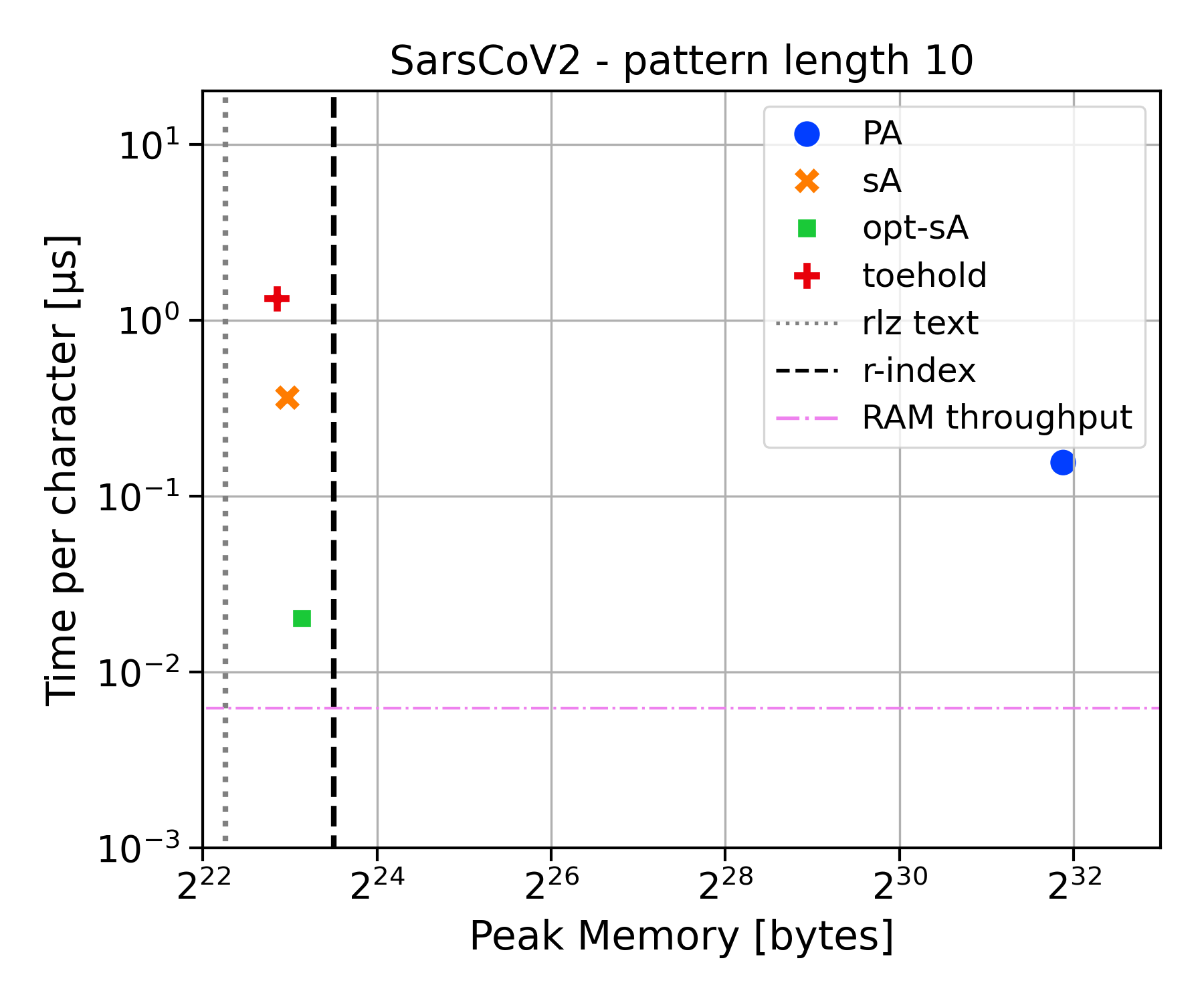} 
    \end{minipage}

    \vspace{0.25cm}  

    \begin{minipage}{0.33\textwidth}
        \centering
        \includegraphics[width=\textwidth]{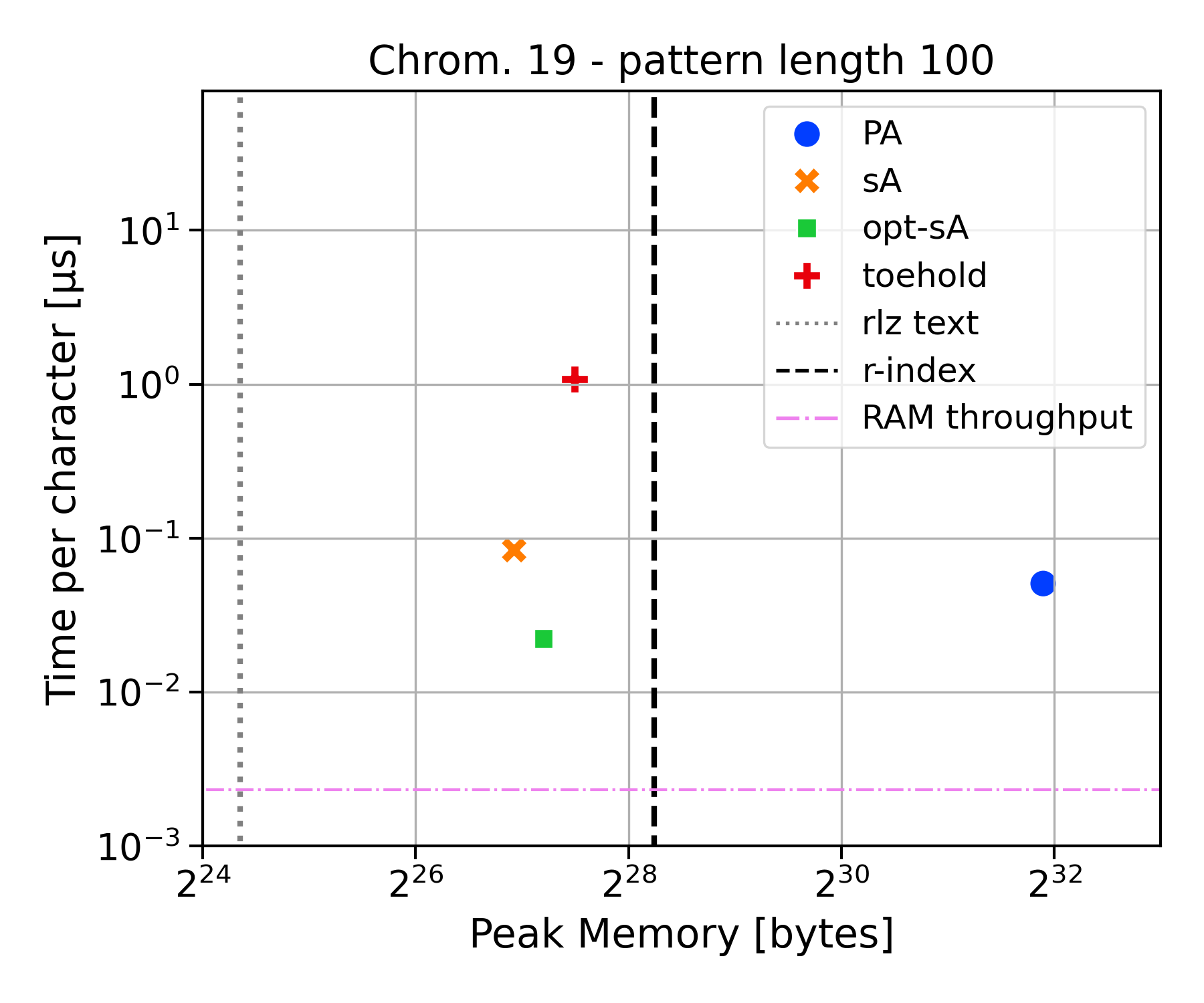} 
    \end{minipage}
    \hspace{-3mm}
    \begin{minipage}{0.33\textwidth}
        \centering
        \includegraphics[width=\textwidth]{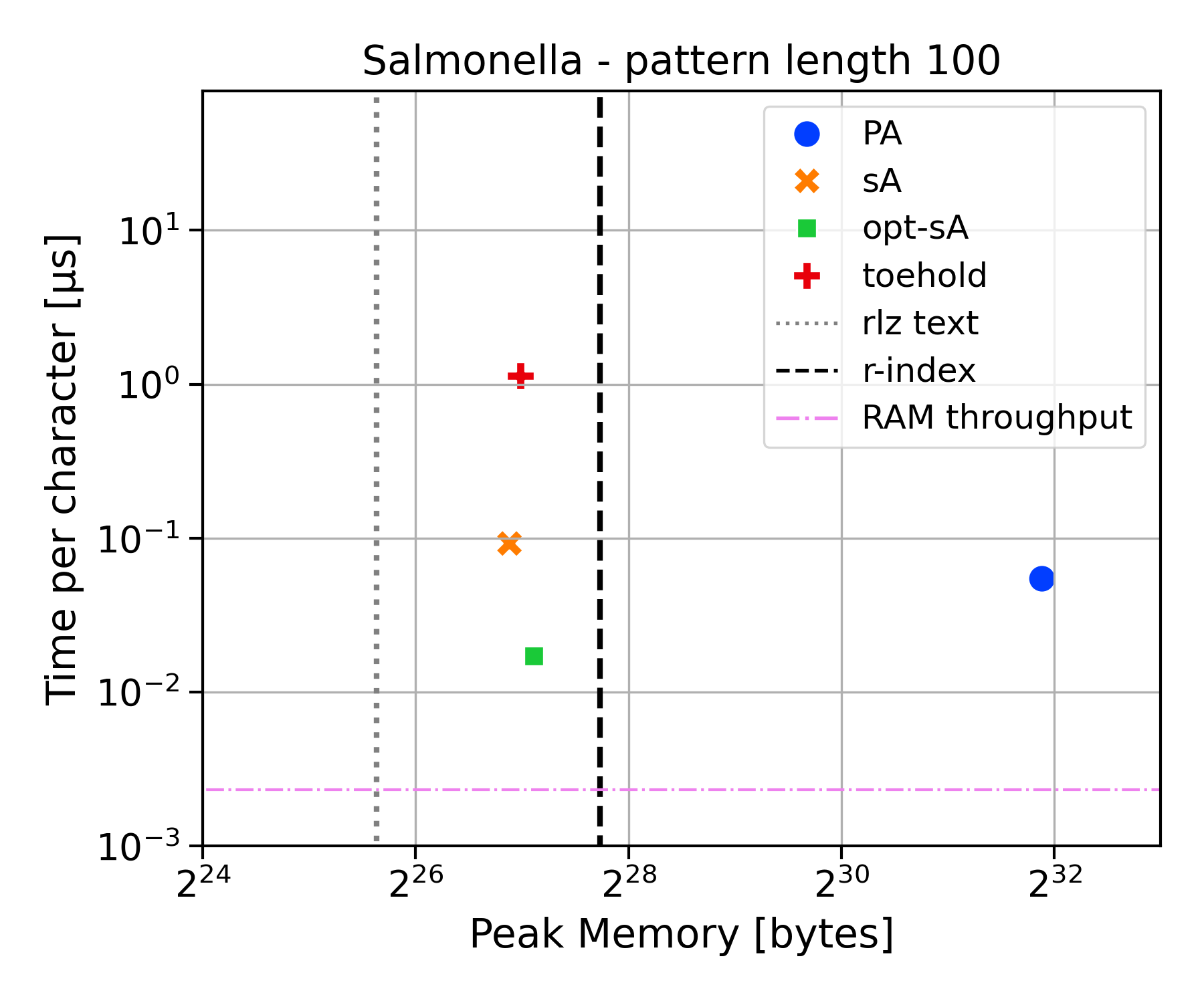} 
    \end{minipage}
    \hspace{-3mm}
    \begin{minipage}{0.33\textwidth}
        \centering
        \includegraphics[width=\textwidth]{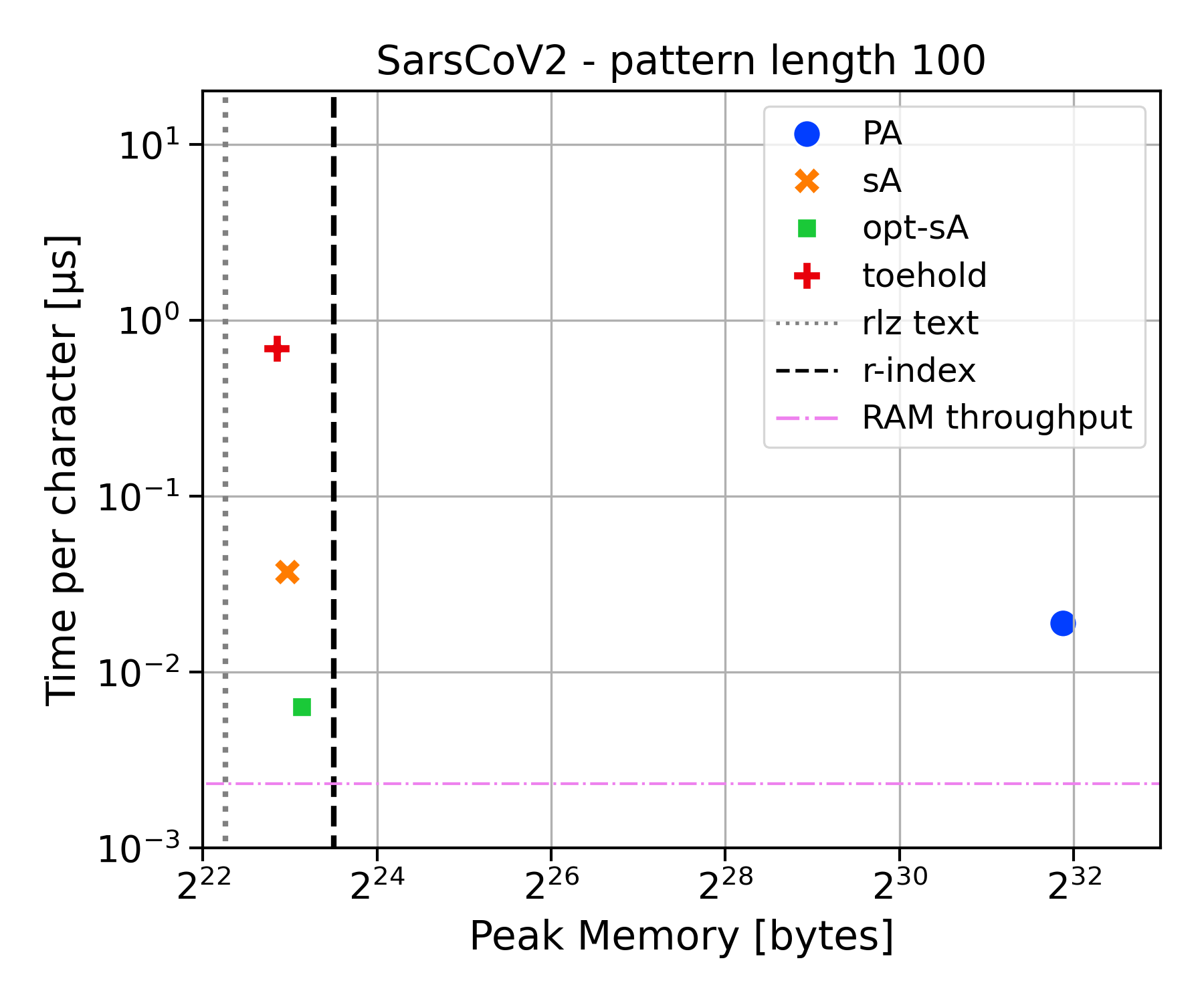} 
    \end{minipage}

    \vspace{0.25cm}  

    \begin{minipage}{0.33\textwidth}
        \centering
        \includegraphics[width=\textwidth]{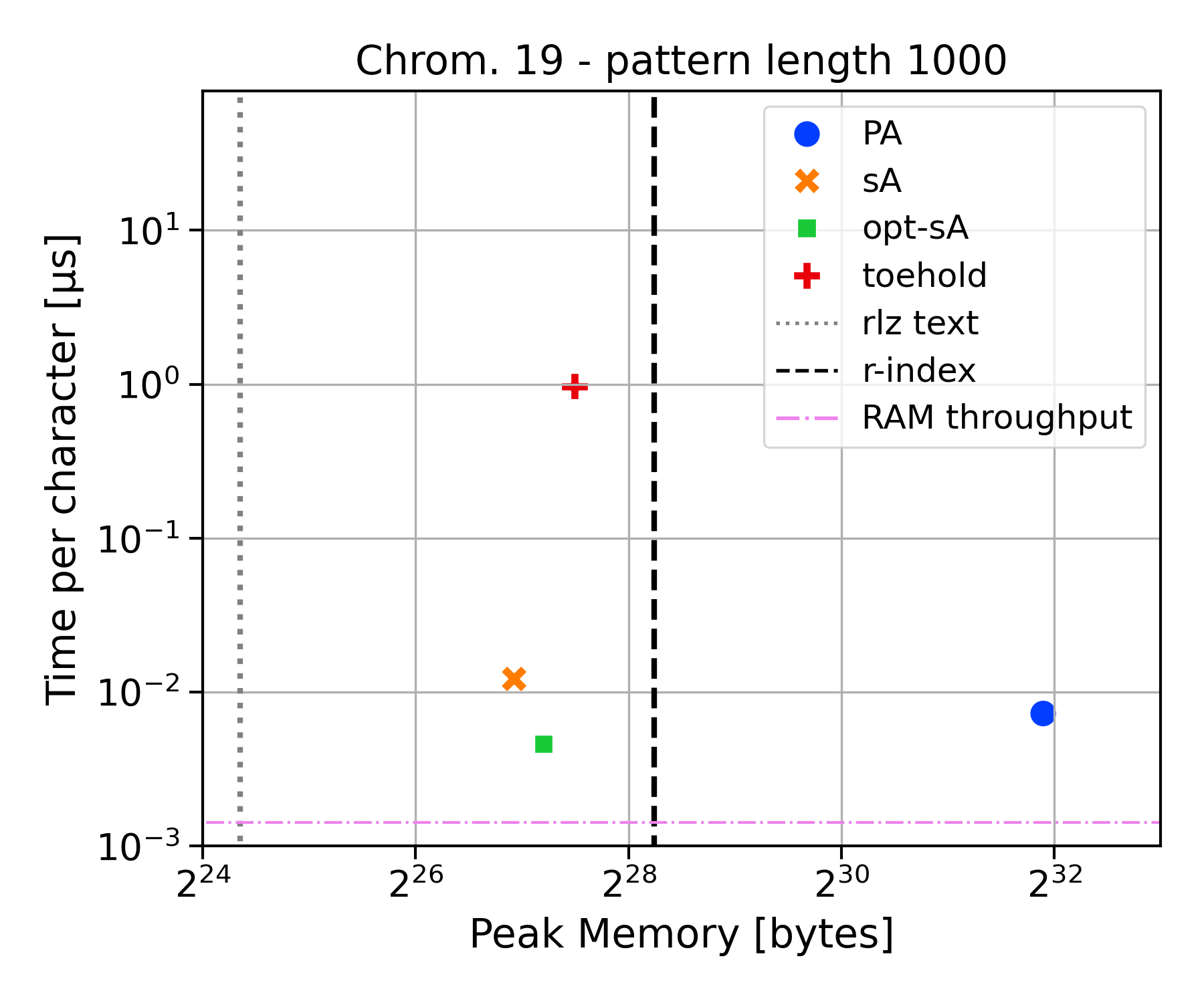} 
    \end{minipage}
    \hspace{-3mm}
    \begin{minipage}{0.33\textwidth}
        \centering
        \includegraphics[width=\textwidth]{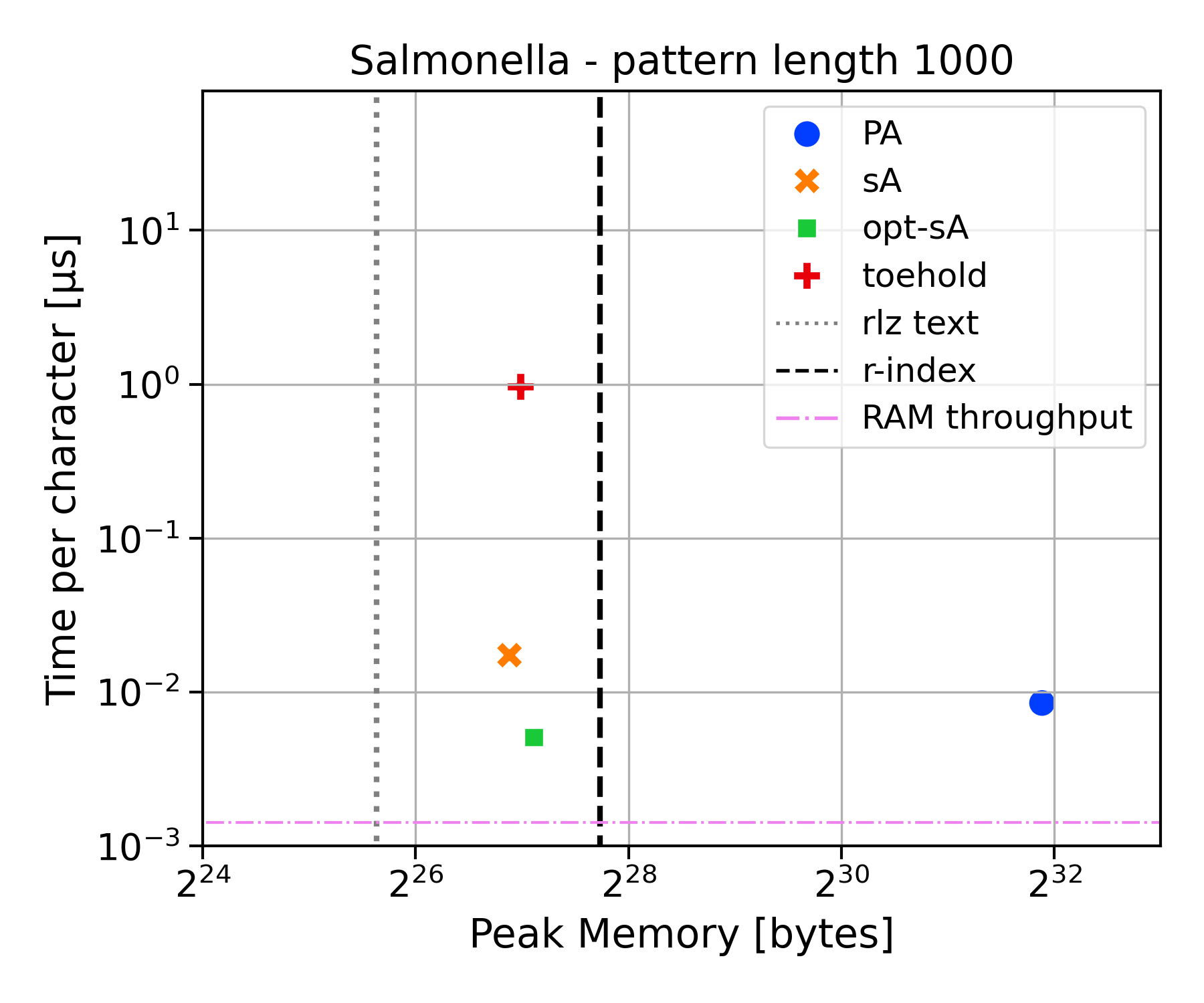} 
    \end{minipage}
    \hspace{-3mm}
    \begin{minipage}{0.33\textwidth}
        \centering
        \includegraphics[width=\textwidth]{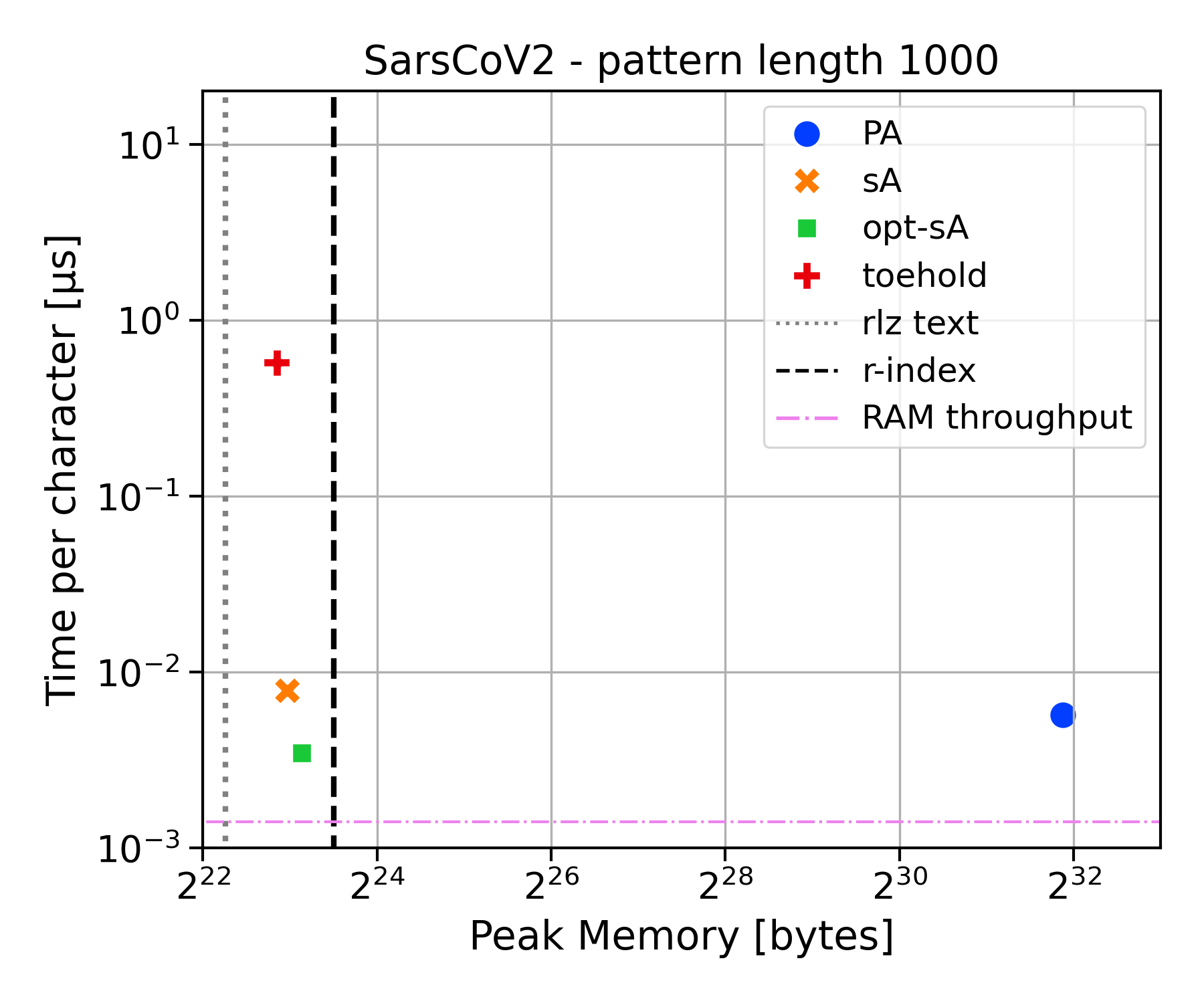} 
    \end{minipage}

    \caption{Summary of results for the three biological datasets. We compare the three implementations of our index, using the RLZ-based random access oracle, against the toehold lemma of the $r$-index. The y-axis represents the time per character required by each index to locate one occurrence across 100,000 patterns of different lengths. The x-axis shows the peak memory recorded during each execution. Both axes are in log scale. Additionally, we include three dashed lines indicating the size of the random-access text oracle (\texttt{rlz}), the size of the original $r$-index, and the RAM throughput of our workstation when extracting substrings of length 10, 100, 1000 from uniform positions in a uniform text of 1 billion characters. 
    }
    \label{fig:results}
\end{figure}

\section{Future Work}

We list some open problems on which we are currently working. 

\paragraph*{Queries in $O(\chi)$ space}

We have already mentioned the problem of determining whether $\chi$ is a reachable measure, i.e. $O(\chi)$ words are sufficient to store the text and to perform efficient random access. 
A possible way to attack this problem is to observe that our random access data structure based on string attractors (using $O(\chi \log(n/\chi))$ space) does not exploit any structural property of suffixient sets other than them being a string attractor. Suffixient sets are, however, very particular string attractors. Can we use other combinatorial properties of suffixient sets to support random access in $o(\chi \log(n/\chi))$ space?

Another interesting problem is to support counting and locating all pattern occurrences (rather than just one) in space $O(\chi)$ on top of the random access oracle.
Eventually, we would also like to compute the BWT range of the pattern with our mechanism, therefore completely replacing the Burrows-Wheeler transform. Locating could be achieved, for example, by encoding the $\phi$ function of the $r$-index \cite{GNP20} in $O(\chi)$ space. As we briefly mentioned in the paper, this is possible in space $O(\bar r)$ on top of the random access oracle (while retaining I/O efficiency), but $\bar r$ can be asymptotically larger than $\chi$.

\paragraph*{The repetitiveness measure $\chi$}

The relation $\chi \le \bar r$ was always true in our experiments, even though our theory only predicts $\chi \leq 2\bar r$ (Lemma~\ref{lem:upper_bound}). Can we prove better upper-bounds for $\chi$ as a function of $\bar r$ or, more in general, as a function of other repetitiveness measures \cite{Nav22a}? 

\paragraph*{Necessary samples}

While Suffixient Arrays are the optimal solution to the problem of minimizing the cardinality of a suffixient set (Definition \ref{def:suffixient}), they are not necessarily the smallest sampling of the Prefix Array guaranteeing the correctness of Algorithm \ref{alg:one-occ}. In a sense, Suffixient Arrays guarantee a \emph{sufficient} but not \emph{necessary} condition for Algorithm \ref{alg:one-occ} to work correctly. While we have determined what this necessary condition is, efficient algorithms to optimize it are still under investigation.

\paragraph*{Optimizations}

Our implementation is not yet heavily optimized: by using more cache-efficient data structures (in particular, Elias-Fano dictionaries), we believe we can get even closer to the RAM throughput of our machine. We will explore this direction in a future practical implementation of our index in the context of a usable DNA aligner for repetitive collections. 

\paragraph*{Approximate pattern matching}

We are interested in whether suffixient sets can be adapted for use with the technique of finding all substrings of $T$ within edit distance $k$ of $P$ by cutting $P$ into $k + 1$ pieces and, for each piece, trying to match that piece exactly and then match the other pieces allowing for a total of $k$ edits to them~\cite{Lam2009}. More generally, this technique is referred to as search schemes~\cite{Kucherov2015}. If we can overcome the challenges then we believe suffixient sets are well-suited to this task since, if $P$ is reasonably long and $k$ is reasonably small, after we have exactly matched one piece we will be fairly deep in the suffix tree of $T$, where edge labels tend to be long and it is advantageous to descend edges without checking all the characters in their labels (which we do via $\LCP$ queries between suffixes of $P$ and $T$).

The problem we face is determining all the ways to approximately match the other pieces.  This can be done with backtracking in the suffix tree, but with our current implementation we do not see how to perform this backtracking: when we compute the $\LCP$ of a suffix of $P$ and any suffix of $T$, we may descend several edges at once without realizing it and, for backtracking, later we should go back and partly explore the subtrees hanging off the path we descended.  We are currently looking for ways to detect when we descend past a node in the suffix tree and find its string depth, so we can return to it later and visit its other children.

\section*{Acknowledgements}

We thank Ragnar Groot Koerkamp and Giulio Ermanno Pibiri for fruitful discussions on the topic.

\section*{Funding}
\textit{Davide Cenzato, Sung-Hwan Kim, Nicola Prezza}: Funded by the European Union (ERC, REGINDEX, 101039208). 
Views and opinions expressed are however those of the author(s) only and do not necessarily reflect those of the European Union or the European Research Council Executive Agency. Neither the European Union nor the granting authority can be held responsible for them.
\textit{Lore Depuydt}: Funded by  a PhD Fellowship FR (1117322N), Research Foundation – Flanders (FWO).
\textit{Travis Gagie}: Funded by NSERC Discovery Grant RGPIN-07185-2020. 
\textit{Govanni Manzini}: Funded by the NextGenerationEU programme PNRR ECS00000017 Tuscany Health Ecosystem (CUP: I53C22000780001), by the project PAN-HUB, Italian Ministry of Health (CUP: I53C22001300001), and by the 
the NextGenerationEU programme 
FutureHPC and BigData of the Centro Nazionale Ricerca in High-Performance Computing, Big Data and Quantum Computing.
\textit{Francisco Olivares}: Funded by scholarship ANID-Subdirección de Capital Humano/Doctorado Nacional/2021-21210579, ANID, Chile and by Basal Funds FB0001, ANID, Chile.


\bibliographystyle{unsrt}
\bibliography{main}

\end{document}